\documentclass{siamltex}
\usepackage{amsfonts}
\usepackage{amssymb,amsmath}
\usepackage{graphicx}
\usepackage{color}
\usepackage{hyperref}
\usepackage{subfigure}

\DeclareMathAlphabet{\itbf}{OML}{cmm}{b}{it}
\newcommand{\nc}{\newcommand}
\nc{\bx}{{\bf x}}
\nc{\bu}{{\bf u}}
\nc{\vx}{\vec{\bx}}
\nc{\ve}{\vec{\bf e}}
\nc{\vn}{\vec{\bf n}}
\nc{\bn}{{\bf n}}
\nc{\de}{\delta}
\nc{\la}{\lambda}
\nc{\ep}{\varepsilon}
\nc{\om}{\omega}

\nc{\vH}{\vec {\bf H}}
\nc{\vE}{\vec {\bf E}}
\nc{\vD}{\vec {\bf D}}
\nc{\vJ}{\vec {\boldsymbol {\mathcal J}}}
\nc{\bcJ}{{\boldsymbol {\mathcal J}}}
\nc{\bcF}{{\boldsymbol {\mathcal F}}}
\nc{\bcV}{{\boldsymbol {\mathcal V}}}
\nc{\bH}{{\bf H}}
\nc{\bE}{{\bf E}}
\nc{\bD}{{\bf D}}
\nc{\bU}{{\bf U}}
\nc{\bJ}{{\bf J}}

\nc{\rJ}{{\rm J}}
\nc{\rH}{{\rm H}}
\nc{\rE}{{\rm E}}
\nc{\rD}{{\rm D}}
\nc{\rU}{{\rm U}}

\nc{\cJ}{\mathcal J}
\nc{\cR}{\mathcal{R}}
\nc{\cS}{\mathcal S}
\nc{\mM}{\mathfrak M}
\nc{\bpsi}{\boldsymbol \Psi}
\nc{\bphi}{{\boldsymbol{\varphi}}}

\nc{\EE}{{\mathbb{E}}}
\nc{\MM}{{\mathbb{M}}}

\renewcommand{\hat}{\widehat}

\begin{document}

\title{Electromagnetic wave propagation in random waveguides}

\author{ Ricardo Alonso\footnotemark[1] and 
Liliana Borcea\footnotemark[2]}

\maketitle 

\renewcommand{\thefootnote}{\fnsymbol{footnote}}

\footnotetext[1]{Departamento de Matem\'{a}tica, PUC--Rio \&
  Computational and Applied Mathematics, Rice University, Houston, TX
  77005. {\tt rja2@rice.edu}} \footnotetext[2]{Department of
  Mathematics, University of Michigan, Ann Arbor, MI 48109. {\tt
    borcea@umich.edu}}

\begin{abstract}
  We study long range propagation of electromagnetic waves in random
  waveguides with rectangular cross-section and perfectly conducting
  boundaries. The waveguide is filled with an isotropic linear
  dielectric material, with randomly fluctuating electric
  permittivity.  The fluctuations are weak, but they cause significant
  cumulative scattering over long distances of propagation of the
  waves.  We decompose the wave field in propagating and evanescent
  transverse electric and magnetic modes with random amplitudes that
  encode the cumulative scattering effects. They satisfy a coupled
  system of stochastic differential equations driven by the random
  fluctuations of the electric permittivity.  We analyze the solution
  of this system with the diffusion approximation theorem, under the
  assumption that the fluctuations decorrelate rapidly in the range
  direction. The result is a detailed characterization of the
  transport of energy in the waveguide, the loss of coherence of the
  modes and the depolarization of the waves due to cumulative
  scattering.
\end{abstract}
\renewcommand{\thefootnote}{\arabic{footnote}}
\begin{keywords}
Waveguides, electromagnetic, random media, asymptotic analysis.
\end{keywords}

\begin{AMS}
35Q61, 35R60
\end{AMS}

\section{Introduction}
\label{sect:intro}
We study electromagnetic wave propagation in waveguides.  There is
extensive applied literature on this subject
\cite{marcuse1974theory,marcuse1982light,LDT-03,ISSY-93,collin1991field,
  mrozowski1997guided} which includes open and closed waveguides,
waveguides with losses, boundary corrugation and heterogeneous
media. Here we consider the setup illustrated in Figure \ref{fig:setup},
for a waveguide with rectangular cross-section $\Omega =
(0,L_1)\times(0,L_2)$, filled with an isotropic linear dielectric
material. The waves are trapped by perfectly conducting boundaries and
propagate in the range direction $z$. The cross-range coordinates are
${\bf x} = (x_1,x_2) \in \Omega$.  The main goal of the paper is to
analyze long range wave propagation in waveguides with
imperfections. We refer to \cite{kohler77,Dozier,GS-08,ABG-12,BG-13}
and \cite[Chapter 20]{LAY_BOOK} for rigorous mathematical studies of
long range wave propagation in imperfect acoustic waveguides, and to
\cite{BIT-10,QUANT-13,G-1} for their application to imaging and time
reversal. Here we extend the theory to electromagnetic waves.

We focus attention on waveguides with imperfections due to a
heterogeneous dielectric, but the ideas should extend to waveguides
with corrugated boundaries. Such waveguides can be analyzed by
changing coordinates to flatten the boundary fluctuations as was done
in \cite{ABG-12} for sound waves, or by using so-called local normal
mode decompositions as proposed in \cite[chapter
  9]{marcuse1982light}.  Our waveguide has straight walls and is
filled with a dielectric material that has numerous inhomogeneities
(imperfections). These are weak scatterers, so their effect is
negligible in the vicinity of the source of the waves.  However, the
inhomogeneities cause significant cumulative wave scattering over long
ranges. To quantify the cumulative scattering effects we study the
following questions: How are the modal wave components coupled by
scattering? How do the waves depolarize? How do the waves lose
coherence?  Can we calculate from first principles the
\emph{scattering mean free paths}, which are the range scales over
which the modal wave components lose coherence?  How is energy
transported at long ranges in the waveguides?  Can we quantify the
\emph{equipartition distance} where cumulative scattering is so strong
that the waves lose all information about the source? How does the
equipartition distance compare with the mode dependent scattering mean
free paths?

To answer these questions we model the scalar valued electric
permittivity $\ep$ of the dielectric as a random process. The random
model is motivated by the fact that in applications the imperfections
can never be known in detail. They are the uncertain microscale of the
medium, the fluctuations of $\ep(\vx)$ in $\vx = (\bx,z)$, so we model
them as random.  The fluctuations are small, on a scale (correlation
length) comparable to the wavelength. We assume that there is no
dissipation in the medium, meaning that $\ep(\vx)$ is real, positive.
Complex valued permittivities $\ep(\om,\vx)$ which are typically
required by causality i.e., Kramers-Kronig relations, can be
incorporated in the model. We do not consider them here for
simplicity, and because we are concentrating on the analysis at a
single frequency $\om$. Extensions to multi frequency analysis of wave
propagation in dispersive and lossy media can be done, using
techniques like in \cite{Dozier,kohler77,ABG-12} and \cite[chapter
  20]{LAY_BOOK}, but we leave them for a different publication.

\begin{figure}[t!]
\vspace{-0.4in}
   \begin{center}
    \hspace{-1in} \input{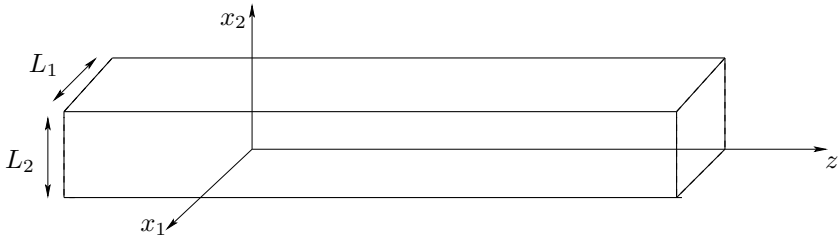}
   \end{center}
   \vspace{-0.1in}
   \caption{Schematic of the setup. The waveguide is unbounded in the
     range direction $z$ and has rectangular cross-section in the
     plane $(x_1,x_2)$, with sides $L_1$ and $L_2$. }
   \label{fig:setup}
\end{figure}

The paper is organized as follows: We begin in section \ref{sect:form}
with the setup. We state Maxwell's equations and the boundary
conditions satisfied by the electromagnetic field. Then we follow the
approach in \cite{marcuse1974theory} and solve for the components
$E_z(\om,\vx)$ and $H_z(\om,\vx)$ of the electric and magnetic fields
in the range direction.  We obtain a $4 \times 4$ system of partial
differential equations for the components $\bE(\om,\vx)$ and
$\bH(\om,\vx)$ of the fields in the cross-range plane. We analyze in
section \ref{sect:ideal} its solution $\bE_o(\om,\vx)$ and
$\bH_o(\om,\vx)$ in ideal waveguides with constant permitivity
$\ep_o$. It is a superposition of uncoupled transverse electric and
transverse magnetic modes. The random model of the waveguide is
introduced in section \ref{sect:asympt}.  Because the amplitude of the
fluctuations of $\ep(\vx)$ is small, of order $\epsilon \ll 1$, the
system of equations for $\bE(\om,\vx)$ and $\bH(\om,\vx)$ is a
perturbation of that in ideal waveguides.  The remainder of the paper
is concerned with the asymptotic analysis of $\bE(\om,\vx)$ and
$\bH(\om,\vx)$ at long ranges, in the limit $\epsilon \to 0$. We
consider long ranges because the $\epsilon \to 0$ limit of
$\bE(\om,\vx)$ and $\bH(\om,\vx)$ is the same as the ideal waveguide
solution $\bE_o(\om,\vx)$ and $\bH_o(\om,\vx)$ when the waves do not
propagate far from the source.  Our analysis is based on the
decomposition of $\bE(\om,\vx)$ and $\bH(\om,\vx)$ in transverse
electric and magnetic modes, with random amplitudes that encode the
cumulative scattering effects, as explained in section
\ref{sect:modec}. The long range scaling and the diffusion limit
approximation for analyzing the wave field as $\epsilon \to 0$ are
stated in section \ref{sect:diff}.  The main results of the paper are
in section \ref{sect:transp}, where we characterize the limit process.
Explicitly, we describe the loss of coherence and depolarization of
the waves due to cumulative scattering, and the transport of
energy. We also show that as we let the range grow, the waves scatter
so much that they eventually reach the equipartition regime, where
they lose all information about the source.  We end with a summary in
section \ref{sect:summary}.

\section{Setup}
\label{sect:form}
Let $\ve_1$, $\ve_2$ and $\ve_z$ be the unit vectors along the
coordinate axes, and use bold letters with an arrow on top for three
dimensional vectors, and bold letters for two dimensional vectors in
the cross-range plane.  Exlicitly, we write
\begin{equation}
\label{eq:f2}
\vH = \rH_1 \ve_1 + \rH_2 \ve_2 + \rH_z \ve_z\, , \qquad 
\bH = (\rH_1,\rH_2)\, ,
\end{equation}
for the magnetic field $\vH(\om,\vx)$, and similarly for the electric
field $\vE(\om,\vx)$ and electric displacement $\vD(\om,\vx)$.  They
satisfy Maxwell's equations
\begin{align}
  \vec{\nabla} \times \vH(\om,\vx) &= \vJ(\om,\vx) - i \om
  \vD(\om,\vx)\, ,
\label{eq:M1} \\
\vec{\nabla} \times \vE(\om,\vx) &= i \om \mu_o \vH(\om,\vx)\, ,
\label{eq:M2} \\
\vec{\nabla} \cdot \vH(\om,\vx) &= 0\, ,
\label{eq:M3} \\
\vec{\nabla} \cdot \vD(\om,\vx) &= \rho(\om,\vx)\, ,
\label{eq:M4} 
\end{align}
where $\vJ$ and $\rho$ are the current source density and free charge
density, and $\mu_o$ is the magnetic permeability, assumed constant.
We denote by
\[
\vec{\nabla} = \partial_{x_1} \ve_1 + \partial_{x_2} \ve_2 + 
\partial_z \ve_z \,
\]
the three dimensional gradient and by 
$\vec{\nabla} \times$ and $\vec{\nabla} \cdot$ the curl and
divergence operators. 

The current source density 
\begin{equation}
  \vJ(\om,\vx) = \left(\bcJ(\om,\vx),\cJ_z(\om,\vx)\right) = 
  \left(\bJ(\om,\bx),\rJ_z(\om,\bx)\right) \delta(z)\, ,
\label{eq:f4}
\end{equation}
models a source at the origin of range, supported in the interior of
$\Omega$.  The Fourier transform $\rho(\om,\vx)$ of the free charge
density can be obtained from the continuity of charge derived from
(\ref{eq:M1}) and (\ref{eq:M4})
\begin{equation}
  -i \om \rho(\om,\vx) + \vec{\nabla} \cdot \vJ(\om,\vx) = 0\, .
\label{eq:f5}
\end{equation}
It vanishes at ranges $z \ne 0$.

The electric displacement is proportional to the electric field
\begin{equation}
\label{eq:f3}
\vD(\om,\vx) = \ep(\vx) \vE(\om,\vx)\,,
\end{equation}
with scalar valued, positive and bounded electric permittivity $\ep$.
The analysis is for a single frequency, so we simplify the notation by
omitting henceforth $\om$ from the arguments of the fields.

\subsection{The $4 \times 4$ system of equations}
\label{sect:system}
We study the evolution of the two dimensional vectors $\bE(\vx)$ and
$\bH(\vx)$ for $z > 0$. They determine the components $E_z(\vx)$ and
$H_z(\vx)$ in the range direction of the electric and magnetic fields,
as follows from equations (\ref{eq:M1})-(\ref{eq:M2})
\begin{align}
\rH_z(\vx) &= -\frac{i}{\om \mu_o} \nabla^\perp \cdot \bE(\vx) \, ,
\label{eq:f6}\\
\rE_z(\vx) &= \frac{i}{\om \ep(\vx)} \left[ \nabla^\perp \cdot \bH(\vx)
  - \cJ_z(\vx) \right]\, ,
\label{eq:f7} 
\end{align}
with $\nabla^\perp = (-\partial_{x_2},\partial_{x_1})$ the
perpendicular gradient in the cross-range plane. The $4\times 4 $
system of equations for $\bE(\vx)$ and $\bH(\vx)$ is
\begin{align}
  \partial_z \bE(\vx) &= \frac{i}{\om} \nabla \left[\frac{1}{\ep(\vx)}
    \nabla^\perp \cdot \bH(\vx)\right] - \frac{i}{\om} \nabla
  \left[\frac{\cJ_z(\vx)}{\ep(\vx)}\right] - i \om \mu_o
  \bH^\perp(\vx)\, ,
\label{eq:f8} \\
\partial_z \bH(\vx) &= -\frac{i}{\om \mu_o} \nabla \left[ \nabla^\perp
  \cdot \bE(\vx)\right] + i \om \ep(\vx) \bE^\perp(\vx) -
\bcJ^\perp(\vx)\, .
\label{eq:f9}
\end{align}
Here $\nabla = (\partial_{x_1},\partial_{x_2})$ is the gradient in the
cross-range plane, and we let ${\bf a}^\perp = (-a_2,a_1)$ denote the
rotation of any vector ${\bf a} = (a_1,a_2)$ by $90$ degrees,
counter-clockwise. 

Note that equations (\ref{eq:f6})-(\ref{eq:f9}) contain all the
information in the Maxwell system (\ref{eq:M1})-(\ref{eq:M4}). Indeed,
(\ref{eq:M3}) follows from (\ref{eq:f6}) and (\ref{eq:f8})
\begin{align*}
  \vec{\nabla} \cdot \vH(\vx) &= \nabla \cdot \bH(\vx) + \partial_z
  \rH_z(\vx) \\
  &= \nabla \cdot \bH(\vx) - \frac{i}{\om \mu_o} \nabla^\perp \cdot
  \partial_z
  \bE(\vx) \\
  &= \nabla \cdot \bH(\vx) - \nabla^\perp \cdot \bH^\perp(\vx)
  \\
  &=0\, ,
\end{align*}
because $\nabla^\perp \cdot \nabla a(\vx) = 0$ for any twice
continuously differentiable function $a(\vx)$. Similarly,
(\ref{eq:M4}) follows from (\ref{eq:f7}) and
(\ref{eq:f9})
\begin{align*}
  \vec{\nabla} \cdot \vD(\vx) &= \nabla \cdot \bD(\vx) + \partial_z
  \rD_z(\vx) \\
  &=\nabla \cdot\left[\ep(\vx) \bE(\vx)\right] + \frac{i}{\om} \left[
    \nabla^\perp \cdot \partial_z \bH(\vx) - \partial_z
    \cJ_z(\vx)\right] \\
  &= \nabla \cdot\left[\ep(\vx) \bE(\vx)\right] - \nabla^\perp
  \cdot\left[\ep(\vx) \bE^\perp(\vx)\right] - \frac{i}{\om} 
\left[\nabla^\perp \cdot \bcJ^\perp(\vx) + \partial_z 
\cJ_z(\vx)\right] \\
&= - \frac{i}{\om} \vec{\nabla} \cdot \vJ(\vx) \\&= \rho(\vx)\,,
\end{align*}
where we used (\ref{eq:f3}) and the continuity of charge relation
(\ref{eq:f5}).

\subsection{Boundary conditions}
\label{sect:formBC}
The boundary conditions at the perfectly conducting boundary $\partial
\Omega$ are \cite[Chapter 8]{jackson1999classical}
\begin{equation}
\label{eq:f13}
  \vn(\bx) \times \vE(\vx) = 0\, 
\end{equation}
for $\vx = (\bx,z)$ and $\bx \in \partial \Omega$. The outer normal
$\vn(\bx) = (\bn(\bx),0)$ at $\partial \Omega$ is independent of the
range and is orthogonal to $\ve_z$. Thus, equations (\ref{eq:f13}) say
that the tangential components of the electric field vanish at the
boundary. Explicitly,
\begin{align} \rE_z(\vx) &=0\,, \qquad \bn^\perp(\bx) \cdot
  \bE(\vx) =0\,. \label{eq:f10}
\end{align}
We need more boundary conditions at $\partial \Omega$ to specify
uniquely the solution of (\ref{eq:f8}-\ref{eq:f9}), but they can be
derived from Maxwell's equations (\ref{eq:M1}-\ref{eq:M2}), conditions
(\ref{eq:f10}), and our assumptions on the source density
(\ref{eq:f4}), as explained in section \ref{sect:ideal}.



The fields are bounded and outgoing at $|z| \to \pm \infty$. We
explain in section \ref{sect:formMod} that the causality of the
problem in the time domain allows us to restrict the fluctuations of
$\ep(\vx)$ to a finite range interval, and thus justify the outgoing
boundary conditions.

\subsection{Conservation of energy}
\label{sect:energy}
The fields $\bE(\vx)$ and $\bH(\vx)$ satisfy an energy conservation
relation, stated in the following proposition, and used in the
analysis in section \ref{sect:transp}. 
\begin{proposition}
\label{prop:EC}
For any $z > 0$, we have the conservation relation
\begin{equation}
  \cS(z) = -\int_\Omega d \bx \, \operatorname{Re} \left[ \bE(\vx) \cdot
    \overline{\bH^\perp(\vx)} \right] = \cS(0+).
\label{eq:ENC}
\end{equation}
where the bar denotes complex conjugate.
\end{proposition}

\vspace{0.1in} \noindent Note that 
\[
\vec{\bf S}(\vx) = \frac{1}{2}\operatorname{Re} \left[ \vE(\vx) \times
  \overline{\vH(\vx)} \right]
\]
is the time average of the Poynting vector of a time harmonic wave
\cite[chapter 7]{jackson1999classical}. Therefore,
\[
\cS(z) = 2\int_\Omega d\bx \, \ve_z \cdot \vec{\bf S}(\vx)\, 
\]
is twice the flux of energy in the range direction, and (\ref{eq:ENC})
states that it is conserved for all $z > 0$.

To derive (\ref{eq:ENC}) we obtain from (\ref{eq:M1})-(\ref{eq:M2})
that
\begin{align*}
  \vec{\nabla} \cdot \left[ \vE(\vx) \times \overline{\vH(\vx)}
  \right] &= \overline{\vH(\vx)} \cdot \left[\vec{\nabla} \times
    \vE(\vx)\right] - \vE(\vx) \cdot \left[ \vec{\nabla} \times
    \overline{\vH(\vx)}\right] \nonumber \\
  &= i \om \mu_o \left| \vH(\vx)\right|^2 - i \om \ep(\vx) \left|
    \vE(\vx)\right|^2 - \vE(\vx) \cdot \overline{\vJ(\vx)}\, ,
\end{align*}
and from the divergence theorem that
\begin{align*}
  \int_\Omega d \bx \vec{\nabla} \cdot \left[ \vE(\vx) \times
    \overline{\vH(\vx)} \right] = &\int_{\partial \Omega} d s(\bx) \,
  \vn(\bx) \cdot \left\{(I-\ve_z \ve_z^T)\left[ \vE(\vx) \times
      \overline{\vH(\vx)} \right]\right\} + \\
&\int_\Omega d \bx \, 
  \partial_z \left\{ \ve_z \cdot \left[ \vE(\vx) \times
      \overline{\vH(\vx)} \right]\right\}\, .
\end{align*}
The boundary term vanishes because of the boundary conditions (\ref{eq:f10})
\[
\vn(\bx) \cdot (I-\ve_z \ve_z^T)\left[ \vE(\vx) \times
  \overline{\vH(\vx)} \right] = E_z(\vx)\, \bn(\bx) \cdot \overline{
  \bH^\perp(\vx)} + \overline{H_z(\vx)}\, \bn^\perp(\bx) \cdot \bE(\vx) = 0,
\]
and the integrand in the second term satisfies
\[
\ve_z \cdot \left[ \vE(\vx) \times \overline{\vH(\vx)} \right] =-
\bE(\vx) \cdot \overline{ \bH^\perp(\vx)}\, .
\]
The current source density $\overline{\vJ(\vx)}$ is supported at $z=
0$, so we conclude that
\begin{align*}
  -\partial_z \int_\Omega d\bx \, \bE(\vx) \cdot
  \overline{\bH^\perp(\vx)} =& \int_\Omega d \bx \left[
i \om \mu_o  \left| \vH(\vx)\right|^2 - i \om \ep(\vx) 
\left| \vE(\vx)\right|^2 \right]\, , \quad z \ne 0\, .  
\label{eq:id}
\end{align*}
The conservation relation (\ref{eq:ENC}) follows by taking the real
part in this equation.

\section{Ideal waveguides}
\label{sect:ideal}
Maxwell's equations are separable in ideal waveguides with constant
permitivity $\ep_o$, and it is typical to solve for the longitudinal
components $E_z(\vx)$ and $H_z(\vx)$ of the electric and magnetic
fields, which then define $\bE(\vx)$ and $\bH(\vx)$
\cite[chapter8]{jackson1999classical}.  The solution is given by a
superposition of waves, called modes.  They are propagating and
evanescent waves and solve Maxwell's equations with boundary
conditions (\ref{eq:f10}). We describe the modes in section
\ref{sect:ModeDec}, and then write the solution in section
\ref{sect:amplit}.

\subsection{The waveguide modes}
\label{sect:ModeDec}
The longitudinal components of the electric and magnetic fields
satisfy the boundary conditions
\begin{equation}
E_z(\vx) = \bn(\bx) \cdot \nabla H_z(\vx) = 0, \qquad \bx \in \partial
\Omega.
\label{eq:BCHz}
\end{equation}
The first condition is just (\ref{eq:f10}), and the second follows
from Maxwell's equations (\ref{eq:M1}-\ref{eq:M2}).  Indeed, 
(\ref{eq:M2}) gives
\begin{equation}
\bH(\vx) = \frac{i}{\om \mu_o} \left[ \nabla^\perp E_z(\vx) -
  \partial_z \bE^\perp(\vx)\right],
\label{eq:bH}
\end{equation}
so the normal component of $\bH$ at $\partial \Omega$ satisfies
\begin{align}
\bn(\bx) \cdot \bH(\vx) &= \frac{i}{\om \mu_o} \left[-\bn^\perp(\bx) \cdot
  \nabla E_z(\vx) + \partial_z \bn^\perp(\bx) \cdot \bE(\vx)\right] =
0, \qquad \bx \in \partial \Omega.
\label{eq:f11}
\end{align}
Similarly, we obtain from equation (\ref{eq:M1}) that 
\begin{equation}
\bD(\vx) = \frac{i}{\om} \left[- \nabla^\perp H_z(\vx) +
  \partial_z \bH^\perp(\vx)\right], 
\label{eq:bD}
\end{equation}
and the boundary condition (\ref{eq:f10}) implies that
\begin{align*}
\bn^\perp(\bx) \cdot \bD(\vx) = \frac{i}{\om} \left[ -\bn(\bx) \cdot
  \nabla H_z(\vx) + \partial_z \bn(\bx) \cdot \bH(\vx) \right] = 0.
\end{align*}
The Neumann boundary condition (\ref{eq:BCHz}) on $H_z$ follows from
this equation and (\ref{eq:f11}).

The waveguide modes are solutions of Maxwell's equations that depend
on the range $z$ as $\exp(\pm i \beta z)$, with mode wavenumber
$\beta$ to be defined.  We write them as
\begin{equation}
\widetilde{\bD} (\bx;\pm \beta) e^{\pm i \beta z}, \qquad \widetilde
\bH(\bx;\pm \beta) e^{\pm i \beta z},
\end{equation}
and similar for the longitudinal components, which satisfy
\begin{align}
\Delta \widetilde E_z(\bx;\pm \beta) + (k^2-\beta^2) \widetilde
E_z(\bx;\pm \beta) &= 0, \label{eq:EzE}\\ \Delta \widetilde
H_z(\bx;\pm \beta) + (k^2-\beta^2) \widetilde H_z(\bx;\pm \beta) &= 0,
\qquad \bx \in \Omega. \label{eq:HzE}
\end{align}
Here $\Delta$ is the Laplacian in $\bx$, $k = \om/c_o$ is the
wavenumber and $c_o = 1/\sqrt{\ep_o \nu_o}$ is the wave speed.

\subsubsection{Spectral decomposition of the Laplacian}
The Laplacian operator acting on functions with homogeneous Dirichlet
conditions is symmetric negative definite, with countable eigenvalues
\begin{equation}
\lambda_{j} =\left( \frac{ \pi j_1}{L_1}\right)^2 + \left( \frac{\pi
  j_2}{L_2}\right)^2 \,,
\label{eq:P10}
\end{equation}
and eigenfunctions
\begin{equation}
\widetilde E_{j,z}(\bx) = \sin\left(\frac{\pi j_1 x_1}{L_1}\right)
\sin \left(\frac{\pi j_2 x_2}{L_2}\right).
\label{eq:P10Ez}
\end{equation}
The indexes $j_1$ and $j_2$ are natural numbers satisfying the
constraint $j_1^2 + j_2^2 \ne 0$. We associate the pair $(j_1,j_2)$ to
the index $j$ because $\mathbb{N} \times \mathbb{N}$ is countable,
and enumerate the eigenvalues in increasing order.

Similarly, the Laplacian operator acting on functions with homogeneous
Neumann conditions is symmetric negative semidefinite, with the same
eigenvalues as (\ref{eq:P10}), and eigenfunctions
\begin{equation}
\widetilde H_{j,z}(\bx) = \cos\left(\frac{\pi j_1 x_1}{L_1}\right) \cos
\left(\frac{\pi j_2 x_2}{L_2}\right).
\label{eq:P10Hz}
\end{equation}

Thus, we see that the electric and magnetic fields have the same mode
wavenumbers $\beta$, which take the discrete values
$\sqrt{k^2-\lambda_j}$. We write them as
\begin{equation}
\beta = \left\{ \begin{array}{ll}\beta_j, \quad &j = 1, \ldots, N, \\
 i \beta_j, & j >N, 
\end{array} \right. \qquad {\rm for} ~ ~ \beta_j = \sqrt{|k^2 - \lambda_j|},
\label{eq:beta}
\end{equation}
to emphasize that only the first $N$ are real.  The infinitely many
modes that correspond to eigenvalues $\lambda_j > k^2$ are
evanescent. We assume that $\beta_N \ne 0$, so there are no standing
waves in the waveguide.

\subsubsection{The transverse electric and magnetic modes}
\label{eq:TMTE}
It follows immediately from (\ref{eq:bH}), (\ref{eq:bD}),
(\ref{eq:P10Ez}) and (\ref{eq:P10Hz}) that $\widetilde \bD$ and
$(\widetilde \bH)^\perp$ are given by superpositions of the vectors
$\nabla^\perp \widetilde H_{j,z}(\bx)$ and $\nabla \widetilde
E_{j,z}(\bx)$. Thus, we define the vectors
\begin{equation}
\bphi_j^{(1)} = \alpha_j \nabla^\perp \widetilde H_{j,z}(\bx) = \alpha_j \left(
\begin{matrix}
\frac{\pi j_2}{L_2} \cos\left(\frac{\pi j_1 x_1}{L_1}\right)
\sin \left(\frac{\pi j_2 x_2}{L_2}\right) \\\\
-\frac{\pi j_1}{L_1} \sin\left(\frac{\pi j_1 x_1}{L_1}\right)
\cos \left(\frac{\pi j_2 x_2}{L_2}\right)
\end{matrix} \right)\, ,
\label{eq:TE}
\end{equation}
and 
\begin{equation}
\bphi_j^{(2)} = \alpha_j \nabla \widetilde E_{j,z}(\bx) = \alpha_j \left(
\begin{matrix}
\frac{\pi j_1}{L_1} \cos\left(\frac{\pi j_1 x_1}{L_1}\right)
\sin \left(\frac{\pi j_2 x_2}{L_2}\right) \\\\
\frac{\pi j_2}{L_2} \sin\left(\frac{\pi j_1 x_1}{L_1}\right)
\cos \left(\frac{\pi j_2 x_2}{L_2}\right)
\end{matrix} \right)\, ,
\label{eq:TM}
\end{equation}
normalized by 
\begin{equation}
\alpha_j = \left\{ \begin{array}{ll} \frac{2}{\sqrt{\lambda_j L_1
      L_2}}, \qquad &j = (j_1, j_2), \quad j_1 j_2 \ne 0,
\\  \\ \sqrt{\frac{2}{\lambda_j L_1 L_2}}, & j = (j_1,j_2), \quad j_1
  j_2 = 0.
\end{array}
\right.
\label{eq:alpha}
\end{equation}
so that
\[
\|\bphi^{(s)}_j\|^2 = \int_{\Omega} d \bx \left| \bphi_j^{(s)}(\bx)
\right|^2 = 1, \qquad s = 1,2 .
\]

The vectors indexed by $s = 1$ correspond to \emph{transverse
  electric} (TE) modes. Indeed, they satisfy
\begin{equation}
\nabla \cdot \bphi^{(1)}(\bx) = 0, \qquad \bx \in \Omega,
\label{eq:P12}
\end{equation}
so when we set $\bH^\perp(\vx) = \bphi^{(1)}(\bx)e^{i \beta_j z}$ in
(\ref{eq:f7}) we get $E_{z}(\vx) = 0$.  Similarly, the vectors indexed
by $s = 2$ correspond to \emph{transverse magnetic} (TM) modes.  They
satisfy
\begin{equation}
\nabla^\perp \cdot \bphi^{(2)}(\bx) = 0, \qquad \bx \in \Omega,
\label{eq:P13}
\end{equation}
and give $H_{z}(\vx) = 0$ by equation (\ref{eq:f6}).

The superposition of $\bphi_j^{(1)}(\bx)$ and $\bphi_j^{(2)}(\bx)$ in
the definition of the fields $\widetilde \bE$ and $(\widetilde
\bH)^\perp$ is their Helmholtz decomposition in a divergence free part
and a curl free part.

\subsubsection{Analogous derivation of the waveguide modes}
\label{sect:analog}
We could have arrived at the same wave decomposition if we worked
directly with the transverse components $\bD$ and $\bH$ of the fields.
This observation is relevant because when the permittivity varies in
$\vx$, as in the random waveguide, it is no longer possible to solve
independently for the longitudinal wave fields $E_z$ and $D_z$.

We let 
\begin{equation}
\bH(\vx) = c_o \bU^\perp(\vx),
\label{eq:defU}
\end{equation}
where $\bU$ is the rotated magnetic field scaled by $1/c_o$. It is convenient
to work in the $\bD$ and $\bU$ variables because as we see below, they
satisfy the same boundary conditions and have the same physical units.
Note from (\ref{eq:TE}) and (\ref{eq:TM}) that $\bphi^{(s)}(\bx)$ are
eigenfunctions of the vector Laplacian
\begin{equation}
\Delta \bphi_j^{(s)}(\bx) = \nabla \left[ \nabla \cdot
  \bphi_j^{(s)}(\bx)\right] + \nabla^\perp \left[ \nabla^\perp \cdot
  \bphi_j^{(s)}(\bx)\right] = -\lambda_j \bphi_j^{(s)}(\bx), 
\end{equation}
for $\bx \in \Omega$, with boundary conditions
\begin{equation}
\bn^\perp(\bx) \cdot \bphi_j^{(s)}(\bx) = 0, \qquad \nabla \cdot
\bphi_j^{(s)}(\bx) = 0, \qquad \bx \in \partial \Omega.
\label{eq:P9}
\end{equation}
The index $s = 1, 2$ corresponds to the multiplicity $\mathfrak{M}_j$
of the eigenvalues. We can limit the multiplicity of $\lambda_j$ by
assuming that the waveguide dimensions satisfy $L_1/L_2 \ne
\mathbb{Q}$. This implies that
\begin{equation}
\lambda_j \ne \lambda_{j'}\,, \qquad {\rm if} ~ ~ j = (j_1,j_2) \ne j'
= (j_1',j_2')\, .
\label{eq:P11}
\end{equation}
When $j = (j_1, j_2)$ and either $j_1$ or $j_2$ are zero,
$\mathfrak{M}_j = 1$, and only the TE modes $\bphi_j^{(1)}(\bx)$
exist. Otherwise $\mathfrak{M}_j = 2$.

The eigenfunctions satisfy the orthogonality relations
\begin{equation}
  \left< \bphi_j^{(s)}\,, \bphi_{j'}^{(s')}\right> = 
  \int_\Omega d \bx\, \bphi_{j}^{(s)}(\bx) \cdot 
  \bphi_{j'}^{(s')}(\bx) = \delta_{jj'}\delta_{ss'}\, .
\label{eq:orthog}
\end{equation}
and $\left\{ \bphi_j^{(s)} \right\}_{1 \le s \le \mM_j, j \ge 1}$ is a
complete set that can be used to describe an arbitrary electromagnetic
wave field in the waveguide \cite[chapter8]{jackson1999classical}.

The boundary conditions (\ref{eq:P9}) are consistent with the
conditions satisfied by $\bD(\vx)$ and $\bU(\vx)$, derived from
Maxwell's equations. Indeed, equations (\ref{eq:f7}), (\ref{eq:f10})
and the assumption (\ref{eq:f4}) on the source density give that
\begin{equation}
\nabla \cdot \bU(\vx) = - \frac{1}{c_o} \nabla^\perp \cdot \bH(\vx) = 0, \qquad 
\bx \in \partial \Omega.
\label{eq:CU1}
\end{equation}
Moreover, equation (\ref{eq:f11}) says that 
\begin{equation}
\bn^\perp(\bx) \cdot \bU(\vx) = -\frac{1}{c_o} \bn(\bx) \cdot \bH(\vx) = 0, 
\qquad \bx \in \partial \Omega.
\label{eq:CU2}
\end{equation}
For the electric displacement we already know from (\ref{eq:f10}) that
\begin{equation}
\bn^\perp(\bx) \cdot \bD(\vx) = 0, \qquad \bx \in \partial \Omega.
\label{eq:CD1}
\end{equation}
The divergence condition follows from (\ref{eq:bD}) and (\ref{eq:CU1})
\begin{equation}
\nabla \cdot \bD(\vx) = 0, \qquad \bx \in \partial \Omega,
\label{eq:CD2}
\end{equation}
and since $\partial_z D_z = 0$, it is consistent with the conservation
of charge.


\subsection{The solution in ideal waveguides}
\label{sect:amplit}
We expand $\bD(\vx)$ and $\bU(\vx)$ in the basis $\left\{
\bphi_j^{(s)} \right\}_{1 \le s \le \mM_j, j \ge 1}$ and associate to
each $\bphi_j^{(s)}(x)$ a mode, which is a propagating or evanescent
wave. We rename the fields $\bD_o(\vx)$ and $\bU_o(\vx)$ to remind us
that we are in the ideal waveguide.

Using the identities
\begin{align}
  k^2 \bphi_j^{(s)}(\bx) + \nabla \left[ \nabla \cdot
    \bphi_j^{(s)}(\vx) \right] &= \left[ k^2 \delta_{s1} +
    (k^2-\lambda_j) \delta_{s2}\right]\bphi_j^{(s)}(\bx)\, , \label{eq:UNPT1} \\
  k^2 \bphi_j^{(s)}(\bx) + \nabla^\perp \left[ \nabla^\perp \cdot
    \bphi_j^{(s)}(\vx) \right] &= \left[(k^2-\lambda_j) \delta_{s1} +
    k^2 \delta_{s2}\right] \bphi_j^{(s)}(\bx)\, ,\label{eq:UNPT2}
\end{align}
we obtain that 
\begin{align}
  \hspace{-0.2in} \bD_o(\vx) =& \sum_{j=1}^N \sum_{s=1}^{\mM_j}
  \bphi_{j}^{(s)}(\bx)\left(\sqrt{\frac{k}{\beta_j}} \de_{s1} +
    \sqrt{\frac{\beta_j}{k}} \de_{s2} \right) \left( A_{j,o}^{\pm (s)}
    e^{i \beta_j
      z} + B_{j,o}^{\pm (s)} e^{-i \beta_j z} \right) + \nonumber \\
  &\sum_{j>N} \sum_{s=1}^{\mM_j}
  \bphi_{j}^{(s)}(\bx)\left(\sqrt{\frac{k}{\beta_j}} \de_{s1} +
    \sqrt{\frac{\beta_j}{k}} \de_{s2} \right) \mathfrak{E}_{j,o}^{\pm (s)}
  e^{-\beta_j |z|} \, , \qquad
\label{eq:ID1}
\end{align}
and 
\begin{align}
  \bU_o(\vx) =& \sum_{j=1}^N
  \sum_{s=1}^{\mM_j}\bphi_{j}^{(s)}(\bx)
  \left(\sqrt{\frac{\beta_j}{k}} \de_{s1} + \sqrt{\frac{k}{\beta_j}}
    \de_{s2} \right) \left( A_{j,o}^{\pm(s)} e^{i \beta_j
      z} - B_{j,o}^{\pm (s)} e^{-i \beta_j z} \right) \pm \nonumber \\
  & i \sum_{j>N} \sum_{s=1}^{\mM_j} \bphi_j^{(s)}(\bx)
  \left(\sqrt{\frac{\beta_j}{k}} \de_{s1} - \sqrt{\frac{k}{\beta_j}}
    \de_{s2}\right) \mathfrak{E}_{j,o}^{\pm (s)} e^{-\beta_j |z|} \, ,
\label{eq:ID2}
\end{align}
for $z \ne 0$.  The normalization coefficients $\sqrt{k/\beta_j}$ and
$\sqrt{\beta_j/k}$ are not important here, and could be absorbed in
the mode amplitudes. We use them for consistency with the mode
expansions for the random waveguide in section \ref{sect:modec}.  There
the normalization symmetrizes the system of equations satisfied by the
mode amplitudes.

The amplitudes in (\ref{eq:ID1}-\ref{eq:ID2}) are constant on
each side of the source, and are determined by the source density and
the outgoing boundary conditions.  There are no backward going modes
to the right of the source, at positive ranges, so we can set
$B_{j,o}^{+(s)} = 0$. Similarly, we let $A_{j,o}^{-(s)} = 0$.  The
remaining amplitudes are obtained from the source conditions
\begin{align*} 
  \bD_o(\bx,0+) - \bD_o(\bx,0-) &= -\frac{i}{c_o k} \nabla \rJ_z(\bx)\, ,
\\
\bU_o(\bx,0+) - \bU_o(\bx,0-) &= -\frac{1}{c_o} \bJ(\bx)\, .
\end{align*}
Substituting (\ref{eq:ID1}-\ref{eq:ID2}) in these conditions and using
the orthogonality relations (\ref{eq:orthog}), we get 
\begin{align}
  A_{j,o}^{+(s)} = &- \frac{1}{2
    c_o}\left(\sqrt{\frac{k}{\beta_j}}\de_{s1} +
    \sqrt{\frac{\beta_j}{k}} \de_{s2}\right) \left< \bphi_j^{(s)}\, , 
\bJ \right>- \nonumber \\
  &\frac{i}{2 c_o k} \left( \sqrt{\frac{\beta_j}{k}} \de_{s1} +
    \sqrt{\frac{k}{\beta_j}}\de_{s2}\right)\left< \nabla
    \rJ_z,\bphi_j^{(s)} \right>\, ,
\label{eq:iniA}
\end{align}
and 
\begin{align}
  B_{j,o}^{-(s)} = &- \frac{1}{2
    c_o}\left(\sqrt{\frac{k}{\beta_j}}\de_{s1} +
    \sqrt{\frac{\beta_j}{k}} \de_{s2}\right)
  \left< \bphi_j^{(s)}\, , \bJ \right>+ \nonumber \\
  &\frac{i}{2 c_o k} \left( \sqrt{\frac{\beta_j}{k}} \de_{s1} +
    \sqrt{\frac{k}{\beta_j}}\de_{s2}\right)\left< \nabla
    \rJ_z,\bphi_j^{(s)} \right>\, ,
\end{align}
for the propagating modes and 
\begin{align}
  \mathfrak{E}_{j,o}^{\pm(s)} = & \frac{i}{2
    c_o}\left(\sqrt{\frac{k}{\beta_j}}\de_{s1} -
    \sqrt{\frac{\beta_j}{k}} \de_{s2}\right)
  \left< \bphi_j^{(s)}\, , \bJ \right> \mp \nonumber \\
  &\frac{i}{2 c_o k} \left( \sqrt{\frac{\beta_j}{k}} \de_{s1} +
    \sqrt{\frac{k}{\beta_j}}\de_{s2}\right)\left< \nabla
    \rJ_z,\bphi_j^{(s)} \right>\, ,
\end{align}
for the evanescent modes.

\subsubsection{Energy conservation}
\label{sect:energyideal}
The energy conservation is obvious in this case, because the
amplitudes are constant. Substituting (\ref{eq:ID1}-\ref{eq:ID2}) in
the expression of the flux $\cS(z)$ and using the orthogonality
relations (\ref{eq:orthog}), we obtain that
\begin{align}
  \cS(z) = \frac{c_o}{\ep_o} \int_{\Omega} \operatorname{Re} \left[
    \bD_o(\vx) \cdot \overline{\bU_o(\vx)}\right] = \sum_{j=1}^N
  \sum_{s=1}^{\mM_j} \left(\left| A_{j,o}^{\pm(s)}\right|^2 -\left|
      B_{j,o}^{\pm(s)}\right|^2\right)\, , \quad \forall z \in
  \mathbb{R}\, .
\label{eq:cons}
\end{align}
The flux changes value at $z = 0$, where the source lies, but it is
constant for $z \ne 0$,
\begin{equation}
  \cS(|z|) = \frac{c_o}{\ep_o} \sum_{j=1}^N \sum_{s=1}^{\mM_j}
  \left| A_{j,o}^{+(s)}\right|^2 = 
  -\cS(-|z|) = \frac{c_o}{\ep_o} \sum_{j=1}^N \sum_{s=1}^{\mM_j}
  \left| B_{j,o}^{-(s)}\right|^2\, , \qquad z \ne 0\, .
\label{eq:fluxideal}
\end{equation}
The evanescent modes play no role in the transport of energy.
\section{Statement of the problem in the random waveguide}
\label{sect:asympt}
We begin with the model of the small fluctuations. Then we write the
perturbed system of equations for the wave fields, which we analyze in
the remainder of the paper.
\subsection{Model of the fluctuations} 
\label{sect:formMod}
Let us denote by $n(\vx)$ the index of refraction
\begin{equation}
  n(\vx) = \frac{c_o}{c(\vx)} = \sqrt{\frac{\ep(\vx)}{\ep_o}}\, .
\end{equation}
It is the ratio of the electromagnetic wave speeds $c_o$ and $c(\vx) =
1/\sqrt{\ep(\vx)\mu_o}$ in the homogeneous and heterogeneous medium,
respectively.  We model the electrical permittivity by
\begin{equation}
\ep(\vx) = \ep_o n^2(\vx)\,, \qquad 
\label{eq:f16}
  n^2(\vx) =  
  \left[1 + \epsilon \nu(\vx)\right]1_{(0+,z_{\rm max})}(z) \,,
\end{equation}
where $\nu(\vx)$ is a dimensionless random function assumed twice
continuously differentiable, with almost sure bounded derivatives. It
has zero mean
\begin{equation}
\EE\left[\nu(\vx)\right] = 0\, ,
\end{equation}
and it is stationary and mixing in $z$. We refer to \cite[Section
4.6.2]{kushner} for a precise statement of the mixing condition.  It
means in particular that the covariance
\begin{equation}
\cR_\nu(\bx,\bx',z) = \EE\left[\nu(\bx,z) \nu(\bx',0)\right]
\label{eq:f.18}
\end{equation}
is integrable in $z$. The amplitude of the fluctuations in
(\ref{eq:f16}) is scaled by $\epsilon \ll 1$, the small parameter
in our asymptotic analysis.

The indicator function $1_{(0+,z_{\rm max})}(z)$ in (\ref{eq:f16})
limits the support of the fluctuations to the range interval $z \in
(0+,z_{\rm max})$, where $0+$ denotes a range that is close to zero,
but strictly larger than it. The bounded support of the fluctuations
is needed to state the outgoing boundary conditions on the
electromagnetic wave fields, and may be justified in practice by the
causality of the problem in the time domain.  During a finite
observation time $t_{\rm max}$, the waves are influenced by the medium
up to a finite range $ z_{\rm max} \approx c_o t_{\rm max}\, , $ so we
may truncate the fluctuations beyond the range $z_{\rm max}$. That
there are no fluctuations at negative ranges may be motivated by two
facts: First, the source is at $z=0$ and we wish to study the waves at
positive ranges.  Second, we will consider a regime where the
backscattered field is negligible.  Thus, we may neglect at $z>0$ the
waves that come from $z <0$, and truncate the fluctuations at $z =
0+$.

\subsection{The perturbed system of equations in the random waveguide}
\label{sect:formE}
We work with the electric displacement $\bD(\vx)$ and the scaled
rotated magnetic field $\bU(\vx)$, defined in equation
(\ref{eq:defU}). As we explained in the previous section, this is
convenient because the fields satisfy the same boundary conditions and
have the same units.

The equations for $\bD(\vx)$ and $\bU(\vx)$ follow from (\ref{eq:f3}),
(\ref{eq:f8}-\ref{eq:f9}), (\ref{eq:f16}) and (\ref{eq:defU}). We have
\begin{align}
  \partial_z \bD(\vx) =& \frac{i}{k} \left\{k^2 n^2(\vx) \bU(\vx) +
    \nabla \left[\nabla \cdot \bU(\vx)\right] 
    -n^{-2}(\vx) \nabla n^2(\vx) \nabla \cdot \bU(\vx)
  \right\} + \nonumber \\
  & n^{-2}(\vx) \partial_z n^2(\vx) \, \bD(\vx) - \frac{i }{c_o k}
  \nabla \cJ_z(\vx)\, ,
\label{eq:P1}
\end{align}
for the electric displacement and 
\begin{align}
  \partial_z \bU(\vx) = & \frac{i }{k} \left\{ k^2 \bD(\vx) +
    \nabla^\perp \left[\nabla^\perp \cdot \left(n^{-2}(\vx)
        \bD(\vx)\right)\right] \right\} - \frac{1}{c_o}\bcJ(\vx)\, ,
\label{eq:P2}
\end{align}
for the rotated magnetic field, where we used that the fluctuations
are supported away from the source.  Morever, substituting the model
(\ref{eq:f16}) of the fluctuations, we obtain
\begin{align} 
  \partial_z \bD(\vx) =& \frac{i}{k} \left\{ k^2 \bU(\vx) + \nabla
    \left[\nabla \cdot \bU(\vx)\right]\right\} - \frac{i}{c_o k}
  \nabla \cJ_z(\vx) + \nonumber \\
  & \epsilon \left\{
    \partial_z \nu(\vx) \, \bD(\vx) + \frac{i}{k} \left[k^2
      \nu(\vx) \bU(\vx) -\nabla \nu(\vx) \, \nabla \cdot
      \bU(\vx)\right] \right\}
  + \nonumber \\
  &\frac{\epsilon^2}{2} \left[ -\partial_z \nu^2(\vx) \,
    \bD(\vx) + \frac{i}{k} \nabla \nu^2(\vx) \nabla \cdot
    \bU(\vx)\right] + O(\epsilon^3)\, ,
\label{eq:P3}
\end{align}
and
\begin{align} 
  \partial_z \bU(\vx) =& \frac{i}{ k} \left\{ k^2 \bD(\vx) +
    \nabla^\perp \left[\nabla^\perp \cdot \bD(\vx)\right]\right\} -
  \frac{1}{c_o}\bcJ(\vx) -
  \nonumber \\
  &\epsilon \frac{i }{k} \left\{ \nabla^\perp \left[ \nu(\vx)
      \nabla^\perp \cdot \bD(\vx)\right] + \nabla^\perp \left[
      \bD(\vx) \cdot \nabla^\perp \nu(\vx)\right]\right\} +
  \nonumber \\
  &\epsilon^2 \frac{i }{k}\left\{ \nabla^\perp \left[ \nu^2(\vx)
      \nabla^\perp \cdot \bD(\vx) \right] + \nabla^\perp \left[
      \bD(\vx) \cdot \nabla^\perp \nu^2(\vx)\right] \right\} +
  O(\epsilon^3)\, ,\label{eq:P4}
\end{align}
with remainder involving powers $(\epsilon \nu)^q$, for $q \ge 3$.  It
is of order $\epsilon^3$ because $\nu(\vx)$ is twice differentiable,
with almost sure bounded derivatives.  

The leading order terms in (\ref{eq:P3}-\ref{eq:P4}) involve the
operators (\ref{eq:UNPT1}-\ref{eq:UNPT2}), so we have a perturbation
of the problem in the ideal waveguide.  The conservation of the energy
flux follows from (\ref{eq:ENC}) and definitions
(\ref{eq:f3}),(\ref{eq:f16}) and (\ref{eq:defU})
\begin{align}
  \cS(z) &= \frac{c_o}{\ep_o}\int_\Omega d \bx \,
  \operatorname{Re}\left[ n^{-2}(\vx) \,
    \bD(\vx) \cdot \overline{\bU(\vx)} \right]\nonumber \\
  &= \frac{c_o}{\ep_o} \int_\Omega d \bx \, \left[ 1 - \epsilon
    \nu(\vx) + \epsilon^2 \nu^2(\vx) + O(\epsilon^3) \right]
  \operatorname{Re}\left[
    \bD(\vx) \cdot \overline{\bU(\vx)}\right] \nonumber \\
  &= \cS(0+)\, , \qquad z > 0.
\label{eq:P5}
\end{align}

\section{Mode decomposition and coupling in random waveguides}
\label{sect:modec}
The equations in the random waveguide are no longer separable, but $\{
\bphi_j^{(s)}(\bx) \}_{1 \le s \le \mM_j, j \ge 1}$ is an orthonormal
basis, so we can still use it to decompose the wave fields for any
range $z$.  The essential difference in the decomposition is that
while the mode amplitudes are constant in ideal waveguides, they vary
in range in the random waveguides, due to scattering. The range
evolution of the mode amplitudes is described by a coupled system of
infinitely many stochastic ordinary differential equations.  We show
in section \ref{sect:evanesc} that we can solve for the amplitudes of
the evanescent modes, and thus obtain in section \ref{sect:propag} a
closed and finite system of equations for the amplitudes of the
propagating modes. This system is the main result of the section. We
use it in section \ref{sect:diff} to obtain an explicit long range
characterization of the statistical distribution of the
electromagnetic wave field.

\subsection{Mode decomposition}
We decompose the fields as 
\begin{align}
  \hspace{-0.2in} \bD(\vx) =& \sum_{j=1}^N \sum_{s=1}^{\mM_j}
  \bphi_{j}^{(s)}(\bx)\left(\sqrt{\frac{k}{\beta_j}} \de_{s1} +
    \sqrt{\frac{\beta_j}{k}} \de_{s2} \right) \left( A_{j}^{
      (s)}(z) e^{i \beta_j
      z} + B_{j}^{(s)}(z) e^{-i \beta_j z} \right) + \nonumber \\
  &\sum_{j>N} \sum_{s=1}^{\mM_j}
  \bphi_{j}^{(s)}(\bx)\left(\sqrt{\frac{k}{\beta_j}} \de_{s1} +
    \sqrt{\frac{\beta_j}{k}} \de_{s2} \right) V_j^{(s)}(z) \, ,
  \qquad
\label{eq:RD1}
\end{align}
and
\begin{align}
  \bU(\vx) =& \sum_{j=1}^N \sum_{s=1}^{\mM_j}\bphi_{j}^{(s)}(\bx)
  \left(\sqrt{\frac{\beta_j}{k}} \de_{s1} + \sqrt{\frac{k}{\beta_j}}
    \de_{s2} \right) \left( A_{j}^{(s)}(z) e^{i \beta_j
      z} - B_{j}^{(s)}(z) e^{-i \beta_j z} \right) \pm \nonumber \\
  & i \sum_{j>N} \sum_{s=1}^{\mM_j} \bphi_j^{(s)}(\bx)
  \left(\sqrt{\frac{\beta_j}{k}} \de_{s1} - \sqrt{\frac{k}{\beta_j}}
    \de_{s2}\right) v_j^{(s)}(z) \, ,
\label{eq:RD2}
\end{align}
for $z \ne 0$. The decomposition is similar to that in ideal
waveguides, but the mode amplitudes vary in $z$ due to scattering in
the random medium. We show in section (\ref{sect:modecouple}) that the
forward and backward going mode amplitudes $A_j^{(s)}$ and $B_j^{(s)}$
are coupled with each other and with the evanescent modes written in
(\ref{eq:RD1}-\ref{eq:RD2}) as $V_j^{(s)}(z)$ and $v_j^{(s)}(z)$. In
ideal waveguides the evanescent modes were equal to
$\mathfrak{E}_{j,o}^{(s)} {\rm exp}(-\beta_j z)$, for $1\le s \le
\mM_j$ and $j > N$. They have a more complicated expression in random
waveguides, as explained in section \ref{sect:evanesc}.

The expansions (\ref{eq:RD1}-\ref{eq:RD2}) satisfy the boundary
conditions (\ref{eq:CU1}-\ref{eq:CD2}) at $\partial \Omega$.  The
outgoing conditions and the finite range support $(0+,z_{\rm max})$ of
the fluctuations give
\begin{align}
B_j^{(s)}(z_{\rm max}) &= 0 \, ,
\label{eq:RD4}\\ 
A_j^{(s)}(0+) &= A_{j,o}^{(s)}\, .
\label{eq:RD3}
\end{align}
The first equation says that there are no backward going waves coming
from infinity, because there are no fluctuations beyond $z = z_{\rm
  max}$. The second equation follows from the source conditions
\begin{align*}
D(\bx,0+) - D(\bx,0-) &= -\frac{i}{c_o k} \nabla \rJ_z(\bx) = 
D_o(\bx,0+) - D_o(\bx,0-)\,, \\
U(\bx,0+) - U(\bx,0-) &= -\frac{1}{c_o} \bJ(\bx) = 
U_o(\bx,0+) - U_o(\bx,0-)\, ,
\end{align*}
and the outgoing condition $ A_j^{(s)}(z) = 0 $ at ranges $z < 0$,
where the medium is homogeneous. The evanescent modes satisfy
\begin{equation}
\lim_{|z| \to \infty}V_j^{(s)}(z) = \lim_{|z| \to \infty}v_j^{(s)}(z) 
= 0\, .
\label{eq:RD5}
\end{equation}
\subsection{Mode coupling}
\label{sect:modecouple}
Substituting (\ref{eq:RD1}-\ref{eq:RD2}) in (\ref{eq:P3}-\ref{eq:P4})
and using identities (\ref{eq:UNPT1}-\ref{eq:UNPT2}) and the
orthogonality relation (\ref{eq:orthog}), we obtain a system of
stochastic differential equations that describes the range evolution
of the mode amplitudes. The rate of change of the amplitudes of the
forward going modes is given by
\begin{align}
  \partial_z A_j^{(s)}(z) =& \epsilon\sum_{j' = 1}^N
  \sum_{s'=1}^{\mM_{j'}} \left[M_{AA,jj'}^{(ss')}(z) + \epsilon \,
    m_{AA,jj'}^{(ss')}(z)\right] A_{j'}^{(s')}(z)
  \, e^{i (\beta_{j'}-\beta_j) z} + \nonumber \\
  & \epsilon\sum_{j' = 1}^N \sum_{s'=1}^{\mM_{j'}}
  \left[M_{AB,jj'}^{(ss')}(z) + \epsilon \,
    m_{AB,jj'}^{(ss')}(z)\right] B_{j'}^{(s')}(z)
  \, e^{-i (\beta_{j'}+\beta_j) z} + \nonumber \\
  & \epsilon \sum_{j' > N} \sum_{s'=1}^{\mM_{j'}}
  \left[M_{AV,jj'}^{(ss')}(z) + \epsilon \,
    m_{AV,jj'}^{(ss')}(z)\right] V_{j'}^{(s')}(z) \, e^{-i \beta_j z}
  +
  \nonumber \\
  & \epsilon \sum_{j' > N} \sum_{s'=1}^{\mM_{j'}}
  \left[M_{Av,jj'}^{(ss')}(z) + \epsilon \,
    m_{Av,jj'}^{(ss')}(z)\right] v_{j'}^{(s')}(z) \, e^{-i \beta_j z}
  + O(\epsilon^3)\, , \label{eq:RD6}
\end{align}
for $z > 0$, with initial condition (\ref{eq:RD3}). The rate of change
of the amplitudes of the backward moving modes is
\begin{align}
  \partial_z B_j^{(s)}(z) =& \epsilon\sum_{j' = 1}^N
  \sum_{s'=1}^{\mM_{j'}} \left[M_{BA,jj'}^{(ss')}(z) + \epsilon \,
    m_{BA,jj'}^{(ss')}(z)\right] A_{j'}^{(s')}(z)
  \, e^{i (\beta_{j'}+\beta_j) z} + \nonumber \\
  & \epsilon\sum_{j' = 1}^N \sum_{s'=1}^{\mM_{j'}}
  \left[M_{BB,jj'}^{(ss')}(z) + \epsilon \,
    m_{BB,jj'}^{(ss')}(z)\right] B_{j'}^{(s')}(z)
  \, e^{-i (\beta_{j'}-\beta_j) z} + \nonumber \\
  & \epsilon \sum_{j' > N} \sum_{s'=1}^{\mM_{j'}}
  \left[M_{BV,jj'}^{(ss')}(z) + \epsilon \,
    m_{BV,jj'}^{(ss')}(z)\right] V_{j'}^{(s')}(z) \, e^{i \beta_j z} +
  \nonumber \\
  & \epsilon \sum_{j' > N} \sum_{s'=1}^{\mM_{j'}}
  \left[M_{Bv,jj'}^{(ss')}(z) + \epsilon \,
    m_{Bv,jj'}^{(ss')}(z)\right] v_{j'}^{(s')}(z) \, e^{i \beta_j z} +
  O(\epsilon^3)\, , \label{eq:RD7}
\end{align}
for $z > 0$, with end condition (\ref{eq:RD4}) at $z = z_{\rm max}$.
The evanescent components $V_{j}^{(s)}(z)$ and $v_{j}^{(s)}(z)$ are
described in the next section.

The coupling coefficients in the right hand side of equations
(\ref{eq:RD6}-\ref{eq:RD6}) are stationary random processes in $z$,
defined in terms of the fluctuations $\nu$. We refer to appendix
\ref{ap:1} for their expression and symmetry relations.  The leading
order terms of these coefficients, denoted by the capital letter $M$
as in $M_{AA,jj'}^{(ss')}(z)$, are linear in $\nu$, so they have zero
expectation.  The second order terms, denoted by the small letter $m$
as in $m_{AA,jj'}^{(ss')}(z)$, are quadratic in $\nu$.

\subsection{The evanescent modes}
\label{sect:evanesc}
The evanescent modes satisfy the equations 
\begin{align}
  \partial_z V_j^{(s)}(z) + \beta_j v_j^{(s)}(z) = & \epsilon
  F_j^{(s)}(z) + \epsilon \sum_{j'>N} \sum_{s'=1}^{\mM_{j'}}
  M_{VV,jj'}^{(ss')}(z) \, V_{j'}^{(s')}(z)\,  + \nonumber \\
&\epsilon \sum_{j'>N} \sum_{s'=1}^{\mM_{j'}} M_{Vv,jj'}^{(ss')}(z)
    \, v_{j'}^{(s')}(z) + O(\epsilon^2) \, , \label{eq:ev1} 
\end{align}
and 
\begin{align}
  \partial_z v_j^{(s)}(z) + \beta_j V_j^{(s)}(z) &= \epsilon
  f_j^{(s)}(z) + \epsilon  \sum_{j'>N}\sum_{s'=1}^{\mM_{j'}}
  M_{vV,jj'}^{(ss')}(z) \, V_{j'}^{(s')}(z) + O(\epsilon^2)\, ,
\label{eq:ev2}
\end{align}
for $z > 0$,  with forcing terms 
\begin{align}
  F_j^{(s)}(z) &= \sum_{j'=1}^N \sum_{s'=1}^{\mM_{j'}} \left[
    M_{VA,jj'}^{(ss')}(z)\, A_{j'}^{(s')}(z) \, e^{i \beta_{j'} z} +
    M_{VB,jj'}^{(ss')}(z)\, B_{j'}^{(s')}(z)\, e^{-i \beta_{j'}
      z}\right]\, ,
  \label{eq:ev3} \\
  f_j^{(s)}(z) &= \sum_{j'=1}^N \sum_{s'=1}^{\mM_{j'}} \left[
    M_{vA,jj'}^{(ss')}(z)\, A_{j'}^{(s')}(z) \, e^{i \beta_{j'} z} +
    M_{vB,jj'}^{(ss')}(z)\, B_{j'}^{(s')}(z) \, e^{-i \beta_{j'}
      z}\right]\,.
\label{eq:ev4}
\end{align}
The coupling coefficients are described in appendix \ref{ap:1}. They
are stationary processes in $z$ that depend linearly on the
fluctuations $\nu$.  

The system of equations (\ref{eq:ev1}-\ref{eq:ev2}) is solved in
appendix \ref{ap:2}. We state the result in Lemma \ref{lem.2} which we
use in the next section to obtain a closed system of equations for the
propagating mode amplitudes.
\begin{lemma}
\label{lem.2}
The evanescent modes are given by 
\begin{align}
  V_j^{(s)}(z) = &\mathfrak{E}_{j,o}^{(s)} \, e^{-\beta_j z} + \frac{\epsilon}{2}
  \int_{-\infty}^\infty d \zeta
  \, f_j^{(s)}(z+\zeta) \, e^{-\beta_j |\zeta|} + \nonumber \\
  & \frac{\epsilon}{2} \int_{0}^\infty d \zeta \,
  \left[F_j^{(s)}(z-\zeta)-F_j^{(s)}(z+\zeta)\right] e^{-\beta_j
    \zeta} + O(\epsilon^2)\, ,
\label{eq:ev5}
\end{align}
and 
\begin{align}
  v_j^{(s)}(z) =&\mathfrak{E}_{j,o}^{(s)}\, e^{-\beta_j z} + \frac{\epsilon}{2}
  \int_{-\infty}^\infty d \zeta
  \, F_j^{(s)}(z+\zeta)\, e^{-\beta_j |\zeta|} + \nonumber \\
  & \frac{\epsilon}{2} \int_{0}^\infty d \zeta \,
  \left[f_j^{(s)}(z-\zeta)-f_j^{(s)}(z+\zeta)\right] e^{-\beta_j
    \zeta} + O(\epsilon^2)\, .
\label{eq:ev6}
\end{align}
\end{lemma}

The first terms in these equations are as in ideal waveguides.  They
decay exponentially with $z$ and have a negligible contribution at
long ranges. The $O(\epsilon)$ terms capture the coupling with the
propagating modes and have long range effects in equations
(\ref{eq:RD6}-\ref{eq:RD7}). The remaining terms are negligible in the
limit $\epsilon \to 0$.

\subsection{Closed system for the propagating modes}
\label{sect:propag}
The substitution of the evanescent mode equations
(\ref{eq:ev5}-\ref{eq:ev6}) in (\ref{eq:RD6}-\ref{eq:RD7}) gives a
closed system of ordinary differential equations for the amplitudes of
the $N$ forward and backward going modes
\begin{align}
  \partial_z A_j^{(s)}(z) =& \epsilon\sum_{j' = 1}^N
  \sum_{s'=1}^{\mM_{j'}} \left[M_{AA,jj'}^{(ss')}(z) + \epsilon
    \widetilde m_{AA,jj'}^{(ss')}(z) \right] A_{j'}^{(s')}(z)
  e^{i (\beta_{j'}-\beta_j) z} + \nonumber \\
  & \epsilon\sum_{j' = 1}^N \sum_{s'=1}^{\mM_{j'}}
  \left[M_{AB,jj'}^{(ss')}(z) + \epsilon 
      \widetilde m_{AB,jj'}^{(ss')}(z) \right]
  B_{j'}^{(s')}(z )e^{-i (\beta_{j'}+\beta_j) z}+ O(\epsilon^3)\, ,
  \label{eq:C1}
\end{align}
and
\begin{align}
  \partial_z B_j^{(s)}(z) =& \epsilon\sum_{j' = 1}^N
  \sum_{s'=1}^{\mM_{j'}} \left[M_{BA,jj'}^{(ss')}(z) + \epsilon \,
    \widetilde m_{BA,jj'}^{(ss')}(z)\right] A_{j'}^{(s')}(z)
  e^{i (\beta_{j'}+\beta_j) z} + \nonumber \\
  & \epsilon\sum_{j' = 1}^N \sum_{s'=1}^{\mM_{j'}}
  \left[M_{BB,jj'}^{(ss')}(z) + \epsilon \,
    \widetilde m_{BB,jj'}^{(ss')}(z)\right] B_{j'}^{(s')}(z)
  e^{-i (\beta_{j'}-\beta_j) z} + 
  O(\epsilon^3)\, . \label{eq:C2}
\end{align}
Here we let 
\begin{align*}
\widetilde m_{AA,jj'}^{(ss')}(z)&= m_{AA,jj'}^{(ss')}(z)+
m_{AA,jj'}^{(ss')e}(z)\, , \\
\widetilde m_{AB,jj'}^{(ss')}(z)&= m_{AB,jj'}^{(ss')}(z)+
m_{AB,jj'}^{(ss')e}(z)\, , \\
\widetilde m_{BA,jj'}^{(ss')}(z)&= m_{BA,jj'}^{(ss')}(z)+
m_{BA,jj'}^{(ss')e}(z)\, ,\\
\widetilde m_{BB,jj'}^{(ss')}(z)&= m_{BB,jj'}^{(ss')}(z)+
m_{BB,jj'}^{(ss')e}(z)\, ,
\end{align*}
with the second terms due to the interaction via the evanescent modes.
They are written explicitly in appendix \ref{ap:2coeff}.

\subsection{Energy conservation}
\label{sect:encons}
Substituting equations (\ref{eq:RD1}-\ref{eq:RD2}) in the energy flux
(\ref{eq:P5}) and using Lemma \ref{lem.2} we obtain that 
\begin{align}
  \sum_{j=1}^N \sum_{s=1}^{\mM_j} \left[
    \left|A_j^{(s)}(z)\right|^2 - \left|B_j^{(s)}(z)\right|^2 \right] 
  =   
  ~ \sum_{j=1}^N \sum_{s=1}^{\mM_j} \left[
    \left|A_{j,o}^{(s)}\right|^2 - 
    \left|B_j^{(s)}(0+)\right|^2 \right] + O(\epsilon) \,.
\label{eq:P5r}
\end{align}
The evanescent modes do not contribute to leading order in the energy
flux, but they appear in the remainder $O(\epsilon)$. Consequently,
the energy carried by the propagating modes is not exactly 
conserved for $\epsilon >0$. However, energy conservation holds in the 
limit $\epsilon \to 0$, where the remainder becomes negligible. 
\section{The diffusion limit}
\label{sect:diff}
In this section we describe the limit $\epsilon \to 0$ of the
propagating mode amplitudes satisfying the system of equations
(\ref{eq:C1}-\ref{eq:C2}) for $z >0$, with initial conditions
(\ref{eq:RD3}) at $z = 0$ and end conditions (\ref{eq:RD4}) at $z =
z_{\rm max}$. 

Since $\partial_z A_j^{(s)}(z)$ and $\partial_z B_j^{(s)}(z)$ are
order $\epsilon$, it is clear that the fluctuations have no effect
over ranges $z$ that are of order one, i.e., similar to the
wavelength. If we let $z$ be of order $\epsilon^{-1}$, the right
hand-side in (\ref{eq:C1}-\ref{eq:C2}) becomes order one, but still
there is no net scattering effect in the limit $\epsilon \to 0$.  The
fluctuations average out because the expectation of the leading
coupling coefficients $M_{AA,jj'}^{(ss')}(z/\epsilon), \ldots,
M_{AA,jj'}^{(ss')}(z/\epsilon)$ is zero. See for example
\cite{khas1966limit,PK-75} and \cite[Chapter 6]{LAY_BOOK}.  We need
longer ranges, of order $\epsilon^{-2}$, to see cumulative scattering
effects, so we let $z = Z/\epsilon^2$ with $Z$ of order one, and
rename the mode amplitudes in this scaling as
\begin{equation}
  A_j^{\epsilon (s)}(Z) :=
  A_j^{(s)}\hspace{-0.05in}\left({Z}/{\epsilon^2}\right)\, , \qquad
  B_j^{\epsilon (s)}(Z) :=
  B_j^{(s)}\hspace{-0.05in}\left({Z}/{\epsilon^2}\right)\, ,
  \label{eq:DL1}
\end{equation}
for $j = 1, \ldots, N,$ and $1 \le s \le \mM_j$. Their $\epsilon \to
0$ limit is obtained with the diffusion approximation theorem
\cite{PK-74}.  The result is simpler under the forward scattering
approximation described in section \ref{sect:FSC}, which is valid when
the covariance (\ref{eq:f.18}) of $\nu(\vx)$ is
smooth in $z$. The limit of the forward going mode amplitudes is
described in detail in section \ref{sect:LIM}. This is the main result
of the section. We use it to analyze the long range cumulative
scattering effects of the random fluctuations in section
\ref{sect:transp}.

\subsection{The forward scattering approximation}
\label{sect:FSC}
The diffusion approximation theorem applies to initial value problems,
so we transform our system to such a problem using the random
propagator matrix ${\bf P}^\epsilon(Z)$. It equals the identity ${\bf
  I}$ at $Z = 0$ and relates the mode amplitudes at $Z >0$ to those at
$Z = 0$ as
\begin{equation}
\left( \begin{matrix}
{\bf A}^\epsilon(Z)\\
{\bf B}^\epsilon(Z)\end{matrix}\right) = 
{\bf P}^\epsilon(Z)\left(\begin{matrix}
{\bf A}_o \\
{\bf B}^\epsilon(0)\end{matrix}\right) \, .
\label{eq:DL2}
\end{equation}
Here ${\bf A}^\epsilon(Z)$ is the vector of components 
$A_j^{\epsilon (s)}(Z)$ for $j = 1, \ldots, N$, $1 \le s \le \mM_j$, and 
similar for ${\bf B}^\epsilon(Z)$. The backward going amplitudes are 
not known at $Z = 0$, but can be determined from the identity
\begin{equation}
  \left( \begin{matrix}
      {\bf A}^\epsilon(Z_{\rm max})\\
      {\bf 0}\end{matrix}\right) = 
  {\bf P}^\epsilon(Z_{\rm max})\left(\begin{matrix}
      {\bf A}_o \\
      {\bf B}^\epsilon(0)\end{matrix}\right) \, , 
  \qquad Z_{\rm max} = \epsilon^2 z_{\rm max}.
\label{eq:DL3}
\end{equation}

The diffusion approximation theorem \cite{PK-74} states that ${\bf
  P}^\epsilon(Z)$ converges in distribution as $\epsilon \to 0$ to a
matrix valued diffusion process ${\bf P}(Z)$.  That is to say, the
entries of ${\bf P}(Z)$ satisfy a system of stochastic differential
equations with initial condition ${\bf P}(0) = {\bf I}$.  We do not
need to write all the details of the limit for the analysis below. Let
us just note that it has the block structure
\[ {\bf P}(Z) = \left( \begin{matrix}
    {\bf P}_{AA}(Z) & {\bf P}_{AB}(Z)\\
    {\bf P}_{BA}(Z) & {\bf P}_{BB}(Z)
\end{matrix} \right),
\]
with entries determined by the $z$--Fourier transform $\hat
\cR_\nu(\bx,\bx',\beta)$ of the covariance (\ref{eq:f.18}), evaluated
at various values of the wavenumber $\beta$. Explicitly, for the
entries in the block ${\bf P}_{AA}(Z)$ that couple the $j$ and $j'$
forward going amplitudes, $\beta = \beta_j - \beta_{j'}$, because the
phases in the first sum in (\ref{eq:C1}) are proportional to $\beta_j
- \beta_{j'}$. Similarly, for the entries in the blocks ${\bf
  P}_{AB}(Z)$ and ${\bf P}_{BA}(Z)$ that couple the $j$ and $j'$
forward and backward going amplitudes, $\beta = \beta_j + \beta_{j'}$,
because the phases in the second sum in (\ref{eq:C1}) and the first
sum in (\ref{eq:C2}) are proportional to $\beta_j + \beta_{j'}$.
Thus, if the covariance is smooth enough in $z$, so that 
\begin{equation}
\left|\hat  \cR_\nu(\bx,\bx',\beta_j + \beta_{j'})\right| \ll 1\, , 
\qquad \forall ~ j, j' = 1, \ldots, N\, ,
\end{equation}
the forward and backward mode amplitudes are essentially uncoupled.
Considering that ${\bf B}^\epsilon(Z)$ vanishes at $Z_{\rm max}$, we
conclude that the backward going mode amplitudes are negligible, and
thus justify the forward scattering approximation.

\subsection{The coupled mode diffusion process}
\label{sect:LIM}
Equations (\ref{eq:C1}) simplify as 
\begin{align}
  \partial_Z {\bf A}^{\epsilon}(Z) &\approx \frac{1}{\epsilon} {\bf G}
  \left[{\bf
      A}^{\epsilon}(Z),\nu\left(\cdot,\frac{Z}{\epsilon^2}\right),
    \frac{Z}{\epsilon^2}\right] + {\bf g}\left[{\bf
      A}^{\epsilon}(Z),\nu\left(\cdot,\frac{Z}{\epsilon^2}\right),
    \frac{Z}{\epsilon^2}\right]\, , \quad Z >0\, ,
\label{eq:DL4}
\end{align}
with initial conditions $ {\bf A}^\epsilon(0) = {\bf A}_o , $ and
right hand-side
\begin{align}
{\bf G}
\left[{\bf A}^{\epsilon}(Z),\nu\left(\cdot,\frac{Z}{\epsilon^2}
  \right),\frac{Z}{\epsilon^2}\right] &= {\bf
  M}\left[\nu\left(\cdot,\frac{Z}{\epsilon^2}\right),
  \frac{Z}{\epsilon^2}\right]{\bf A}^{\epsilon}(Z)\, , \label{eq:DL5}\\
{\bf g} \left[{\bf
    A}^{\epsilon}(Z),\nu\left(\cdot,\frac{Z}{\epsilon^2}\right),
  \frac{Z}{\epsilon^2}\right] &= \widetilde {\bf
  m}\left[\nu\left(\cdot,\frac{Z}{\epsilon^2}\right),
  \frac{Z}{\epsilon^2}\right]{\bf A}^{\epsilon}(Z)\,.
\label{eq:DL6}
\end{align}
Here we let ${\bf M}$ be the matrix with entries
$M_{AA,jj'}^{(ss')}(Z/\epsilon^2)
e^{i(\beta_j-\beta_{j'})Z/\epsilon^2}$, and emphasize in the notation
that it depends on $Z/\epsilon^2$ via the fluctuations $\nu$ and the
phase. A similar notation applies to matrix $\widetilde {\bf m}$. The
approximation sign in (\ref{eq:DL4}) reminds us that we made the
forward scattering approximation and neglected the $O(\epsilon)$
remainder that plays no role in the limit $\epsilon \to 0$.

To apply the diffusion approximation theorem stated and proved in
\cite{PK-74} to (\ref{eq:DL4}), we rewrite the system in real form,
for the concatenated vector $({\bf A}_R^\epsilon,{\bf A}_I^\epsilon)$
of real and imaginary values of ${\bf A}^\epsilon$. We also recall
from complex differentiation that for any vector ${\bf a} = {\bf a}_R
+ i {\bf a}_I$, we have
\[
\nabla_{{\bf a}_R} = \nabla_{\bf a} + \nabla_{\overline{\bf a}}, \quad
\nabla_{{\bf a}_I} = i \left(\nabla_{\bf a} - \nabla_{\overline{\bf
      a}}\right)\, ,
\]
where the bar denotes complex conjugation. Therefore, if we let ${\bf
  G}_R$ and ${\bf G}_I$ be the real and imaginary parts of ${\bf G}$,
we can write 
\[ ({\bf G}_R,{\bf G}_I) \cdot (\nabla_{{\bf a}_R},\nabla_{{\bf a}_I})
= {\bf G} \cdot \nabla_{{\bf a}} + \overline{{\bf G}}\cdot
\nabla_{{\overline{\bf a}}}\, .
\]
With these observations we state in the next lemma the limit $\epsilon
\to 0$ given by the diffusion approximation theorem.

\begin{lemma}
\label{lem:limit}
The mode amplitudes $\{A_j^{\epsilon(s)}(Z)\}_{j=1, \ldots, N, 1 \le s
  \le \mM_j}$ converge in distribution as $\epsilon \to 0$ to a
diffusion Markov process denoted by $\{A_j^{(s)}(Z)\}_{j=1, \ldots, N,
  1 \le s \le \mM_j}$, with generator ${\mathcal G}$. It is defined on
smooth enough, scalar valued test functions $\varphi({\bf
  A},\overline{{\bf A}})$ as follows
\begin{align*} {\mathcal G} \varphi({\bf A},\overline{{\bf A}}) =&
  \lim_{T \to \infty} \int_0^T \frac{d \tau}{T} \int_0^\infty dz \,
  \EE\left\{ \left[{\bf G}\left[ {\bf A}, \nu(\cdot, 0), \tau\right]
      \cdot \nabla_{{\bf A}} + \overline{{\bf G}}\left[ {\bf A},
        \nu(\cdot, 0), \tau\right] \cdot \nabla_{\overline{\bf
          A}}\right] \times \right. \nonumber \\
  & \left. \hspace{0.38in}\left[{\bf G}\left[ {\bf A}, \nu(\cdot, z),
        \tau+z\right] \cdot \nabla_{{\bf A}} + \overline{{\bf
          G}}\left[ {\bf A}, \nu(\cdot, z), \tau+z\right] \cdot
      \nabla_{\overline{\bf A}}\right] \right\} \varphi({\bf
    A},\overline{\bf A}) + \nonumber \\
  & \lim_{T \to \infty} \int_0^T \frac{d \tau}{T} \, \EE\left\{
    \left[{\bf g}\left[ {\bf A}, \nu(\cdot, 0), \tau\right] \cdot
      \nabla_{{\bf A}} + \overline{{\bf g}}\left[ {\bf A}, \nu(\cdot,
        0), \tau \right] \right]\cdot \nabla_{\overline{\bf A}}
  \right\} \varphi({\bf A},\overline{\bf A}).
\end{align*}
\end{lemma}
\subsection{Conservation of energy}
\label{sect:DLCons}
Recall the conservation relation (\ref{eq:P5r}), and rewrite it using
the forward scattering approximation as
\begin{equation}
  \sum_{j=1}^N \sum_{s=1}^{\mM_j} \left| A_j^{\epsilon(s)}(Z)
  \right|^2 = \sum_{j=1}^N \sum_{s=1}^{\mM_j} \left| A_{j,o} \right|^2
  + \mathcal{R}(\epsilon)\, ,
\label{eq:DLC1}
\end{equation}
with negligible remainder $\mathcal{R}(\epsilon)$ as $\epsilon \to 0$.
The diffusion limit gives that 
\begin{equation}
 \sum_{j=1}^N \sum_{s=1}^{\mM_j} \left| A_j^{\epsilon(s)}(Z)
  \right|^2  \stackrel{\epsilon \to 0}{\longrightarrow} 
\sum_{j=1}^N \sum_{s=1}^{\mM_j} \left| A_j^{(s)}(Z)
  \right|^2 = \sum_{j=1}^N \sum_{s=1}^{\mM_j} \left| A_{j,o} \right|^2\, ,
\label{eq:DLC2}
\end{equation}
where the convergence is in probability, because the limit is
deterministic.

\section{Cumulative scattering effects}
\label{sect:transp}
We use the limit stated in Lemma \ref{lem:limit} to derive the main
result of the paper: a detailed characterization of cumulative
scattering effects of the random fluctuations of the electric
permeability.

We begin in sections \ref{sect:moments1} and \ref{sect:moments2} with
the calculation of the first and second moments of the mode
amplitudes.  They determine the coherent part of the waves and the
intensity of their fluctuations.  Then, we start from the energy
conservation relation (\ref{eq:DLC2}) and derive in section
\ref{sect:CE} an important matrix identity, needed in sections
\ref{sect:LC} and \ref{sect:TR} to describe the loss of coherence of
the waves and the energy exchange between the modes. We also prove in
section \ref{sect:TR} that as the range grows, the waves scatter so
much that they enter the equipartition regime, where they forget all
the information about the source. We illustrate the results of the
analysis with numerical simulations.
\subsection{The mean  mode amplitudes}
\label{sect:moments1}
Let us denote by 
\begin{equation}
\left< {\bf A} \right>\hspace{-0.03in}(Z) = \EE\left\{{\bf
  A}(Z)\right\}\,,
\label{eq:E1}
\end{equation}
the expectation of the mode amplitudes with respect to their limit
distribution. Using the generator $\mathcal{G}$ in Lemma
\ref{lem:limit} and Kolmogorov's equation \cite[chapter
  8]{oksendal2003stochastic}, we obtain 
\begin{equation}
  \partial_Z\left< {\bf A}_j \right> \hspace{-0.03in} (Z) = {\bf Q}_j
  \left< {\bf A}_j\right> \hspace{-0.03in} (Z)\, , \quad Z>0\, ,
\label{eq:E2}
\end{equation}
with initial condition 
\begin{equation}
\left< {\bf A}_j\right> \hspace{-0.03in} (0) = {\bf A}_{j,o}\, .
\label{eq:E3}
\end{equation}
This is a block diagonal system of differential equations, for vectors
${\bf A}_j$ of components ${ A}_j^{(s)}$. There are $N$ blocks
${\bf Q}_j \in \mathbb{C}^{\mM_j \times \mM_j}$, indexed by $j = 1,
\ldots, N$. Each one of them is constant, with entries given by
\begin{equation} {\bf Q}_j^{(ss')} = \sum_{l=1}^N
  \sum_{q=1}^{\mM_l}\int_0^\infty d z\, \EE\left\{M_{AA,jl}^{(sq)}(z) 
M_{AA,lj}^{(qs')}(0) \right] e^{i (\beta_l-\beta_j)z} + 
 E\left\{m_{j}^{(ss')}(0)\right\}\, ,
\label{eq:E4}
\end{equation}
where we introduced the simplified notation 
\begin{equation}
m_{j}(z) := \widetilde m_{AA,jj}(z).
\label{eq:E4p}
\end{equation}
We give a few details of the calculation of ${\bf Q}_j$ in appendix
\ref{ap:explicit}, and use the result in the numerical simulations of
sections \ref{sect:LC} and \ref{sect:TR}. Here it suffices to point
out that the last term in (\ref{eq:E4}) is purely imaginary, so we can
write it as
\begin{equation}
E\left\{m_{j}(0)\right\} = i {\boldsymbol \kappa}_j\, ,
\label{eq:E6}
\end{equation}
with real matrix ${\boldsymbol \kappa}_j \in \mathbb{R}^{\mM_j \times
  \mM_j}$. This is the only term of ${\bf Q} = {\rm diag}\left({\bf
  Q}_1, \ldots, {\bf Q}_N\right)$ that is affected by the coupling of
the propagating modes with the evanescent ones.

The mean amplitudes are decoupled for different indexes $j$ of the
modes.  However, for each $j$ we have $\mM_j$ coupled transverse
electric and magnetic mode amplitudes, as described by the matrix
exponential in
\begin{equation}
  \left< {\bf A}_j\right> \hspace{-0.03in} (Z) = e^{{\bf Q}_jZ} {\bf A}_{j,o}\, , 
\qquad j = 1, \ldots, N.
\label{eq:E5}
\end{equation}
We expect from physical arguments that the right hand-side in
(\ref{eq:E5}) decays with $Z$, on some mode dependent range scales
${\mathcal S}_j^{(s)}$, the scattering mean free paths. The coherent
part of the amplitudes, the entries in $\left<{\bf
  A}_j\right> \hspace{-0.03in} (Z)$, become negligible beyond these
scales, and all the energy lies in their random fluctuations.

It is difficult to see the loss of coherence directly from
(\ref{eq:E4}). The expression of ${\bf Q}_j$ in appendix
\ref{ap:explicit} is useful for numerical calculations, but it is too
complicated to prove that the spectrum of ${\bf Q}_j$ lies in the left
half of the complex plane. However, the result follows from the energy
conservation relation (\ref{eq:DLC2}), as explained in section
\ref{sect:LC}.

\subsection{The mean powers}
\label{sect:moments2}
We denote the mean power matrices of the amplitudes of the modes
with wavenumber $\beta_j$ by 
\begin{equation} {\bf P}_j(Z) = \left(P_{j}^{ss'}(Z)\right)_{1\le
    s,s'\le \mM_j} := \EE \left\{ {\bf A}_j(Z) \otimes \overline{\bf
      A}_j(Z)\right\}.
\label{eq:MP1}
\end{equation}
They are Hermitian, positive definite matrices, satisfying a coupled
system of differential equations derived from the generator in Lemma
\ref{lem:limit} and Kolmogorov's equation. Explicitly, we have
\begin{align}
  \partial_Z {\bf P}_j(Z) = & {\bf Q}_j{\bf P}_j(Z) +
  {\bf P}_j(Z) {\bf Q}_j^\star + \nonumber \\
  & \sum_{l=1}^N \int_{-\infty}^ \infty dz \EE\left\{ M_{AA,jl}(z)
    {\bf P}_l(Z) M_{AA,jl}^\star(0)\right\} e^{i (\beta_l-\beta_j)z}\,
  ,
\label{eq:MP2}
\end{align}
for $Z >0$, with initial condition 
\begin{equation}
{\bf P}_j(0) = {\bf A}_{j,o} \otimes \overline{\bf A}_{j,o}.
\label{eq:MP3}
\end{equation}
The matrix ${\bf Q}_j$ is defined in (\ref{eq:E4}), and the star
superscript denotes complex conjugate and transpose.

Equations (\ref{eq:MP2}) describe the exchange of energy between the
modes and the loss of polarization of the waves. Say for example that
the source emits a single transverse electric mode indexed by $j$
\[ {\bf P}_{l}(0) = 
\de_{lj}\left(\begin{matrix}
    |A_{j,o}^{(1)}|^2 & 0 \\
    0 & 0
\end{matrix}\right), \quad \forall~ l = 1, \ldots, N\,.
\] 
Cumulative scattering distributes the energy to all propagating modes
for $Z>0$, as given by (\ref{eq:MP2}), and the wave loses its
initial polarization.
\subsection{Conservation of energy identity}
\label{sect:CE}
The conservation of energy relation (\ref{eq:DLC2}) states that 
the mean power matrices satisfy 
\begin{equation}
  \sum_{j=1}^N {\rm trace}[P_j(Z)] = \sum_{j=1}^N {\rm trace}[P_j(0)]
= \sum_{j=1}^N \sum_{s=1}^{\mM_j} \left| A_{j,o}^{(s)}\right|^2 .
\label{eq:MP4}
\end{equation}
Therefore, equations (\ref{eq:MP2}) and the properties of the
trace operator imply that 
\begin{equation}
  \sum_{j=1}^N {\rm trace}\hspace{-0.01in}\left[\left( {\bf Q}_j + 
    {\bf Q}_j^\star + {\bf C}_j \right) {\bf P}_j(Z) \right] = 0\,, 
\quad \forall ~ Z \ge 0\,,
 \label{eq:MP5}
\end{equation}
with Hermitian matrix ${\bf C}_j$ defined by
\begin{equation} 
{\bf C}_j = \sum_{l=1}^N \int_{-\infty}^ \infty dz\, 
  \EE\left\{ M_{AA,lj}^\star(z) M_{AA,lj}(0)\right\} e^{i
    (\beta_l-\beta_j)z}\,.
\label{eq:MP6}
\end{equation}
The terms in this sum are the power spectral densities of the
stationary, matrix valued processes $M_{AA,jl}(z)$, evaluated at the
wavenumber difference $\beta_j - \beta_l$. This implies that ${\bf
  C}_j$ is a positive definite matrix, as shown in appendix
\ref{ap:psd}.

The following lemma gives a matrix identity used in the next
sections to prove the loss of coherence of the waves and the
equipartition regime as $Z \to \infty$.
\begin{lemma}
\label{lem:MId}
The matrices ${\bf Q}_j$ and ${\bf C}_j$ defined by 
(\ref{eq:E4}) and (\ref{eq:MP6}) satisfy 
\begin{equation}
{\bf Q}_j + {\bf Q}_j^\star + {\bf C}_j = {\bf 0}\, , 
\qquad \forall ~ j = 1, \ldots, N.
\end{equation}
\end{lemma}

\vspace{0.1in} 
\noindent \emph{Proof.}  The result is a consequence of the fact that
(\ref{eq:MP5}) holds for any correlation matrices ${\bf P}_j(Z)$ and
all $Z \ge 0$.  Indeed, let $\mathfrak{X}_j$ be the $\mM_j^2$
dimensional vector space of $\mM_j \times \mM_j$ Hermitian matrices
with inner product
  \[
  \left(U,V\right)_{\mathfrak{X}_j} = {\rm trace}[U V^\star]\, , \quad
  \forall ~ U, V \in \mathfrak{X}_j.
  \]
  Let also ${\mathfrak X} = {\mathfrak X}_1 \times {\mathfrak X}_1
  \ldots \times {\mathfrak X}_N$ be the vector space defined by the
  product of the spaces $ {\mathfrak X}_j$, with inner product
\[
\left( \textbf{U},\textbf{V} \right)_{\mathfrak X} = \sum_{j=1}^N \left( U_j,V_j
\right)_{\mathfrak X_j}, \quad \forall~ \textbf{U} = (U_1, \ldots, U_N), ~ ~ \textbf{V}
= (V_1, \ldots, V_N), ~ ~ U_j, V_j \in {\mathfrak X}_j\, .
\]
Equation (\ref{eq:MP5}) evaluated at $Z = 0$ becomes 
\[
\left( {\bf Q} + {\bf Q}^\star + {\bf C},{\bf P}_o \right)_{\mathfrak
  X} = 0\, , \qquad \forall ~ {\bf P}(0) = {\bf P}_o \in {\mathfrak
  X}.
\]
We can take in particular the initial conditions
\[ {\bf P}_{o} = ({\bf 0}, \ldots, {\bf 0}, {\bf P}_{j,o}, {\bf 0},
\ldots, {\bf 0})\,, \qquad  \forall ~{\bf P}_{j,o} = {\bf
  A}_{j,o} \otimes {\bf A}_{j,o}^\star \in {\mathfrak X}_j,
\] 
and conclude that
\[
\left( {\bf Q}_j + {\bf Q}_j^\star + {\bf C}_j,{\bf P}_{j,o}
\right)_{\mathfrak X_j} = 0.
\]
The statement of the lemma follows from this equation and the
observation that matrices like ${\bf P}_{j,o}$ span ${\mathfrak X}_j$.
For example,
\begin{align*}
\left(\begin{matrix} 1\\ 0\end{matrix}\right)(1,0) =  
\left(
\begin{matrix}
  1 & 0\\
  0 & 0
\end{matrix}
\right), \quad
\left(\begin{matrix} 0\\ 1\end{matrix}\right)(0,1) =\left(
\begin{matrix}
  0 & 0\\
  0 & 1
\end{matrix}
\right), \quad 
\left(\begin{matrix} 1\\ 1\end{matrix}\right)(1,1) =\left(
\begin{matrix}
  1& 1\\
  1 & 1
\end{matrix}
\right),\\  \left(\begin{matrix} i\\
    1\end{matrix}\right)(-i,1) = \left(
\begin{matrix}
  1 & i\\
  -i & 1
\end{matrix}
\right)\, ,
\end{align*}
is a basis of ${\mathfrak X}_j$. $ ~ \Box$

\subsection{The loss of coherence}
\label{sect:LC}
Lemma \ref{lem:MId} and equation (\ref{eq:E2}) give that 
\begin{align*}
\partial_Z \|\left< {\bf A}_j \right> \hspace{-0.03in} (Z)\|^2 = -
\left< {\bf A}_j \right>^\star \hspace{-0.03in} (Z) {\bf C}_j \left<
       {\bf A}_j \right> \hspace{-0.03in} (Z), \quad Z > 0, \qquad 
\|\left< {\bf A}_j \right> \hspace{-0.03in} (0)\|^2 = \|{\bf A}_{j,o}\|^2,
\end{align*}
where $\|\cdot\|$ is the Euclidian norm, and we recall that ${\bf
  C}_j$ is Hermitian, positive definite.  The result stated in the
next theorem follows from Gronwall's lemma:
\begin{theorem}
\label{thm:1} Let $\mu_{j,q}>0$ be the eigenvalues of ${\bf C}_j$ in increasing order, 
for all $ j = 1, \ldots, N$ and $1 \le q \le \mM_j$. We have that
\begin{equation}
e^{- \mu_{j,2} Z} \|{\bf A}_{j,o}\|^2 \le  \|\left< {\bf A}_j
\right> \hspace{-0.03in} (Z)\|^2 \le  e^{- \mu_{j,1} Z} \|{\bf A}_{j,o}\|^2, \qquad 
{\rm if} ~ \mM_j = 2,
\label{eq:LC}
\end{equation}
and 
\begin{equation}
\|\left< {\bf A}_j
\right> \hspace{-0.03in} (Z)\|^2 =  e^{- \mu_{j,1} Z} \|{\bf A}_{j,o}\|^2, \qquad 
{\rm if} ~ \mM_j = 1.
\label{eq:LCp}
\end{equation}
Thus, the mean amplitudes decay exponentially with $Z$, on mode
dependent range scales (scaled scattering mean free paths)
\begin{equation}
\cS_j = {1}/{\mu_{j,1}}\, .
\label{eq:SCMP}
\end{equation}
\end{theorem}

\begin{figure}[!t]
\begin{minipage}{1\textwidth}
\centering
\includegraphics[width=0.49\textwidth]{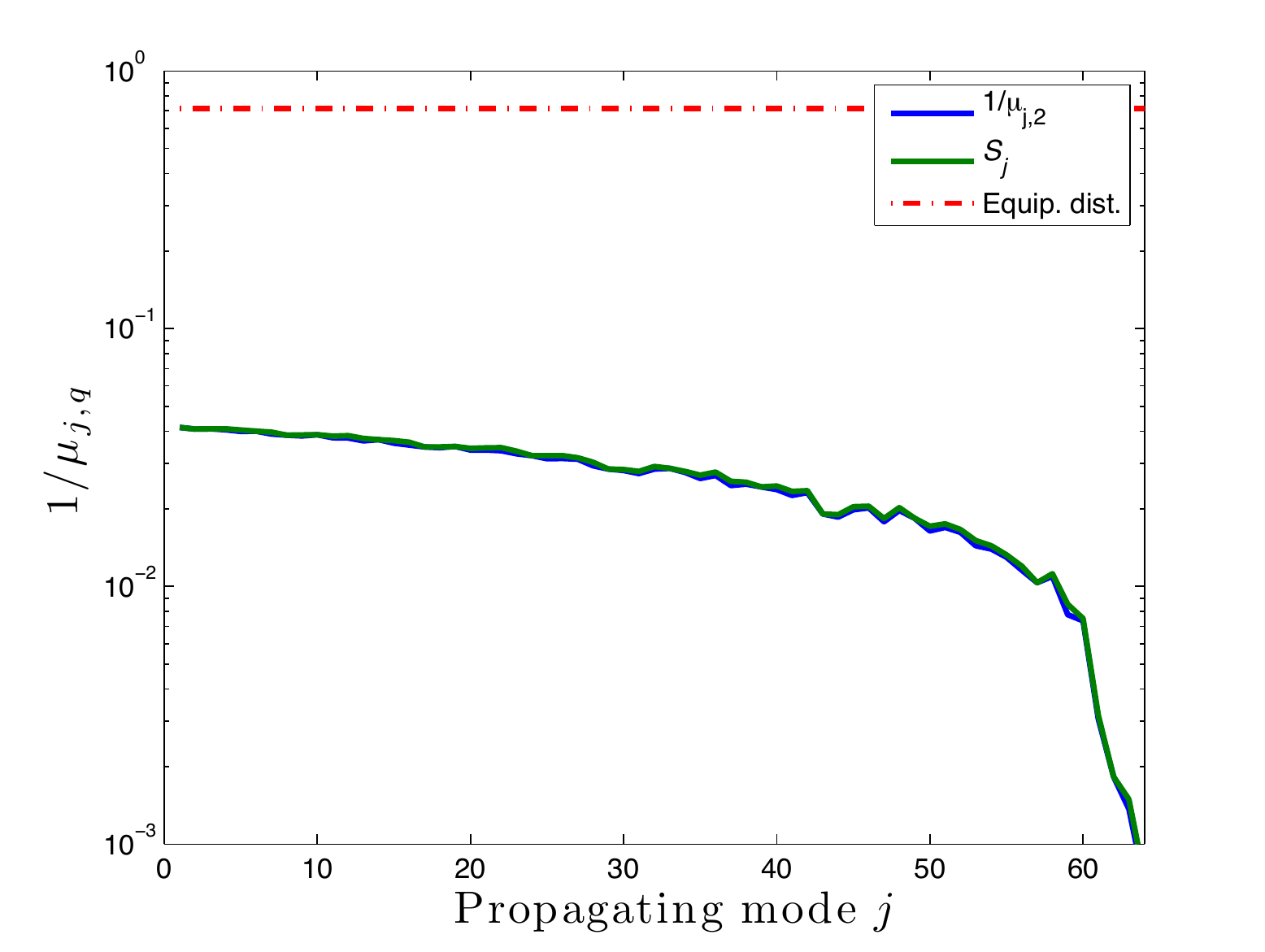}
\includegraphics[width=0.49\textwidth]{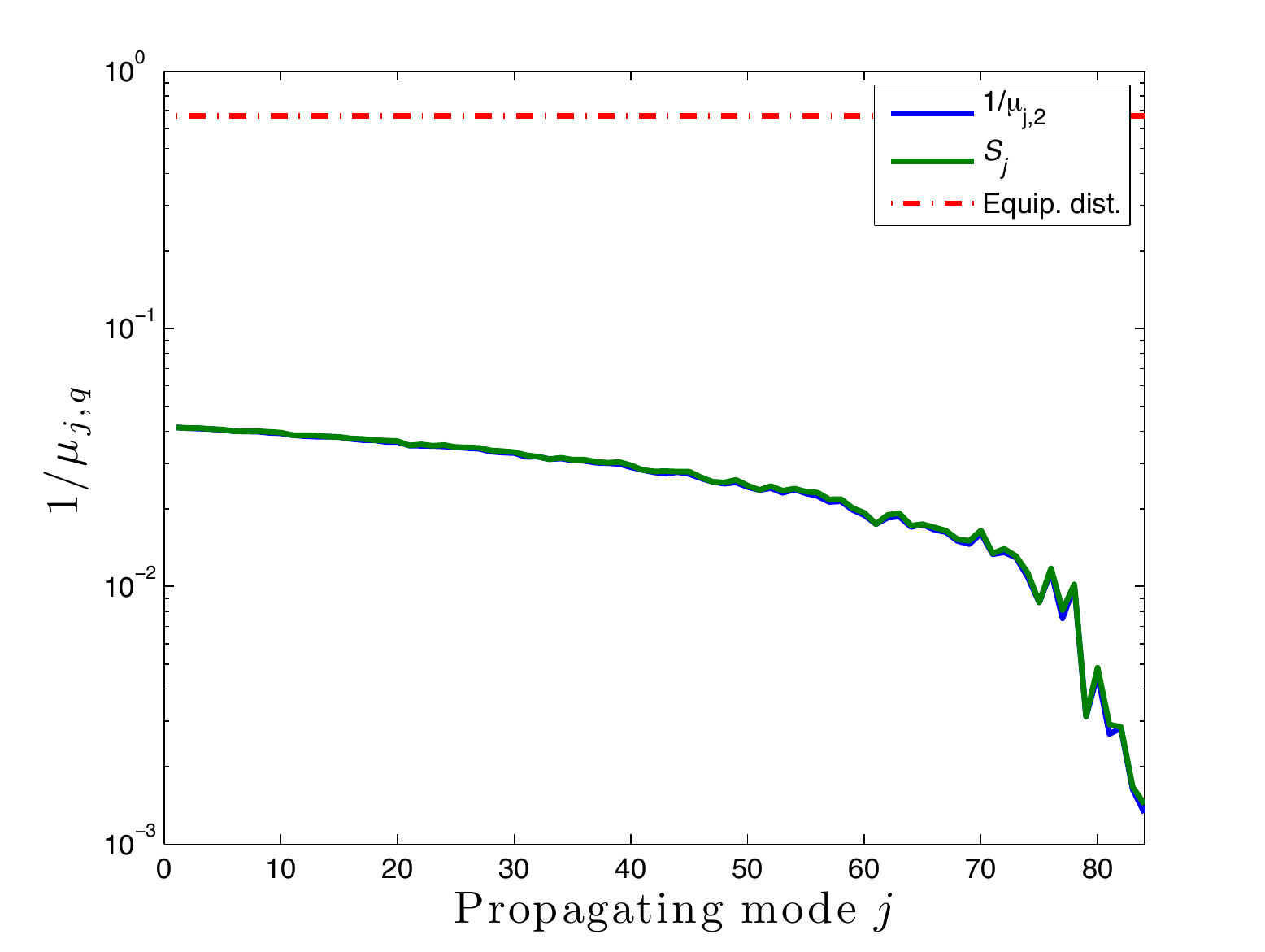}
\end{minipage}
\vspace{-0.1in}
\caption{We plot in green $\cS_j$, the reciprocal of the minimum
  eigenvalue of ${\bf C}_j$, and in blue the reciprocal of the maximum
  eigenvalue. The equipartition distance is shown in red.  We show
  results for two waveguides (from left to right) (1) $L_{1}=3.03$ and
  $L_{2} = 5.84$ giving $N = 64$, and (2) $L_{1}=4.08$ and
  $L_{2}=5.77$ giving $N = 84$. The abscissa is the mode index $j$ and
  the ordinate is in units of the wavelength $\lambda$.  }
\label{coherent}
\end{figure}

\noindent \textbf{Discussion and numerical illustration.} The decay of
the mean mode amplitudes is a manifestation of the loss of coherence
of the modes. This is a gradual process, with the last indexed modes
losing coherence faster than the first ones, as illustrated by the
numerical results displayed in Figure \ref{coherent}.  We plot $\cS_j
= 1/\mu_{j,1}$ in green and $1/\mu_{j,\mM_j}$ in blue. Note that
$\cS_j$ are the scaled scattering mean free paths, as follows from
(\ref{eq:DL1}). The actual scattering mean free paths are given by $
\cS_j^\epsilon = \cS_j/\epsilon^2$, and are much larger than the
wavelength $\lambda$.  The matrix ${\bf C}_j$ is computed as in
(\ref{eq:MP6}), using the coefficients defined in appendix \ref{ap:1},
for an isotropic random medium that is stationary in $x_1, x_2$ and
$z$, with covariance
\[
\EE[\nu(\vx) \nu(\vx')] = {\rm exp}\left( - \frac{|\vx-\vx'|^2}{2
  \ell^2}\right), \qquad \ell = \lambda.
\]
The left plot is for a waveguide with dimensions $L_1 = 3.03 \lambda$
and $L_2 = 5.84 \lambda$, so that $N = 64$.  In the right plot $L_1 =
4.08\lambda$ and $L_2 = 5.77 \lambda$, so that $N = 84$.

Note that for any $j$ the eigenvalues $\mu_{j,s}$ are almost the same
for $1\ll s \ll \mM_j$, indicating that the equality in (\ref{eq:LC})
holds independent of the multiplicity $\mM_j$. Moreover, $\cS_j$
decreases with $j$, and the rate of decrease acellerates for $j$ close
to $N$. The scale shown with red in Figure \ref{coherent} is the
equipartition distance, up to the $\ep^{-2}$ factor. This is the range
where cumulative scattering by the medium distributes the energy of
the waves uniformly over the modes, independent of their initial
state.  We give more details in the next section, but it is important
to note that the equipartition distance is larger, by a factor of ten,
than all the scattering mean free paths. This is very similar to the
result obtained for sound waves in random waveguides with straight
boundaries \cite[Figure 4.2]{QUANT-13}.

To interpret the results, let us note that the modes $\bphi_j^{(s)}
{\rm exp}(i \beta_j z)$, for $1 \le s \le \mM_j$, are superpositions
(component-wise) of the plane waves $ {\rm exp}( i \vec{\bf K}_j \cdot
\vx)$, with wave vectors
\[
\vec{\bf K}_j = \left(\pm {\pi j_1}/{L_1},\pm {\pi j_2}/{L_2},
  \beta_j\right),  \qquad j = (j_1,j_2).
\]
The plus and minus signs are due to the reflections of the waves at
the walls of the waveguide.  Recall that $\beta_j = \sqrt{k^2
  -\lambda_j}$, with eigenvalues $\lambda_j$ defined by (\ref{eq:P10})
and enumerated in increasing order. When $j$ is small, the wave vector
$\vec{\bf K}_j$ is almost aligned with the range axis, and the waves
propagate with large (group) range velocity
\[
{1}/{\beta'_j(\om)} = c_o \sqrt{1 - {\lambda_j}/{k^2}} \approx
c_o.
\]
They arrive quickly to range $Z$ because they travel along shorter
paths, with a small number of reflections at the walls, and are least
affected by the random medium. For the high index modes $\lambda_j
\approx k^2$, and the wave vectors $\vec{\bf K}_j$ are almost
orthogonal to the range axis.  The waves propagate very slowly along
range because they strike the waveguide walls many times. The
interaction with the random medium accumulates over the long travel
paths of these modes, and the waves lose coherence over shorter range
scales, as modeled by the small scattering mean free paths.

For any given $j$ the modes $\bphi_j^{(s)} {\rm exp}(i \beta_j z)$ are
the superposition of the same plane waves for $s = 1$ and $2$, so
their interaction with the medium is the same. This is why the
eigenvalues $\mu_{j,1}$ and $\mu_{j,2}$ are almost equal.

\subsection{The equipartition regime}
\label{sect:TR}
The transport of energy in the waveguides is modeled by the evolution
equations (\ref{eq:MP2}).  Our goal in this section is to describe
their solution in the limit $Z \to \infty$. 

 We begin by writing equations \eqref{eq:MP2} as
\begin{equation}\label{eq:TE1}
\partial_Z {\bf P}(Z) = \Upsilon\big(\textbf{P} \big)(Z) =
\Upsilon^{+}\big(\textbf{P} \big)(Z) - \Upsilon^{-}\big(\textbf{P}
\big)(Z),\quad Z>0,
\end{equation}
with initial condition ${\bf P}(0) = {\bf P}_o$. Here $\Upsilon,
\Upsilon^{\pm}:\mathfrak{X}\rightarrow\mathfrak{X}$ are linear
operators acting on the vector space $\mathfrak{X}$ defined in section
\ref{sect:CE}, with values in $\mathfrak{X}$. We have $\Upsilon=
\Upsilon^+ - \Upsilon^{-}$ and 
\begin{equation}\label{eq:TE2}
\Upsilon^{+}\big(\textbf{P} \big)_{j}(Z) =
\sum^{N}_{l=1}\Upsilon^{+}_{jl}\big(\textbf{P}_{l}\big)(Z)\,, \hspace{.4cm}
\Upsilon^{-}\big(\textbf{P} \big)_{j}(Z) =\sum^{N}_{l=1}
\Upsilon^{-}_{jl}\big({\bf P}_j \big)(Z),
\end{equation}
with operators
$\Upsilon^{\pm}_{jl}:\mathfrak{X}_{l}\rightarrow\mathfrak{X}_{j}$
acting on the spaces $\mathfrak{X}_l$ of Hermitian matrices defined in
section \ref{sect:CE}, and given by
\begin{align}\label{eq:TE3}
\Upsilon^{+}_{jl}\big(U\big)& = \int_{-\infty}^ \infty dz \EE\left\{
M_{AA,jl}(z)\,U\, M_{AA,jl}^\star(0)\right\} e^{i
  (\beta_l-\beta_j)z}\,,\nonumber\\ \Upsilon^{-}_{jl}\big( U \big) & =
-\Big( {\bf Q}_l\,U + U\, {\bf Q}_l^\star \Big) \delta_{jl}\,.
\end{align}
 We may think of $\Upsilon^{+}_{jl}$ and
$\Upsilon^{-}_{jl}$ as modeling the inflow/outflow of energy of the $j
\leftrightarrows l$ modes, because
\begin{equation}
\label{eq:TE3p}
\Upsilon^{+}_{jl}(U)\geq 0\,\quad\text{and}\quad
\text{trace}\big(\Upsilon^{-}_{jl}(U)\big)\geq0\,,
\end{equation}
for $1 \leq l, j \leq N,$ and all $U \in \mathfrak{C}_{l}$, the cone
of positive semidefinite matrices in $\mathfrak{X}_l$.
The limit of ${\bf P}(Z)$ as $Z \to \infty$ depends on the spectrum of
the operator $\Upsilon$, and in particular its kernel, described in the 
next theorem.
\begin{theorem}\label{thm:2}
The operator $\Upsilon$ has the following spectral properties: \\ (i)
The eigenvalues of $\Upsilon$ lie in $(-\infty,0]$. \\ (ii) The
  $\text{Kernel}\big(\Upsilon\big)$ is not trivial, it has an
  eigenbase, and it intersects the cone $\mathfrak{C} = \mathfrak{C}_1
  \times \mathfrak{C}_2 \times \ldots \mathfrak{C}_N \subset
  \mathfrak{X}$. \\ (iii) The $\text{Kernel}\big(\Upsilon\big)$ is one
  dimensional under the additional assumption that
\begin{equation}\label{eq:conirr}
\Upsilon^{+}_{jl}\big(U\big) > 0\,,\quad \forall\,0\neq
U\in\mathfrak{C}_{l} ~~ {\rm and} ~~  1 \leq j, l \leq N.
\end{equation}
\end{theorem}

\vspace{0.1in} For any initial condition ${\bf P}_o \in \mathfrak{C}$
we have ${\bf P}(Z) \in \mathfrak{C}$ for all $Z$, as shown in
appendix \ref{ap:cone}. This is why we are interested in the cone
$\mathfrak{C}$ of the space $\mathfrak{X}$. The assumption
\eqref{eq:conirr} says that there is positive flux of energy for all
the waveguide modes. It is the generalization of the condition stated
in \cite[Section 20.3.3]{LAY_BOOK} which gives the equipartition
regime for sound waves. The statement there is that the power spectral
density of the fluctuations of the wave speed does not vanish when
evaluated at the differences of the wavenumbers of the modes. Our
condition (\ref{eq:conirr}) is similar, but with $\mM_j \times \mM_j$
matrices.  The following corollary follows immediately from parts (i)
and (iii) of Theorem \ref{thm:2}.
\begin{corollary}\label{cor:1}
Suppose that condition \eqref{eq:conirr} holds, and let $\textbf{U}_o$
be the unique vector that spans
$\text{Kernel}\big(\Upsilon\big)\cap\mathfrak{C}$, normalized by
$\|\textbf{U}_o\| = \|\textbf{P}_o\|$. We have
\begin{equation*}
\big| \textbf{P}(Z) - \textbf{U}_o \big| \leq
C\,(1+Z)^{\mathfrak{m}_{_\Upsilon}}\, e^{\lambda_{_\Upsilon}\,Z},
\end{equation*}
where $\lambda_{_\Upsilon}$ is the smallest (in magnitude) nonzero
eigenvalue of $\Upsilon$, and $\mathfrak{m}_{_\Upsilon}$ is its
multiplicity.  
\end{corollary}

We display in Figure \ref{stationary} the element ${\bf U}_o \in
\text{Kernel}\big(\Upsilon\big)$ for the same two random waveguides
considered in Figure \ref{coherent}.  We normalize ${\bf U}_o$ so that
its maximum entry is equal to one. Because ${\bf U}_o$ is a
concatenation of $N$ matrices of size $\mM_j \times \mM_j$, we embed
it in a square matrix for display purposes. The entries of interest in
the square matrices displayed in Figure \ref{stationary} are the
$\mM_j \times \mM_j$ blocks along the diagonal. We note from the
figure that the result is almost the matrix identity. Therefore,
Corollary \ref{cor:1} says that in the limit $Z \to \infty$, the
energy is distributed uniformly over the modes, independent of the
initial mode power distribution ${\bf P}_o$. This is the
\emph{equipartition regime}, and it is reached when the waves travel
beyond the \emph{equipartition distance}
\begin{equation}
L_{\rm eq} = 1/|\lambda_{_{\Upsilon}}|.
\label{eq:EQUIPL}
\end{equation}

\begin{figure}[!t]
\begin{minipage}{1\textwidth}
\centering \includegraphics[width=0.45\textwidth]{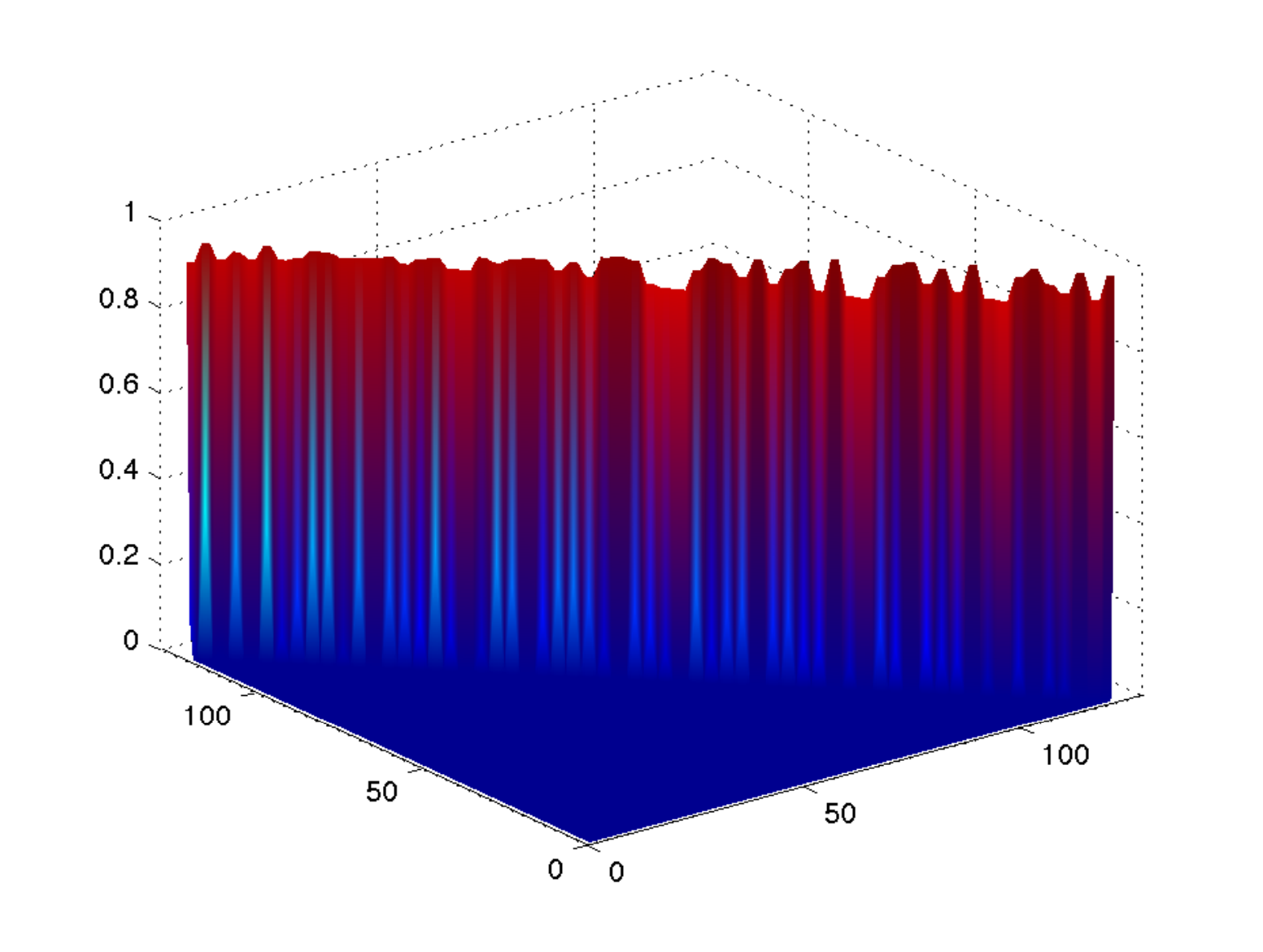}
\includegraphics[width=0.45\textwidth]{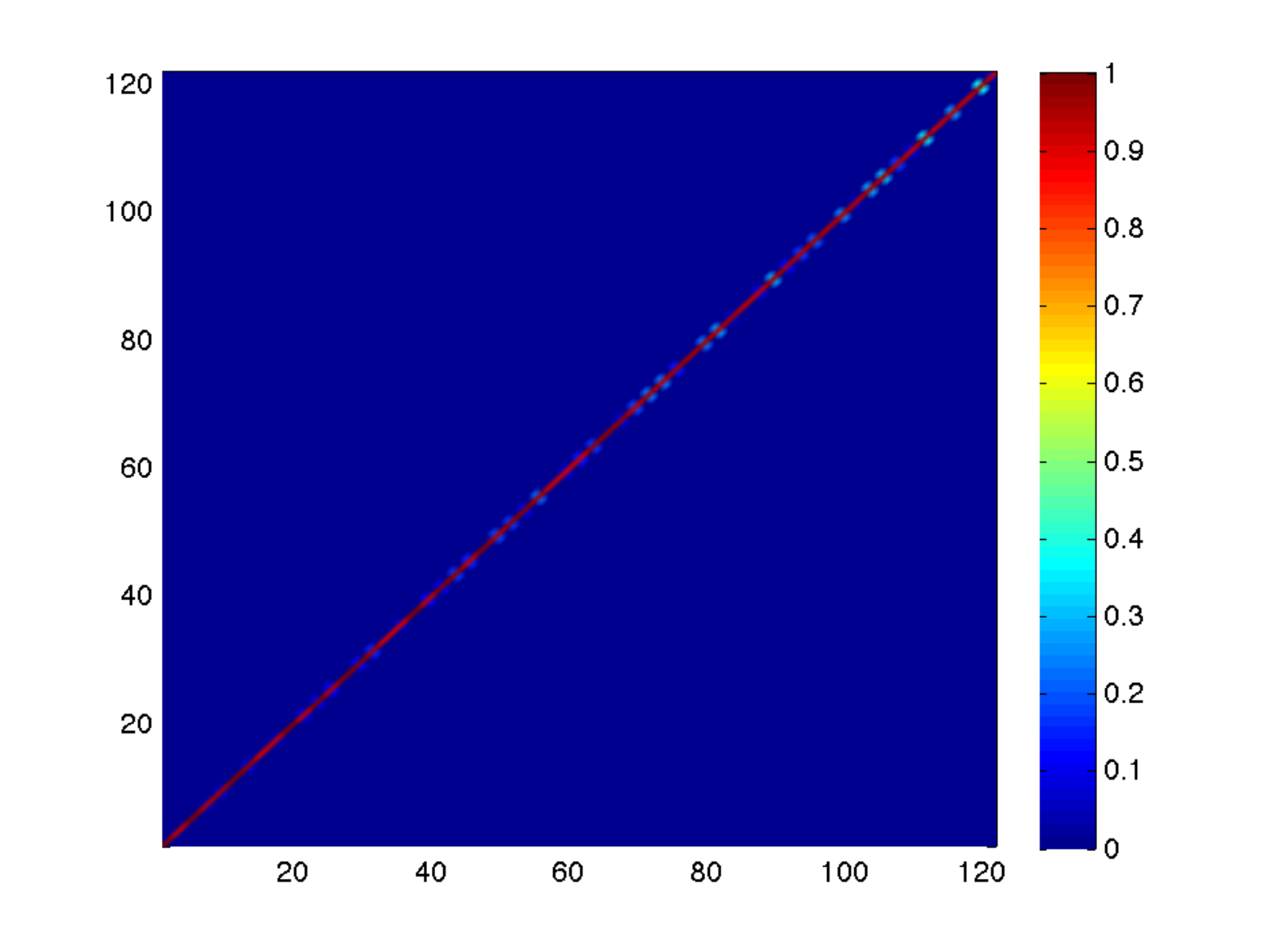}
\end{minipage}
\begin{minipage}{1\textwidth}
\centering
\includegraphics[width=0.45\textwidth]{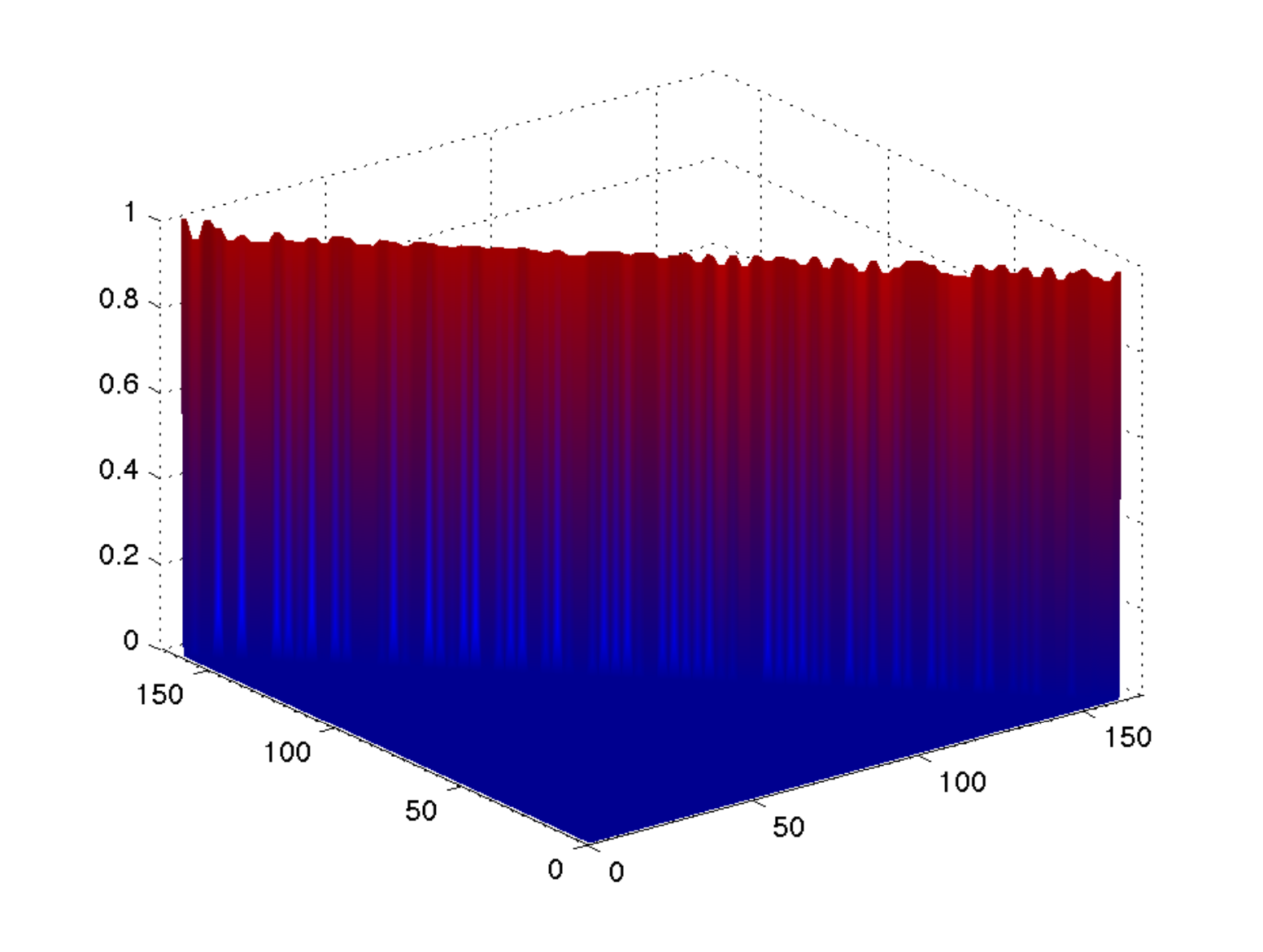}
\includegraphics[width=0.45\textwidth]{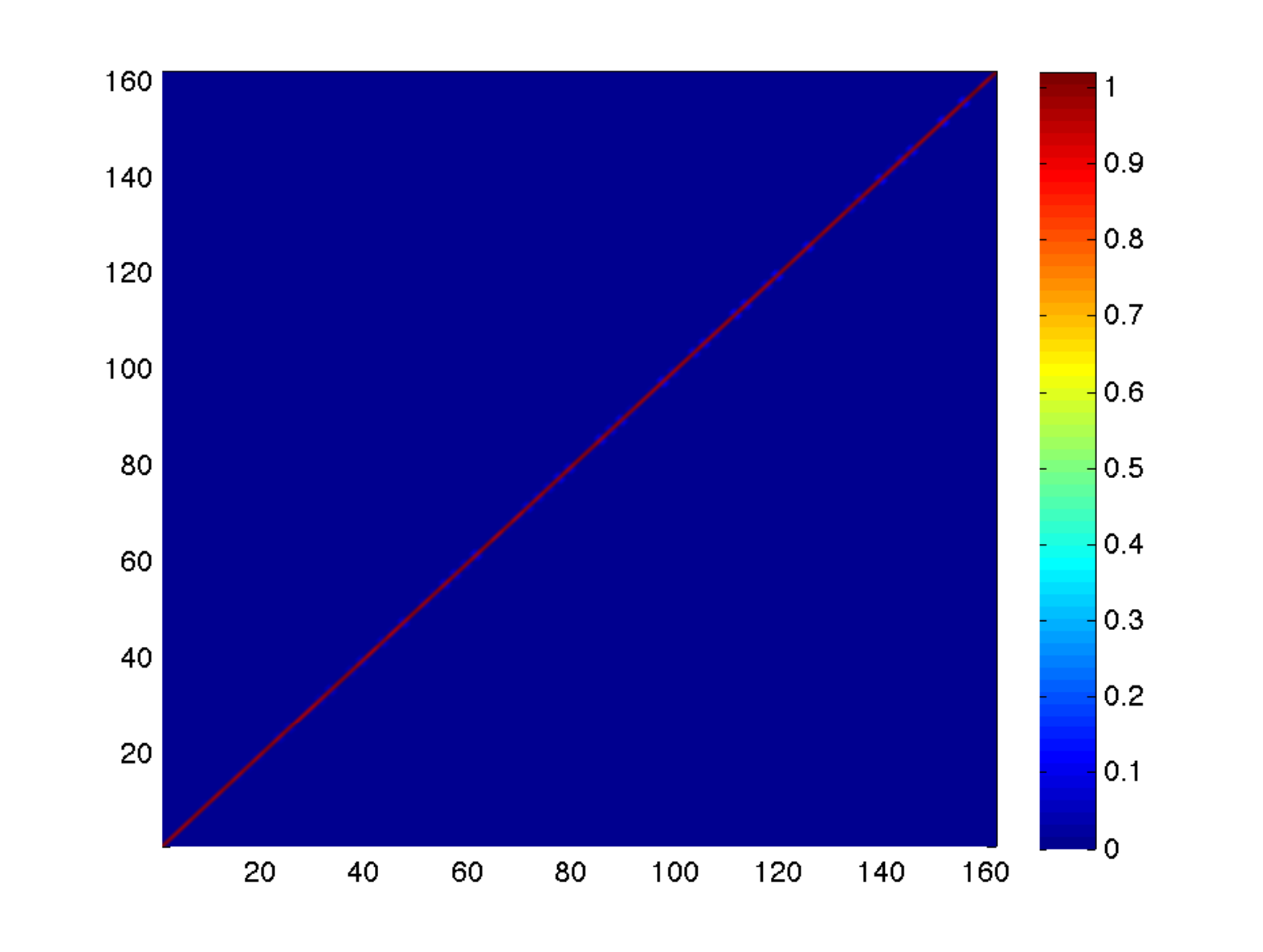}
\end{minipage}
\vspace{-0.1in}
\caption{The element (matrix) ${\bf U}_o \in
  \text{Kernel}\big(\Upsilon\big)$ for the same two waveguides
  considered in Figure \ref{coherent}. The top plots are for the
  waveguide with $N = 64$ modes and the bottom plots for the waveguide
  with $N = 84$ modes.  We display ${\bf U}_o$ from the lateral side
  (left) to show the magnitude of its entries (left), and from above
  (right) to show that it is almost diagonal. }
\label{stationary}
\end{figure}

\noindent \textbf{Proof of Theorem \ref{thm:2}:}

\textit{Item (i).} Let $0\neq\lambda\in\mathbb{C}$ be an eigenvalue of
$\Upsilon$ and $\textbf{U}\in\mathfrak{X}$ an associated eigenvector.
Therefore, $\Upsilon\big(\textbf{U}\big)=\lambda\,\textbf{U}$, or
componentwise,
\begin{equation*}
\lambda\,U_{j}=\Upsilon\big(\textbf{U}\big)_{j}=\Upsilon\big(\textbf{U}\big)^{\star}_{j}=
\big(\lambda\, U_{j}\big)^{\star}=\overline{\lambda}\,U_{j},\quad
1\leq j \leq N,
\end{equation*}
where we used definitions (\ref{eq:TE2}-{\ref{eq:TE3}) to obtain the
  second equality.  Consequently, the eigenvalues of $\Upsilon$ are
  real valued.  To see that they cannot be positive, we use that
  $\textbf{P}(Z) \in \mathfrak{C}$ for all $Z$ as shown in appendix
  \ref{ap:cone}, and the conservation of energy
\[
\sum^{N}_{j=1}\text{trace}\big(\textbf{P}_{j}(Z)\big)=\sum^{N}_{j=1}\text{trace}
\big(\textbf{P}_{j,o}\big)\,,\quad Z\geq0.
\]
Since $\textbf{P}_{j}(Z)\in\mathfrak{C}_{j}$, we have $ 0 \leq
P^{ss}_{j}(Z) \leq \text{trace}\big(\textbf{P}_{j}(Z) \big), $ for all
$Z \ge 0$, $1\leq j \leq N$ and $1\leq s \leq
\mathfrak{M}_{j}$. Moreover
\begin{equation*}
\big| P^{12}_{j}(Z) \big| \leq \sqrt{P^{11}_{j}(Z)P^{22}_{j}(Z)}\leq
\tfrac{1}{2}\text{trace}\big(\textbf{P}_{j}(Z) \big),\quad Z\geq 0,
\end{equation*}
when $\mathfrak{M}_{j}=2$, where we used the Cauchy-Schwarz
inequality.  Therefore 
\[ \sup_{Z\geq0}\big|\textbf{P}(Z)\big| \leq
\sum^{N}_{j=1}\text{trace}(\textbf{P}_o), \qquad \forall \, {\bf P}_o
\in \mathfrak{C}.  \] Now consider an arbitrary initial state in
$\mathfrak{X}$, not necessarily in $\mathfrak{C}$. We denote such a
state by $\widetilde{{\bf P}}_o$ to distinguish it from a physical
initial mode power state that is necessarily in $\mathfrak{C}$, and
the corresponding solution of \eqref{eq:TE1} by $\widetilde{P}(Z)$.
Since any $\widetilde{\bf P}_o$ can be writen as a linear combination of
elements in $\mathfrak{C}$, and \eqref{eq:TE1} is linear,
$\widetilde{P}(Z)$ is a linear combination of solutions with initial
states in $\mathfrak{C}$, which are bounded as shown above. We
conclude that all trajectories are bounded, independent of the initial
state.

Now, let $\textbf{U}$ be an eigenvector of $\Upsilon$ for an
eigenvalue $\lambda$, and set $\widetilde{\bf P}_0 = \textbf{U}$.  The
solution of \eqref{eq:TE1} is  $\widetilde{\bf P}(Z)=
\textbf{U}\,e^{\lambda Z}$ and it is uniformly bounded if and only if
$\lambda\leq0$.

\textit{Item (ii).} To characterize $\text{Kernel}(\Upsilon)$ we
rewrite the conservation of energy as
\begin{equation}\label{eq:TE4}
\big( \Upsilon(\textbf{U}), \textbf{1} \big)_{\mathfrak{X}} =\big(
\textbf{U}, \Upsilon^{\star}(\textbf{1}) \big)_{\mathfrak{X}} =
0\,,\quad \forall\,\textbf{U}\in\mathfrak{X},
\end{equation}
where $\textbf{1}$ is the vector of concatenated $\mM_j \times \mM_j$
identity matrices, for $j = 1, \ldots, N$, and $\Upsilon^\star$ is the
adjoint of $\Upsilon$.  We conclude that $\textbf{1} \in
\text{Kernel}\big(\Upsilon^{*}\big)$, and therefore that the kernel of
$\Upsilon$ is not trivial.

We also infer from the boundedness of the solutions of \eqref{eq:TE1}
that $\text{Kernel}\big(\Upsilon\big)$ has an eigenbase. Otherwise the
solutions could grow polynomially in $Z$.  

Now let us write the solutions of \eqref{eq:TE1} as
\begin{equation}\label{eq:TE5}
\textbf{P}(Z) = \textbf{U}_{o} + \sum^{\mathfrak{n}}_{j=1}
\sum^{\mathfrak{m}_{j}}_{q=0} c_{jq}\,\textbf{U}_{jq}\; (Z+1)^{q}
e^{\lambda_{j}Z}, 
\end{equation}
where $\textbf{U}_{o}\in\text{Kernel}\big(\Upsilon\big)$, and $n$ is
the number of distinct eigenvalues $\lambda_j < 0$ of $\Upsilon$, with
multiplicity $m_j$.  The finite sequences
$\{\textbf{U}_{jq}\}^{\mathfrak{m}_{j}}_{q=1}$ are linear independent
sets of generalized eigenvectors corresponding to the eigenvalue
$\lambda_{j}$.  Equation \eqref{eq:TE5} implies 
\begin{equation}\label{eq:TE6}
\big| \textbf{P}(Z) - \textbf{U}_o \big| \leq
C\,(1+Z)^{\mathfrak{m}_{_\Upsilon}}\, e^{\lambda_{_\Upsilon}\,Z},
\end{equation}
where we denote by $\lambda_{_\Upsilon}$ the eigenvalue with smallest
magnitude and $m_{_\Upsilon}$ its multiplicity.  Now, take $0 \neq
\textbf{P}_{o} \in\mathfrak{C}$, and therefore,
$\textbf{P}(Z)\in\mathfrak{C}$ for all $Z\geq0$.  Letting
$Z\rightarrow\infty$ in \eqref{eq:TE6} we conclude that
$\textbf{U}_{o}\in\mathfrak{C}$ i.e.,
$\textbf{U}_o\in\text{Kernel}\big(\Upsilon\big)\cap\mathfrak{C}$.
Moreover, $\textbf{U}_o \ne 0$ by the conservation of energy.

\textit{Item (iii).} We prove that condition (\ref{eq:conirr})
guarantees a one dimensional kernel of the adjoint operator
$\Upsilon^\star$, and therefore of $\Upsilon$.  We consider first 
a large family of linear mappings defined in terms of $\Upsilon^{*}$
and restricted to  the subspace
$\mathfrak{D}\subset\mathfrak{X}$ of vectors of diagonal matrices. We
show that this family has the one dimensional kernel ${\rm
  span}\{{\bf 1}\}$ in $\mathfrak{D}$. Then we show that ${\rm Kernel}(\Upsilon^\star)
= {\rm span}\{{\bf 1}\}$ in $\mathfrak{X}$.

To define the family of mappings, we recall that the
elements ${\bf U} = (U_1, \ldots, U_N)$ of $\mathfrak{X}$ are vectors
of Hermitian matrices which are diagonalized by similarity
transformations with orthogonal matrices of their eigenvectors.  For
any vector $\textbf{V} = (V_1, \ldots, V_N)$ of orthogonal matrices
we define the transformation
$\varphi_{\textbf{V}}:\mathfrak{X}\rightarrow\mathfrak{X}$ by
\begin{equation*}
\varphi_{\textbf{V}}\big(\textbf{U}\big) =
\textbf{V}^{\star}\,\textbf{U}\,\textbf{V} :=
\big(V^{\star}_{1}U_{1}V_{1},\cdots,V^{\star}_{N} U_{N}V_{N}\big),
\end{equation*}
with dimensions of $U_{j}$ and $V_{j}$ assumed to match for each
$1\leq j\leq N$. We also define the family of operators
$\Upsilon^\star_{\textbf{V}}:\mathfrak{X}\rightarrow\mathfrak{X}$ by
$\Upsilon^\star_{\textbf{V}}= \varphi_{\textbf{V}}^{-1}\circ
\Upsilon^{\star} \circ \varphi_{\textbf{V}}$, or more explicitly,
\begin{equation}\label{eq:TE7}
\Upsilon^\star_{\textbf{V}}\big(\textbf{U}\big) = \textbf{V}\,
\Upsilon^{\star}\big(\textbf{V}^{\star}\,\textbf{U}\,\textbf{V}
\big)\, \textbf{V}^{\star}\,,\quad \textbf{U}\in\mathfrak{X}.
\end{equation} 
Its restriction to the subspace $\mathfrak{D}\subset\mathfrak{X}$ of
vectors of diagonal matrices is denoted by
\begin{equation*}
\Upsilon^\star_{\textbf{V}\big|\mathfrak{D}}\big(\textbf{D}\big)=
\Upsilon^\star_{\textbf{V}}\big(\textbf{D}\big)\,\quad\text{for
  any}\;\; \textbf{D}\in\mathfrak{D},
\end{equation*}
and we wish to prove that 
\begin{equation}\label{eq:TE7-1}
{\rm Kernel}\big(\Upsilon^\star_{\textbf{V}\big|\mathfrak{D}}\big)
= {\rm span}\{{\bf
  1}\}, \qquad \forall \, {\bf V}.
\end{equation} 
Note that $\Upsilon^\star_{\textbf{V}\big|\mathfrak{D}}(\textbf{1}) = 0$ by 
the definition \eqref{eq:TE7} and $\Upsilon^{\star}({\bf 1}) = 0$. 

The statement of the theorem is implied by (\ref{eq:TE7-1}).  Indeed,
take an arbitrary
$\textbf{U}\in\text{Kernel}\big(\Upsilon^{\star}\big)$.  Since
$\mathbf{U}$ is a vector of Hermitian matrices, there exists a vector
$\textbf{V}$ of orthogonal matrices and a vector $\textbf{D}$ of
diagonal matrices such that $\textbf{U}=\textbf{V}^{*}\,\textbf{D}
\textbf{V}$.  Then
$\textbf{D}\in\text{Kernel}\big(\Upsilon^\star_{\textbf{V}\big|\mathfrak{D}}\big)$
and \eqref{eq:TE7-1} implies that $\textbf{D}=\alpha\,\mathbf{1}$ for
some $\alpha\in\mathbb{R}$. Consequently,
\[
\textbf{U}={\bf V}^\star \alpha {\bf 1} {\bf V} = \alpha\,\mathbf{1},
\]
which means that ${\rm Kernel}(\Upsilon^\star) = {\rm span}\{{\bf
  1}\}$.

The proof of \eqref{eq:TE7-1} is based on the Perron--Frobenius
theorem for irreducible matrices.  We start by computing the matrix
representation of the operator $\Upsilon^\star_{\textbf{V}}$ using the
properties of the adjoint operators $\Upsilon^{\pm
  \,\star}_{jl}:\mathfrak{X}_{j}\rightarrow\mathfrak{X}_{l}$ that
define $\Upsilon^{\star}$.  Then, we extract a suitable irreducible
matrix $\Lambda_{\textbf{V}}$ from it to apply the Perron--Frobenius theorem.

The matrix representation of $\Upsilon^\star _{\textbf{V}}$ consists
of $N^{2}$ blocks, where the $lj$--block is the matrix representation
of the operator $\Upsilon^{+ \star}_{\textbf{V} lj} - \Upsilon^{-
  \star}_{\textbf{V} lj}$, with $\Upsilon^{\pm \star}_{\textbf{V}
  lj}:\mathfrak{X}_{j}\rightarrow\mathfrak{X}_{l}$ defined naturally,
in light of \eqref{eq:TE7}, as
\begin{equation*}
\Upsilon^{\pm \star}_{\textbf{V} lj}(U) = V_{l}\,
\Upsilon^{\pm\,\star}_{jl}\big( V^{\star}_{j}\, U\, V_{j} \big)\,
V^{\star}_{l}\,,\quad U\in\mathfrak{X}_{j}.
\end{equation*}
Recall that the dimension of $\mathfrak{X}_{j}$ is
$\mathfrak{M}^{2}_{j}$, so the $lj$--block has dimension
$\mathfrak{M}^{2}_{l}\times \mathfrak{M}^{2}_{j}$. To be more precise,
consider the case $\mathfrak{M}_{j} = \mathfrak{M}_{l} = 2$, and write
explicitly
\[
\Upsilon^{\pm \star}_{\textbf{V} lj} = \Big(\upsilon^{\pm\,ss'}_{lj}
\Big)\,,\quad 1\leq s \leq \mathfrak{M}^{2}_{l}\,,\quad 1\leq s' \leq
\mathfrak{M}^{2}_{j},
\]
In terms of the canonical basis for $\mathfrak{X}_{j}$
\begin{equation*}
\left\{
\Big(\begin{array}{cc}
1 & 0 \\
0 & 0
\end{array}\Big),
\Big(\begin{array}{cc}
0 & 0 \\
0 & 1
\end{array}\Big),
\tfrac{1}{\sqrt{2}}\Big(\begin{array}{cc}
0 & 1 \\
1 & 0
\end{array}\Big),
\tfrac{1}{\sqrt{2}}\Big(\begin{array}{cc}
0 & i \\
-i & 0
\end{array}\Big)
 \right\} =: \big\{E_{1},E_{2},E_{3},E_{4}\big\},
\end{equation*}
we have
\begin{align}\label{eq:TE9}
\upsilon^{+\,ss'}_{lj}&=\big( E_{s},\Upsilon^{+
  \star}_{\textbf{V}{lj}}(E_{s'})\big)_{\mathfrak{X}_{l}}
\nonumber\\ &=\big( E_{s},V_{l}\, \Upsilon^{+\,\star}_{jl}\big(
V^{\star}_{j}\, E_{s'}\, V_{j} \big)\, V^{\star}_{l}
\big)_{\mathfrak{X}_{l}}
\nonumber\\ &=\big(V^{\star}_{l}\,E_{s}\,V_{l},
\Upsilon^{+\,\star}_{jl}\big( V^{\star}_{j}\, E_{s'}\,V_{j}\big)
\big)_{\mathfrak{X}_{l}}.
\end{align}
Note that $E_{1}$ and $E_{2}$ are positive semidefinite matrices, and
so are $V^{\star}_{l}\,E_1\,V_{l}$ and $V^{\star}_{j}\, E_2\,V_{j}$.
It follows from the explicit expression of $
\Upsilon^{+\,\star}_{jl}$ computed from \eqref{eq:TE3} as
\begin{equation}
\label{eq:defUpStar}
\Upsilon^{{+\,\star}}_{jl}(U) = \int_{-\infty}^{\infty} dz \, \EE\left\{ M_{{AA,jl}}^{\star}(z) U 
M_{AA,jl}(0)\right\}e^{-i(\beta_{l}-\beta_{j})z},
\end{equation}
 and
equation \eqref{eq:TE9}, that
\begin{equation}\label{eq:TE10}
\upsilon^{+\,ss'}_{lj} \geq 0\,,\quad 1\leq s,s'\leq 2.
\end{equation}
The remaining entries of $\Upsilon^{+ \star}_{\textbf{V} lj}$ do not
have a definite sign, in general.  The entries of $\Upsilon^{-
  \star}_{\textbf{V} lj}$ are computed similarly, and take the form
\begin{align}\label{eq:TE11}
\upsilon^{-\,ss'}_{lj} &=\big(
E_{s},\Upsilon^{-\star}_{\textbf{V}{lj}}(E_{s'})
\big)_{\mathfrak{X}_{l}} \nonumber\\ &=\big( E_{s}, V_{l}\,
\Upsilon^{-\,\star}_{jl} \big(V^{\star}_{j}\, E_{s'}\, V_{j} \big)\,
V^{\star}_{l} \big)_{\mathfrak{X}_{l}} \nonumber\\ &=\big(
E_{s},\big(V_{l}\,\textbf{C}_{l}\, V^{\star}_{l}\big)\,
E_{s}\big)_{\mathfrak{X}_{l}}\,\delta_{ss'}\,\delta_{jl}\geq 0\,,\quad
1\leq s,s'\leq 2.
\end{align}
We used that $E_{s}E_{s'}=E_{s}\delta_{ss'}$ for $s, s' \in\{1,2\}$
for the third equality, and properties of positive definite matrices
in the last inequality.  In summary, the $lj$--block is of the form
\begin{equation}
\label{eq:TE11C}
 \Big(\upsilon^{+\,ss'}_{lj} \Big) - \Big(\upsilon^{-\,ss'}_{lj} \Big)
 = \Bigg(\begin{array}{cc} \boldsymbol\upsilon^{+}_{lj} & \cdot
   \\ \cdot & \cdot
\end{array}\Bigg)-
\Bigg(\begin{array}{cc}
\boldsymbol\upsilon^{-}_{ll} & \cdot \\
\cdot & \cdot
\end{array}\Bigg)\,\delta_{lj},
\end{equation}
with $2\times2$ matrices $\boldsymbol\upsilon^{\pm}_{lj}$ with
nonnegative entries, and nonnegative diagonal matrices
$\boldsymbol\upsilon^{-}_{ll}$. Cases where $\mM_j$ or $\mM_l$ or both
are equal to one are handled similarly.

Now, define the $\mM \times \mM$ square matrix
\begin{equation}
\Lambda_{\textbf{V}}=\text{Block}\big(\boldsymbol\upsilon^{+}_{lj}\big)-
\text{Block}\big(\boldsymbol\upsilon^{-}_{lj}\big),
\end{equation}
where $\mM = \displaystyle \sum_{j=1}^N \mathfrak{M}_{j}$, and observe
that the dimension of its kernel satisfies
\begin{equation}\label{eq:TE12}
1\leq\text{Dim}\big\{\text{Kernel}
\big(\Upsilon^\star_{\textbf{V}\big|\mathfrak{D}}\big)\big\}\leq
\text{Dim}\big\{\text{Kernel}\big(\Lambda_{\textbf{V}}\big)\big\}.
\end{equation}
This is because for any
$\textbf{D}\in\text{Kernel}\big(\Upsilon^\star_{\textbf{V}\big|\mathfrak{D}}\big)$,
by construction of $\Lambda_{\textbf{V}}$, the $\mM$ vector formed by
the diagonal entries of $\textbf{D}$ lies in
$\text{Kernel}\big(\Lambda_{\textbf{V}}\big)$.  In particular,
$\big(1,1,\cdots,1\big)\in\text{Kernel}\big(\Lambda_{\textbf{V}}\big)$ so the 
entries in each row of $\Lambda_{\textbf{V}}$ sum to zero. Then, 
Gershgorin's circle theorem and the special structure of $\Lambda_{\textbf{V}}$ 
give that any eigenvalue $\lambda_{\textbf{V}}$ of $\Lambda_{\textbf{V}}$  satisfies
\[
\lambda_{\textbf{V}}  \le v_{jj}^{+ ss}  - v_{jj}^{-} + v_{jj}^{+ ss'} + 
\sum_{l\ne j}\sum_{s''=1}^{\mM_{l}}v_{jl}^{+ss''} = 0,
\]
for all $1 \le s,s' \le \mM_{j}$ and  $s \ne s'$.  This shows that the largest 
eigenvalue of $\lambda_{\textbf{V}}$ is zero.

Consider now the matrix
\begin{equation*}
\Lambda_{\textbf{V}}+\alpha \textbf{I},
\end{equation*}
with $\alpha$ a sufficiently large positive real number such that the diagonal is positive.  
The spectrum of this matrix is a translation by $\alpha$ of the spectrum of $\Lambda_{\textbf{V}}$, 
therefore, by the previous discussion $\alpha$ must be its largest positive real eigenvalue. 
Assuming that
  $\text{Block}\big(\boldsymbol\upsilon^{+}_{lj}\big)$ is irreducible,
  $\Lambda_{\textbf{V}}+\alpha \textbf{I}$ is a Perron--Frobenius matrix implying that
  its largest positive real eigenvalue, in this case $\alpha$, is simple, and therefore, zero is a simple eigenvalue of $\Lambda_{\textbf{V}}$. 
  Thus, the kernel of $\Lambda_{\textbf{V}}$ is one dimensional and so is the kernel of 
 $\Upsilon^\star_{\textbf{V}\big|\mathfrak{D}}$,  by \eqref{eq:TE12}. 
This is precisely \eqref{eq:TE7-1}.

It remains to show  that the condition \eqref{eq:conirr}
 ensures that 
  $\text{Block}\big(\boldsymbol\upsilon^{+}_{lj}\big)$ is irreducible,
  regardless of the transformation $\varphi_\textbf{V}$. We already know that 
 this is a matrix with nonnegative entries, but we need to show that 
 each entry is positive.  From
\eqref{eq:TE9} we observe that this means that\footnote{For dimension one, it is
  understood that the basis is given by $\{E_{1}\}=\{1\}$.}
\begin{equation*}
\upsilon^{+,
  ss'}_{lj}=\text{trace}\Big(E_{s}\,V_{l}\,\Upsilon^{+\,\star}_{jl}\big(
V^{\star}_{j}\, E_{s'}\,V_{j}\big)\,V^{\star}_{l}\Big)> 0\,,\quad
1\leq s \leq \mathfrak{M}_{l},\;\; 1\leq s' \leq \mathfrak{M}_{j},
\end{equation*} 
which is equivalent to saying that the diagonal entries of the matrix
\begin{equation*}
V_{l}\,\Upsilon^{+\,\star}_{jl}\big( V^{\star}_{j}\,
E_{s'}\,V_{j}\big)\,V^{\star}_{l}\,,\quad 1 \leq s' \leq
\mathfrak{M}_{j},
\end{equation*}
are positive.  To see that this is true, observe that $0\neq V^{\star}_{j}\,
E_{s'}\,V_{j}\in\mathfrak{C}_{j}$ for $1 \leq s' \leq
\mathfrak{M}_{j}$,  and that condition \eqref{eq:conirr} and definition \eqref{eq:defUpStar} imply
\begin{equation*}
\Upsilon^{+\,\star}_{jl}\big( V^{\star}_{j}\, E_{s'}\,V_{j}\big) >
0\,,\quad 1 \leq s' \leq \mathfrak{M}_{j}.
\end{equation*}
The result follows immediately from this. $\Box$

\section{Summary}
\label{sect:summary}
We presented a rigorous analysis of electromagnetic wave propagation
in waveguides with rectangular cross-section. The dielectric materials
that fill the waveguides are lossless isotropic, and contain numerous
weak inhomogeneities (imperfections). Consequently, their electric
permittivity $\ep(\vx)$ has small fluctuations in $\vx$ that are
uncertain in applications, which is why we model them with a random
process. The main result of the paper is a detailed characterization
of long range cumulative scattering effects in random waveguides.

Our method of analysis decomposes the electromagnetic wave field in
transverse electric and magnetic modes, which are propagating and
evanescent waves. The modes are coupled by scattering in the random
medium, so their amplitudes are random processes. They satisfy a
stochastic system of equations driven by the random fluctuations of
the permittivity $\ep(\vx)$, and can be analyzed at long range using
the diffusion approximation theorem. The result is a detailed
characterization of the loss of coherence of the modes, the
depolarization of the waves due to scattering, and the transport of
energy by the modes. Loss of coherence means that the expectation of
the mode amplitudes (the coherent part) is overwhelmed by their random
fluctuations (the incoherent part) once the waves travel beyond
distances called scattering mean free paths.  These are range scales
that depend on the modes, the wavelength and the covariance of the
fluctuations of $\ep(\vx)$. Our analysis of long range transport of
energy shows how scattering in the random medium redistributes the
energy among the waveguide modes. In particular, it identifies a range
scale, called the equipartition distance, beyond which the energy is
uniformly distributed among the waveguide modes, independent of the
initial conditions.

Our results have applications in long range communications and imaging
in waveguides. See for example the imaging and time reversal studies
\cite{BIT-10,QUANT-13,G-1} that are based on the theory of sound wave
propagation in random waveguides developed in
\cite{kohler77,Dozier,GS-08,ABG-12,BG-13}. Here we extended the theory
to electromagnetic wave propagation in random waveguides.

\section*{Acknowledgements}
The work of R. Alonso was partially supported by the AFOSR Grant
FA9550-12-1-0117 and the ONR Grant N00014-12-1-0256.  The work of L.
Borcea was partially supported by the AFOSR Grant FA9550-12-1-0117,
the ONR Grant N00014-12-1-0256 and by the NSF Grants DMS-0907746,
DMS-0934594.

\appendix
\section{The coupling coefficients}
\label{ap:1}
Let us denote by
\begin{align}
  \mathbb{O}_{\bD\bD}^{\epsilon} = \left[\partial_z \nu(\vx) -
  \frac{\epsilon}{2} \partial_z \nu^2(\vx) \right] I  \,,
\end{align}
and 
\begin{align}
  \mathbb{O}_{\bD\bU}^{\epsilon} = i k \nu(\vx) {\bf I} - \frac{i}{k}\nabla 
\nu(\vx) \nabla \cdot + \frac{i \epsilon}{2 k} \nabla 
\nu^2(\vx) \nabla \cdot\, ,
\end{align}
the perturbation operators in (\ref{eq:P3}) acting on $\bD$ and $\bU$,
where ${\bf I}$ is the identity. Similarly, we let
$\mathbb{O}_{\bU\bD}^{\epsilon}$ be the perturbation operator in
(\ref{eq:P4}) acting on $\bU$.  The coupling coefficients in equations
(\ref{eq:RD6}-\ref{eq:RD7}) are linear combinations of
\begin{align*}
  \alpha_{jj'}^{(ss')\epsilon}(z) &= \left(\sqrt{\frac{\beta_j}{k}}
    \de_{s1} + \sqrt{\frac{k}{\beta_j}} \de_{s2}\right) \left<
    \bphi_j^{(s)}, \mathbb{O}_{\bD\bD}^{\epsilon} \bphi_{j'}^{(s')}
  \right> \left( \sqrt{\frac{k}{\beta_{j'}}}
    \de_{s'1} + \sqrt{\frac{\beta_{j'}}{k}} \de_{s'2} \right)\,, \\
  \gamma_{jj'}^{(ss')\epsilon}(z) &= \left(\sqrt{\frac{\beta_j}{k}}
    \de_{s1} + \sqrt{\frac{k}{\beta_j}} \de_{s2}\right) \left<
    \bphi_j^{(s)}, \mathbb{O}_{\bD\bU}^{\epsilon} \bphi_{j'}^{(s')}
  \right> \left(\sqrt{\frac{\beta_{j'}}{k}} \de_{s'1} +
    \sqrt{\frac{k}{\beta_{j'}}} \de_{s'2}\right)\, ,
  \\
  \eta_{jj'}^{(ss')\epsilon}(z) &= \left(
      \sqrt{\frac{k}{\beta_j}} \de_{s1} + \sqrt{\frac{\beta_j}{k}}
    \de_{s2}\right) \left< \bphi_j^{(s)},
    \mathbb{O}_{\bU\bD}^{\epsilon} \bphi_{j'}^{(s')} \right>
  \left(\sqrt{\frac{k}{\beta_{j'}}} \de_{s'1} +
    \sqrt{\frac{\beta_{j'}}{k}} \de_{s'2} \right)\,,
\end{align*}
where $\left< \cdot, \cdot \right>$ denotes inner product in
$L^2(\Omega)$.  These expressions become after integration by parts
\begin{align}
  \alpha_{jj'}^{(ss')\epsilon}(z) = &\left(\sqrt{\frac{\beta_j}{k}}
    \de_{s1} + \sqrt{\frac{k}{\beta_j}} \de_{s2}\right)\left(
    \sqrt{\frac{k}{\beta_{j'}}} \de_{s'1} +
    \sqrt{\frac{\beta_{j'}}{k}} \de_{s'2} \right)
  \times \nonumber \\
  &\left[\partial_z \Psi_{jj'}^{(ss')}(z) - \frac{\epsilon}{2} \partial_z
    \psi_{jj'}^{(ss')}(z)\right]\,, \label{eq:defalpha}
\end{align}
and 
\begin{align}
  \gamma_{jj'}^{(ss')\epsilon}(z) =&
  \frac{i}{k}\left(\sqrt{\frac{\beta_j}{k}} \de_{s1} +
    \sqrt{\frac{k}{\beta_j}} \de_{s2}\right)
  \left(\sqrt{\frac{\beta_{j'}}{k}} \de_{s'1} +
    \sqrt{\frac{k}{\beta_{j'}}} \de_{s'2}\right) \times \nonumber \\
  &\hspace{-0.4in}\left\{ (k^2-\lambda_{j'} \de_{s'2})
    \Psi_{jj'}^{(ss')}(z) + \Theta_{jj'}^{(ss')}(z) + \frac{ \epsilon}{2}
    \left[\lambda_{j'} \de_{s'2}\psi_{jj'}^{(ss')}(z) -
      \theta_{jj'}^{(ss')}(z)\right]\right\}\, , \label{eq:defgamma}
\end{align}
and 
\begin{align}
  \eta_{jj'}^{(ss')\epsilon}(z) = \frac{i \lambda_j \de_{s1} }{
    \sqrt{k \beta_j}} \left( \sqrt{\frac{k}{\beta_{j'}}} \de_{s'1} +
    \sqrt{\frac{\beta_{j'}}{k}} \de_{s'2} \right) \left[
    \Psi_{jj'}^{(ss')}(z) - \epsilon \psi_{jj'}^{(ss')}(z)\right] \, ,
  \label{eq:defeta}
\end{align}
where we introduced the notation
\begin{align}
  \Psi_{jj'}^{(ss')}(z) &= \int_{\Omega} d \bx \, \nu(\vx)
  \phi_j^{(s)}(\bx)
  \cdot \phi_{j'}^{(s')}(\bx)\, , \nonumber \\
  \psi_{jj'}^{(ss')}(z) & = \int_{\Omega} d \bx \, \nu^2(\vx)
  \phi_j^{(s)}(\bx) \cdot \phi_{j'}^{(s')}(\bx)\, , \label{eq:PP}
\end{align}
and 
\begin{align}
  \Theta_{jj'}^{(ss')}(z) &= \int_{\Omega} d \bx \, \nu(\vx)\, \nabla \cdot
  \phi_j^{(s)}(\bx)
  \nabla \cdot \phi_{j'}^{(s')}(\bx)\, , \nonumber \\
  \theta_{jj'}^{(ss')}(z) &= \int_{\Omega} d \bx \, \nu^2(\vx)\, \nabla \cdot
  \phi_j^{(s)}(\bx) \nabla \cdot \phi_{j'}^{(s')}(\bx)\, .
  \label{eq:DD}
\end{align} 
Note that since $\nabla\cdot \phi_j^{(1)} = 0$, we have 
\begin{equation}
\label{eq:structD}
\Theta_{jj'}^{(ss')}(z) = \de_{s2} \de_{s'2} \Theta_{jj'}^{(22)}(z)\,,
\end{equation}
and similar for $\theta_{jj'}^{(ss')}(z)$.

The coupling coefficients in (\ref{eq:RD6}) are given by 
\begin{align}
  M_{AA,jj'}^{(ss')}(z) + \epsilon \, m_{AA,jj'}^{(ss')}(z) &=
  \frac{1}{2} \left[\alpha_{jj'}^{(ss')\epsilon}(z) +
    \gamma_{jj'}^{(ss')\epsilon}(z) +
    \eta_{jj'}^{(ss')\epsilon}(z)\right]\, ,
  \label{eq:MAA}\\
  M_{AB,jj'}^{(ss')}(z) + \epsilon \, m_{AB,jj'}^{(ss')}(z) &=
  \frac{1}{2} \left[\alpha_{jj'}^{(ss')\epsilon}(z) -
    \gamma_{jj'}^{(ss')\epsilon}(z) +
    \eta_{jj'}^{(ss')\epsilon}(z)\right] \, ,
  \label{eq:MAB}\\
M_{AV,jj'}^{(ss')}(z) + \epsilon \, m_{AV,jj'}^{(ss')}(z) &=
  \frac{1}{2} \left[\alpha_{jj'}^{(ss')\epsilon}(z)  +
    \eta_{jj'}^{(ss')\epsilon}(z)\right] \, ,
  \label{eq:MAV}\\
  M_{Av,jj'}^{(ss')}(z) + \epsilon \, m_{Av,jj'}^{(ss')}(z) &=
  -\frac{i}{2}(-1)^{s'}
  \gamma_{jj'}^{(ss')\epsilon}(z) \, ,
  \label{eq:MAv}
\end{align}
and those in (\ref{eq:RD7}) are 
\begin{align}
  M_{BA,jj'}^{(ss')}(z) + \epsilon \, m_{BA,jj'}^{(ss')}(z) &=
  \frac{1}{2} \left[\alpha_{jj'}^{(ss')\epsilon}(z) +
    \gamma_{jj'}^{(ss')\epsilon}(z) -
    \eta_{jj'}^{(ss')\epsilon}(z)\right]\, ,
  \label{eq:MBA}\\
  M_{BB,jj'}^{(ss')}(z) + \epsilon \, m_{BB,jj'}^{(ss')}(z) &=
  \frac{1}{2} \left[\alpha_{jj'}^{(ss')\epsilon}(z) -
    \gamma_{jj'}^{(ss')\epsilon}(z) -
    \eta_{jj'}^{(ss')\epsilon}(z)\right] \, ,
  \label{eq:MBB}\\
M_{BV,jj'}^{(ss')}(z) + \epsilon \, m_{BV,jj'}^{(ss')}(z) &=
  \frac{1}{2} \left[\alpha_{jj'}^{(ss')\epsilon}(z)  -
    \eta_{jj'}^{(ss')\epsilon}(z)\right] \, ,
  \label{eq:MBV}\\
  M_{Bv,jj'}^{(ss')}(z) + \epsilon \, m_{Bv,jj'}^{(ss')}(z) &=-
  \frac{i}{2} (-1)^{s'}\gamma_{jj'}^{(ss')\epsilon}(z) \, .
  \label{eq:MBv}
\end{align}
Because $\gamma_{jj'}^{(ss')\epsilon}(z)$ and
$\eta_{jj'}^{(ss')\epsilon}(z)$ are imaginary, we obtain the relations
\begin{align}
  M_{BB,jj'}^{(ss')}(z) &= \overline{M_{AA,jj'}^{(ss')}(z)}\, , \qquad
  m_{BB,jj'}^{(ss')}(z) = \overline{
    m_{AA,jj'}^{(ss')}(z)}\, ,  \label{eq:rel1}\\
  M_{BA,jj'}^{(ss')}(z) &= \overline{M_{AB,jj'}^{(ss')}(z)}\, , \qquad
  m_{BA,jj'}^{(ss')}(z) = \overline{
    m_{AB,jj'}^{(ss')}(z)}\, ,\label{eq:rel2}\\
  M_{BV,jj'}^{(ss')}(z)&=\overline{M_{AV,jj'}^{(ss')}(z)}\, , \qquad
  m_{BV,jj'}^{(ss')}(z) = \overline{
    m_{AV,jj'}^{(ss')}(z)}\, , \label{eq:rel3}\\
  M_{Bv,jj'}^{(ss')}(z) &= \overline{M_{Av,jj'}^{(ss')}(z)}\, , \qquad
  m_{Bv,jj'}^{(ss')}(z) =\overline{ m_{Av,jj'}^{(ss')}(z)}\, .
  \label{eq:rel4}
 \end{align}

 Similarly, the coupling coefficients in equations
 (\ref{eq:ev1}-\ref{eq:ev2}) satisfied by the evanescent modes are
 given by
\begin{align}
  M_{VA,jj'}^{(ss')}(z) + \epsilon \, m_{VA,jj'}^{(ss')}(z) &=
  \alpha_{jj'}^{(ss')\epsilon}(z) + 
    \gamma_{jj'}^{(ss')\epsilon}(z)\, ,
    \label{eq:MVA}\\
    M_{vA,jj'}^{(ss')}(z) + \epsilon \, m_{vA,jj'}^{(ss')}(z) &= i (-1)^s
    \eta_{jj'}^{(ss')\epsilon}(z) \, ,
    \label{eq:MvA}\\
    M_{VV,jj'}^{(ss')}(z) + \epsilon \, m_{VV,jj'}^{(ss')}(z) &=
    \alpha_{jj'}^{(ss')\epsilon}(z) \, ,
    \label{eq:MVV}\\
    M_{Vv,jj'}^{(ss')}(z) + \epsilon \, m_{Vv,jj'}^{(ss')}(z) &= -i
    (-1)^{s'} \gamma_{jj'}^{(ss')\epsilon}(z) \, ,
    \label{eq:MVv}\\
    M_{vV,jj'}^{(ss')}(z) + \epsilon \, m_{Vv,jj'}^{(ss')}(z) &= i(-1)^s
     \eta_{jj'}^{(ss')\epsilon}(z) \, ,
  \label{eq:MvV}
\end{align}
and by 
\begin{align}
  M_{VB,jj'}^{(ss')}(z) &= \overline{M_{VA,jj'}^{(ss')}(z)}\, , \qquad
  m_{VB,jj'}^{(ss')}(z) = \overline{m_{VA,jj'}^{(ss')}(z)}\, ,
    \label{eq:MVB}\\
M_{vB,jj'}^{(ss')}(z) & = \overline{M_{vA,jj'}^{(ss')}(z)}\, ,
\qquad m_{vB,jj'}^{(ss')}(z) = \overline{m_{vA,jj'}^{(ss')}(z)}\, .
  \label{eq:MvB}
\end{align}
In equations (\ref{eq:ev1}-\ref{eq:ev2}) we use only the leading part 
of these coefficients denoted by the capital letter $M$, as in 
$M_{VA,jj'}^{(ss')}(z)$. 

\section{Analysis of the evanescent modes}
\label{ap:2}
Consider the $2 \times 2$ system
\begin{equation}
\partial_z \bcV_j^{(s)}(z) + \beta_j \left(\begin{array}{cc}
0 ~ ~ & 1\\
1 ~ ~ &0 
\end{array} \right) \bcV_j^{(s)}(z) = \epsilon \bcF_j^{(s)}(z,\bcV)\, ,
\label{eq:A2.1}
\end{equation}
 for vectors 
\begin{equation}
\bcV_j^{(s)}(z) = \left( \begin{array}{c} V_j^{(s)}(z) \\ v_j^{(s)}(z)
  \end{array} \right)\, , 
\end{equation}
which we string together in the infinite vector $\bcV(z)$. The system
(\ref{eq:ev1}-\ref{eq:ev2}) is of this form, with right hand-side
\begin{equation}
\bcF_j^{(s)}(z,\bcV) = 
\left( \begin{array}{c} 
F_j^{(s)}(z)\\
f_j^{(s)}(z)\end{array} \right)  + \sum_{j'>N} \sum_{s'=1}^{\mM_{j'}}
\left( \begin{array}{cc} 
M_{VV,jj'}^{(ss')}(z) ~ ~ &  M_{Vv,jj'}^{(ss')}(z) \\
M_{vV,jj'}^{(ss')}(z) & 0 \end{array}\right) \bcV_{j'}^{(s')}(z)\,.
\label{eq:A2}
\end{equation}
We neglect the $O(\epsilon^2)$ remainder because it plays no role in
our setup.

We diagonalize (\ref{eq:A2.1}) by writing $\bcV^{(s)}_j$ in the
orthonormal basis $\{ \bu^+, ~ \bu^- \}$ of $\mathbb{R}^2$,  where 
$
\bu^\pm = \frac{1}{\sqrt{2}} \left(\begin{array}{c} 
1 \\ \pm 1 \end{array} \right)\, \label{eq:A2.3}
$
are the eigenvectors of matrix $\left(\begin{array}{cc}
    0 ~ ~ & 1\\
    1 ~ ~ &0 \end{array} \right)$ for eigenvalues $\pm 1$.
Explicitly, we write
\begin{equation}
\bcV_j^{(s)}(z) = \theta_{j}^{(s)+} (z) \bu^+ + 
 \theta_{j}^{(s)-} (z) \bu^-\, , \label{eq:A2.4}
\end{equation}
where $\theta_j^{(s)\pm}(z)$ are decaying and growing evanescent 
waves, satisfying
\begin{align}
\left[\partial_z \pm \beta_j \right] \theta_{j}^{(s)\pm} (z) = 
\epsilon \bu^\pm \cdot \bcF_j^{(s)}(z,\bcV)\, ,  \label{eq:A2.5}
\end{align}
the initial condition
\begin{equation}
  \theta_j^{(s)+}(0+) = \bu^+ \cdot \bcV_j^{(s)}(0+) = 
\sqrt{2} E_{j,o}^{(s)}\,,\label{eq:A2.6}
\end{equation}
and the end condition
\begin{equation}
\theta_j^{(s)-}(z_{\rm max}) =0.\label{eq:A2.7}
\end{equation}
We obtain after integrating equations (\ref{eq:A2.5}) and using 
(\ref{eq:A2.4}) that 
\begin{align}
  V_{j}^{(s)}(z) =& E_{j,o}^{(s)} e^{-\beta_j z} +
  \frac{\epsilon}{\sqrt{2}} \int_0^z d \zeta \bu^+ \cdot
  \bcF_j^{(s)}(\zeta,\bcV) e^{-\beta_j(z-\zeta)} - \nonumber \\&
  \frac{\epsilon}{\sqrt{2}} \int_z^{z_{\rm max}} d \zeta \, \bu^-
  \cdot \bcF_j^{(s)}(\zeta,\bcV) e^{\beta_j(z-\zeta)}\,,
\label{eq:A2.8}
\end{align} 
and 
\begin{align}
  v_{j}^{(s)}(z) =& E_{j,o}^{(s)} e^{-\beta_j z} +
  \frac{\epsilon}{\sqrt{2}} \int_0^z d \zeta \bu^+ \cdot
  \bcF_j^{(s)}(\zeta,\bcV) e^{-\beta_j(z-\zeta)} + \nonumber \\&
  \frac{\epsilon}{\sqrt{2}} \int_z^{z_{\rm max}} d \zeta \, \bu^-
  \cdot \bcF_j^{(s)}(\zeta,\bcV) e^{\beta_j(z-\zeta)}\,,
\label{eq:A2.9}
\end{align} 
for $j > N$ and $1 \le s \le \mM_j$. 

Equations (\ref{eq:A2.8}-\ref{eq:A2.9}) form an infinite system of
integral equations for the vector $\bcV(z)$, which we write in compact
form as
\begin{equation}
\left[ \mathbb{I} - \epsilon \mathbb{Q} \right] \bcV(z) =
{\bf F}(z)\, , 
\label{eq:A2.10}
\end{equation}
with infinite vector ${\bf F}(z)$ given by the terms in the right
hand-side of (\ref{eq:A2.8}-{\ref{eq:A2.9}) that are independent of
  $\bcV$.  The operator in the left hand-side is a perturbation of the
  identity $\mathbb{I}$, with $\mathbb{Q}$ the linear integral
  operator that takes the infinite vector $\bcV(z)$ and returns the
  infinite vector obtained by concatenating the entries
\begin{align*}
  \frac{1}{\sqrt{2}} \int_0^z d \zeta \, \sum_{j'>N}
  \sum_{s'=1}^{\mM_{j'}} \bu^+ \cdot 
\left(\begin{array}{cc}
M_{VV,jj'}^{(ss')}(\zeta) ~ ~ & 
    M_{Vv,jj'}^{(ss')}(\zeta)\\
      M_{vV,jj'}^{(ss')}(\zeta) & 0 \end{array} \right) 
  \bcV_{j'}^{(s')}(\zeta) e^{-\beta_j(z-\zeta)}\, \mp \\
\frac{1}{\sqrt{2}}
    \int_z^{z_{\rm max}} d \zeta \, \sum_{j'>N} \sum_{s'=1}^{\mM_{j'}}
      \bu^- \cdot \left(\begin{array}{cc}
          M_{VV,jj'}^{(ss')}(\zeta) ~ ~ & M_{Vv,jj'}^{(ss')}(\zeta)\\
          M_{vV,jj'}^{(ss')}(\zeta) & 0 \end{array} \right) 
      \bcV_{j'}^{(s')}(\zeta) e^{\beta_j(z-\zeta)}\, ,
\end{align*}
for $j > N$ and $1 \le s \le \mM_j$. We show in Lemma \ref{lem:bdd}
that $\mathbb{Q}$ is a bounded linear operator, so we can solve
(\ref{eq:A2.10}) using Neumann series
\[
\bcV(z) = \left[ \mathbb{I} + \epsilon \mathbb{Q} + \ldots \right]{\bf
  F}(z)\,,
\]
to obtain
\begin{align}
  V_{j}^{(s)}(z) =& E_{j,o}^{(s)} e^{-\beta_j z} + \frac{\epsilon}{2}
  \int_0^z d \zeta \left[ F_j^{(s)}(\zeta) + f_j^{(s)}(\zeta)\right]
  e^{-\beta_j(z-\zeta)} - \nonumber \\& \frac{\epsilon}{2}
  \int_z^{z_{\rm max}} d \zeta \left[F_j^{(s)}(\zeta)
    -f_j^{(s)}(\zeta)\right] e^{\beta_j(z-\zeta)} + O(\epsilon^2)\,,
\label{eq:A2.11}
\end{align} 
and 
\begin{align}
  v_{j}^{(s)}(z) =& E_{j,o}^{(s)} e^{-\beta_j z} + \frac{\epsilon}{2}
  \int_0^z d \zeta \left[ F_j^{(s)}(\zeta) + f_j^{(s)}(\zeta)\right]
  e^{-\beta_j(z-\zeta)} + \nonumber \\& \frac{\epsilon}{2}
  \int_z^{z_{\rm max}} d \zeta \left[F_j^{(s)}(\zeta)
    -f_j^{(s)}(\zeta)\right] e^{\beta_j(z-\zeta)} + O(\epsilon^2)\,,
\label{eq:A2.12}
\end{align} 
for $j > N$ and $1 \le s \le \mM_j$.  The result in Lemma \ref{lem.2}
follows from (\ref{eq:A2.11}-\ref{eq:A2.12}) and the approximations
\[
\int_0^z d\zeta \, \psi(\zeta) e^{-\beta_j(z-\zeta)} = 
\int_{-z}^0  \hspace{-0.1in}dt \, \psi(z + t)e^{-\beta t} \approx 
\int_{-\infty}^0  \hspace{-0.1in}dt \,  \psi(z + t)e^{-\beta t}\, ,
\]
and 
\[
\int_{z}^{z_{\rm max}}  \hspace{-0.1in}d \zeta \, \psi(\zeta) e^{\beta_j(z-\zeta)}
= \int_0^{z_{\rm max}-z}  \hspace{-0.1in}dt \, \psi(z + t) e^{-\beta_j |t|} 
\approx \int_0^\infty dt\, \psi(z + t) e^{-\beta_j |t|} \, ,
\]
for an arbitrary bounded function $\psi(z)$. The error in these
approximations is similar to $e^{-\beta_j z}$, and we can neglect it
for large $z$.
\subsection{Coupling of the propagating modes via the evanescent ones}
\label{ap:2coeff}
The substitution of the evanescent components defined in Lemma
\ref{lem.2} in equations (\ref{eq:RD6}-\ref{eq:RD7}) gives a closed
system of equations for the amplitudes of the propagating modes.  The
effect of the evanescent modes is captured by the coefficients
$\epsilon^2 m_{AA,jj'}^{(ss')e}(z)$, $\epsilon^2
m_{AB,jj'}^{(ss')e}(z)$, $\epsilon^2 m_{BA,jj'}^{(ss')e}(z)$ and
$\epsilon^2 m_{BB,jj'}^{(ss')e}(z)$ in equations
(\ref{eq:C1}-\ref{eq:C2}). We write explicitly just the first of them
\begin{align}
  m_{AA,jj'}^{(ss')e}(z)=& \frac{1}{2}\sum_{l>N} \sum_{q = 1}^{\mM_l}
  \left[ \int_{-\infty}^\infty \hspace{-0.1in}d \zeta \, M_{AV,jl}^{(sq)}(z)
    M_{vA,lj'}^{(qs')}(z+\zeta) e^{i \beta_j' \zeta - \beta_l|\zeta|}
    + \right. \nonumber \\
  & \int_{-\infty}^\infty \hspace{-0.1in} d \zeta \, M_{Av,jl}^{(sq)}(z)
  M_{VA,lj'}^{(qs')}(z+\zeta)
  e^{i \beta_j' \zeta - \beta_l|\zeta|} + \nonumber \\
  & \int_{0}^\infty \hspace{-0.1in}d
  \zeta \, M_{AV,jl}^{(sq)}(z) \sum_{r=0}^1 (-1)^r M_{VA,lj'}^{(qs')}(z-
  (-1)^r\zeta) e^{-i(-1)^r \beta_j'\zeta - \beta_l\zeta} +
  \nonumber \\
  &\left. \int_{0}^\infty  \hspace{-0.1in} d
  \zeta \, M_{Av,jl}^{(sq)}(z) \sum_{r=0}^1 (-1)^r
  M_{vA,lj'}^{(qs')}(z-(-1)^r\zeta) e^{-i (-1)^r \beta_j' \zeta
    -\beta_l\zeta} \right]\, .
\label{eq:MAAe}
\end{align}
The other three coefficients are similar.

\subsection{Proof that $\mathbb{Q}$ is a bounded operator}
\label{ap:2proof}
The operator $\mathbb{Q}$ is defined by
\begin{equation*}
\big(\mathbb{Q}\,\boldsymbol{\mathcal{V}}\big)^{(s)}_{j}(z) = \left(
\begin{array}{c}
\big(\mathbb{Q}\,V\big)^{(s)}_{j}(z) \\
\big(\mathbb{Q}\,v\big)^{(s)}_{j}(z)
\end{array}\right), 
\end{equation*}
for $j>N$ and $1\leq s\leq \mathfrak{M}_{j}$, where
\begin{align*}
  \big(\mathbb{Q}\,V\big)^{(s)}_{j}(z) =
  &\frac{1}{2}\int^{z}_{0}d\zeta\,
  e^{-\beta_{j}(z-\zeta)}\sum_{j'>N}\sum^{\mathfrak{M}_{j'}}_{s'=1}
  \Big[\big[\alpha^{(ss')}_{jj'}(\zeta) + i\big(-1\big)^{s}
  \eta^{(ss')}_{jj'}(\zeta)\big] V^{(s')}_{j'}(\zeta) -\\
  & \hspace{1.8in}i(-1)^{s'}\gamma^{(ss')}_{jj'}(\zeta)\,v^{(s')}_{j'}
  (\zeta) \Big] - \\
  & \frac{1}{2}\int^{z_{\text{max}}}_{z}\hspace{-0.1in}d\zeta\,
  e^{-\beta_{j}(\zeta - z)}\sum_{j'>N}\sum^{\mathfrak{M}_{j'}}_{s'=1}
  \Big[\big[\alpha^{(ss')}_{jj'}(\zeta) - i\big(-1\big)^{s}
  \eta^{(ss')}_{jj'}(\zeta)\big] V^{(s')}_{j'}(\zeta) -\\
  & \hspace{1.8in}i(-1)^{s'}\gamma^{(ss')}_{jj'}(\zeta)\,v^{(s')}_{j'}
  (\zeta) \Big].
\end{align*}
The components $\big(\mathbb{Q}v\big)^{(s)}_{j}(z)$ are similar, with
addition (instead of subtraction) of the integrals above.

\begin{lemma}
\label{lem:bdd}
The linear operator $\mathbb{Q}$ is bounded in the space of square
summable sequences of $L^{2}(\mathbb{R}^{+})$ vector--functions with
$w$--weights, equipped with the norm
\begin{equation*}
  \| \boldsymbol{\mathcal{V}} \|_{w} = \left(
    \sum_{j>N}\sum^{\mathfrak{M}_{j}}_{s=1} 
    \big(\beta_{j}\,w^{(s)}_{j}\| 
    \boldsymbol{\mathcal{V}}^{(s)}_{j} \|_{L^{2}(\mathbb{R}^{+})} 
    \big)^{2} \right)^{1/2} \;,\quad w^{(s)}_{j} = 
  \left(\sqrt{\frac{k}{\beta_{j}}}\delta_{s1} + 
    \sqrt{\frac{\beta_{j}}{k}}\delta_{s2} \right).
\end{equation*}
\end{lemma}
\begin{proof}
We present in detail the estimation of the most critical terms in
$\mathbb{Q}$, involving the processes $\gamma^{(ss')}_{jj'}(z)$.  The
remaining terms are treated similarly.  We rewrite the processes
$\gamma^{(ss')}_{jj'}(z)$ as
\begin{align}
  i\big(-1\big)^{s}\gamma^{(ss')}_{jj'}(z)&= \frac{w^{(s')}_{j'}
  }{ w^{(s)}_{j} }\Big(\frac{\beta_{j'}}{k} \delta_{s'1} -
  \frac{k}{\beta_{j'}}\delta_{s'2} \Big)
  \frac{1}{k}\Big[\big(k^{2} -
  \delta_{s'2}\lambda_{j'}\big)\Psi_{jj'}^{(ss')}(z) + \Theta_{jj'}^{(ss')}(z)
 \Big] \nonumber \\
  &=:\frac{w^{(s')}_{j'} }{ w^{(s)}_{j}
  }\Big(\frac{\beta_{j'}}{k}\delta_{s'1} -
  \frac{k}{\beta_{j'}}\delta_{s'2}
  \Big)\,\tilde{\gamma}^{(ss')}_{jj'}(z), \label{AEe1}
\end{align}
and introduce the auxiliary operator,
\begin{align*}
  \big(\tilde{\mathbb{Q}}_{\gamma} v\big)^{(s)}_{j}(z) &=
  \frac{1}{2}\int^{z}_{0}d\zeta\,e^{-\beta_{j}(z-\zeta)}
  \sum_{j'>N}\sum^{\mathfrak{M}_{j'}}_{s'=1}
  \Big(\frac{\beta_{j'}}{k}\delta_{s'1} - \frac{k}{\beta_{j'}}
  \delta_{s'2}\Big)\tilde{\gamma}^{(ss')}_{jj'}(\zeta)\,v^{(s')}_{j'}
  (\zeta) .
\end{align*}
We show that this operator is bounded in the space of square summable
sequences of $L^{2}(\mathbb{R}^{+})$ vector--functions equipped with
the norm
\begin{equation*}
 \| \boldsymbol{\mathcal{V}} \|
  = \left( \sum_{j>N}\sum^{\mathfrak{M}_{j}}_{s=1} \big(\beta_{j}\|
    \boldsymbol{\mathcal{V}}^{(s)}_{j} \|_{L^{2}(\mathbb{R}^{+})}
    \big)^{2} \right)^{1/2}.
\end{equation*}
Indeed,  Young's inequality for convolutions implies
\begin{align}\label{AEe2}
  \| \tilde{\mathbb{Q}}_{\gamma} v \|^{2} &= \sum_{j>N}
  \sum^{\mathfrak{M}_{j}}_{s=1}\beta^{2}_{j}\,\|
  \big(\tilde{\mathbb{Q}}_{\gamma}v\big)^{(s)}_{j}(z)\|^{2}_{L^{2}
    (\mathbb{R}^{+})} \nonumber \\
  & \hspace{-0.5in}\leq \frac{1}{2^{p}}\sum_{j>N}
  \sum^{\mathfrak{M}_{j}}_{s=1}\beta^{2}_{j}
  \| e^{-\beta_{j} z} \|^{2}_{L^{1}(\mathbb{R}^{+})}  \Big\| \sum_{j'>N}\sum^{\mathfrak{M}_{j'}}_{s'=1}
  \Big(\frac{\beta_{j'}}{k}\delta_{s'1} - \frac{k}{\beta_{j'}}
  \delta_{s'2}\Big)\tilde{\gamma}^{(ss')}_{jj'}(z)\,v^{(s')}_{j'}(z)
  \Big\|^{2}_{L^{2}(\mathbb{R}^{+})}  \nonumber \\
  &\hspace{-0.5in}\leq \frac{1}{2}\sum_{j>N} \sum^{\mathfrak{M}_{j}}_{s=1}
  \Big(\text{a}^{(s)}_{j} + \text{b}^{(s)}_{j}\Big),
\end{align}
where we used that $\beta_{j}\, \| e^{-\beta_{j} z}
\|_{L^{1}(\mathbb{R}^{+})}=1$, and let
\begin{align}
  \text{a}^{(s)}_{j} :&= \Big\|
  \sum_{j'>N}\sum^{\mathfrak{M}_{j'}}_{s'=1}
  \frac{\beta_{j'}}{k}\,\delta_{s'1}\,\tilde{\gamma}^{(ss')}_{jj'}(z)\,
  v^{(s')}_{j'}(z)  \Big\|^{2}_{L^{2}(\mathbb{R}^{+})} \nonumber \\
  &= \int_{\mathbb{R}^{+}}dz\,\Big|\sum_{j'>N}\Psi_{jj'}^{(s1)}(z)
  \beta_{j'}v^{(1)}_{j'}(z) \Big|^{2}, \label{eq:defa}
\end{align}
and
\begin{align}
  \text{b}^{(s)}_{j} :&= \Big\|
  \sum_{j'>N}\sum^{\mathfrak{M}_{j'}}_{s'=1}
  \frac{k}{\beta_{j'}}\,\delta_{s'2}\,\tilde{\gamma}^{(ss')}_{jj'}(z)\,
  v^{(s')}_{j'}(z) \Big\|^{2}_{L^{2}(\mathbb{R}^{+})} \nonumber \\ & =
  \int_{\mathbb{R}^{+}}dz\,\Big|\sum_{j'>N}
  \sum^{\mathfrak{M}_{j'}}_{s'=1} \frac{\de_{s'2}}{\beta^2_{j'}} \Big[
    \big(k^{2} - \delta_{s1}\,\lambda_{j'}\big)\Psi_{jj'}^{(ss')}(z) +
    \nonumber \\ & \hspace{3cm}
    \delta_{s2}\big(\Theta_{jj'}^{(ss')}(z) -
    \lambda_{j'}\Psi_{jj'}^{(ss')}(z)\big) \Big]
  \beta_{j'}v^{(s')}_{j'}(z) \Big|^{2}\, . \label{eq:defb}
\end{align}

To estimate $\text{a}^{(s)}_j$, set $j=(n,l)$ and $j'=(n',l')$ and use
integration by parts to obtain for $n\neq n'$ and $l\neq l'$
\begin{align*}
  \Psi_{jj'}^{(s1)}(z) &= \int_{\Omega}d\textbf{x}\,\nu(\vx)\,
  \phi^{(s)}_{j}(\textbf{x})\cdot\phi^{(1)}_{j'}(\textbf{x}) \\&= \frac{
    \pi^2}{(L_{1}L_{2})^2} \frac{l\,n'}{
    \sqrt{\lambda_{j}\lambda_{j'}}}\int_{\Omega}d\textbf{x}\,
  \nu(\vx)\,\cos\big(\tfrac{\pi(n-n')x_1}{L_{1}}\big)\,
  \cos\big(\tfrac{\pi(l-l')x_2}{L_{2}}\big) + \cdots \\
  & =: \frac{l\,n'\,\hat \nu_{x_1x_2}(n-n',l-l',z) }{2\sqrt{L_1
      L_2}\,(n-n')(l-l')\sqrt{\lambda_{j}\lambda_{j'}}} + \cdots.
\end{align*}
Here we let 
\[
\hat \nu_{x_1 x_2}(n,l,z) = \frac{2}{\sqrt{L_1 L_2}} \int_\Omega d \bx
\, \partial^2_{x_1x_2} \nu(\vx) \sin \left(\frac{\pi n
    x_1}{L_1}\right) \sin \left(\frac{\pi l
    x_2}{L_2}\right)\, .
\]
The dots stand for three similar terms in the case $s=2$ and four in
the case $s=1$, having the rest of index combinations $n\pm n'$ and $l
\pm l'$.  For the cases $n=n'$ or $l=l'$ we assume for convenience
that the integral of $\nu$ over $x_1$ or $x_2$ is zero.  Then, with
the understanding that the sums are performed over $n\neq n'$ or
$l\neq l'$, we have
\begin{equation}\label{AEe3}
  \text{a}^{(s)}_{j} \leq \frac{1}{L_{1} L_{2}} \int_{\mathbb{R}^{+}}dz\, 
  \Big|\sum_{j'>N}\frac{l\,n'\,\hat \nu_{x_1x_2}(n-n',l-l',z) }{
    (n-n')(l-l')\sqrt{\lambda_{j}\lambda_{j'}}}\,
  \beta_{j'}v^{(1)}_{j'}(z)\Big|^{2} + \cdots,
\end{equation}
where the $j'$--sum can be regarded as a discrete convolution
\begin{align*}
  \sum_{j'>N}\frac{l\,n'\,\hat \nu_{x_1x_2}(n-n',l-l',z) }{
    (n-n')(l-l')\sqrt{\lambda_{j}\lambda_{j'}}}\,\beta_{j'}\,&v^{(1)}_{j'}
  (z) \\
  = \frac{l}{\sqrt{\lambda_{j}}}&\Big(\frac{\hat \nu_{x_1x_2}(n',l',z)
  }{n'\,l'} \Big) \star
  \Big(\frac{n'\,\beta_{j'}}{\sqrt{\lambda_{j'}}}\,v^{(1)}_{j'}(z)
  \Big)(j)\, .
\end{align*}
Thus, invoking Young's inequality for discrete sums, and the simple
inequality 
\[l\,n' \leq
\frac{L_{1}L_{2}}{\pi^2}\sqrt{\lambda_{j}\lambda_{j'}},\] 
we obtain that 
\begin{align}\label{AEe4}
  \frac{1}{2}\sum_{j>N} \sum^{\mathfrak{M}_{j}}_{s=1}
  \text{a}^{(s)}_{j} &\leq C\,\int_{\mathbb{R}^{+}} dz\,
  \left(\sum_{n\neq 0}\sum_{l\neq0} \Big|\frac{\hat\nu_{x_1x_{2}}(n,l,z)
    }{n\,l}\Big| \right)^{2}\,\sum_{j>N} \big|\beta_{j}\,v^{(1)}_{ j
  }(z)\big|^{2}
  \nonumber\\
  &\leq C\, \Big(\sum_{n\neq0}
  \frac{1}{n^{2}}\Big)^{2}\,\sup_{z\geq0}\big\|\hat
  \nu_{x_1x_2}(n,l,z)
  \big\|^{2}_{\ell^{2}}\,\big\|v\big\|^{2}\,,
\end{align}
with constant $C:=C\big(p,L_{1},L_{2}\big)$. The last inequality follows
from Cauchy-Schwarz, and 
\[
\big\|\hat
  \nu_{x_1x_2}(n,l,z)
  \big\|^{2}_{\ell^{2}} = \sum_{n\ne 0} \sum_{l \ne 0} 
\big|\hat \nu_{x_1x_2}(n,l,z)\big|^{2}\, .
\] 
Moreover, since the set
$\big\{\tfrac{2}{\sqrt{L_{1}L_{2}}}\sin\big(\tfrac{\pi\,n}{L_{1}}x_{1}\big)
\sin\big(\tfrac{\pi\,l}{L_{2}}x_{2}\big)\big\}_{n,l}$ is a subset of
the orthonormal Fourier basis, we have 
\begin{equation*}
  \big\|\hat \nu_{x_1x_2}(n,l,z)
  \big\|_{\ell^{2}} \leq \|\partial^{2}_{x_1 x_2}
  \nu(\textbf{x}, z)\|_{L^{2}(\Omega)} \leq \|
  \nu(\textbf{x}, z)\|_{H^{2}(\Omega)},\quad \forall z\geq0.
\end{equation*}

It remains to estimate the term $\text{b}^{(s)}_{j}$ given by
(\ref{eq:defb}). The first term in the $j'$--sum is similar to that in
$\text{a}^{(s)}_{j}$, estimated above, because
$\big|k^{2}-\lambda_{j'}\big| = \beta^{2}_{j'}$.  The second term in
the sum appears only for $s = s' = 2$,
\begin{align*}
  \Theta_{jj'}^{(22)}(z) - \lambda_{j'}\Psi_{jj'}^{(22)}(z) &=
  \int_{\Omega}d\textbf{x}\,\nu(\vx)\,
  \nabla\cdot\phi^{(2)}_{j}\,\nabla\cdot\phi^{(2)}_{j'} -
  \lambda_{j'}\int_{\Omega}d\textbf{x}\,\nu(\vx)\,\phi^{(2)}_{j}\cdot\phi^{(2)}_{j'}
  \\ &= \frac{\sqrt{\lambda_{j'}}}{\sqrt{\lambda_{j}}L_1 L_2}
  \Big(\lambda_{j} - \Big(\frac{\pi}{L_{1}}\Big)^{2} n\,n' -
  \Big(\frac{\pi}{L_{2}}\Big)^{2} l\,l'\Big)\times
  \\ &\hspace{0.2in}\int_\Omega d \bx \, \nu(\vx) \cos\big(
  \frac{\pi(n-n') x_1}{L_1} \big) \cos\big( \frac{\pi(l-l') x_2}{L_2}
  \big) + \ldots,
\end{align*}
with the dots denoting similar terms, as before. Using integration
by parts twice,
\begin{align*}
 \Theta_{jj'}^{(22)}(z) - \lambda_{j'}\Psi_{jj'}^{(22)}(z)
 &=\frac{\sqrt{\lambda_{j'}}}{\pi^{2}\sqrt{\lambda_{j}}(n - n')(l -
   l')} \Big(\lambda_{j} - \Big(\frac{\pi}{L_{1}}\Big)^{2} n\,n' -
 \Big(\frac{\pi}{L_{2}}\Big)^{2} l\,l'\Big)\times
 \\ & \hspace{1.cm}\int_{\Omega}d\textbf{x}\,\partial^2_{x_1
   x_2}\nu(\vx)\,\sin\big(\tfrac{\pi(n-n')}{L_{1}}x_{1}\big)\,
 \sin\big(\tfrac{\pi(l-l')}{L_{2}} x_{2}\big) + \cdots \\ & =
 \frac{L_{1}L_{2}\sqrt{\lambda_{j'}}}{\pi^{4}\sqrt{\lambda_{j}}(n-n')^{2}
   (l-l')^{2}} \Big(\lambda_{j} - \Big(\frac{\pi}{L_{1}}\Big)^{2}
 n\,n' - \Big(\frac{\pi}{L_{2}}\Big)^{2} l\,l'\Big)
 \times\\ & \hspace{-2cm}\Big[\nu^{c}_{x_1 x_2}(z) +
   \hat\nu^{b}_{x^2_1 x_2}(n-n',z) +\hat\nu^{b}_{x_1 x^2_2}(l-l',z) +
   \hat\nu_{x^2_1 x^2_2}(n-n',l-l',z) \Big] + \cdots\,,
\end{align*}
where we let, 
\begin{align*}
\nu^{c}_{x_1
  x_2}(z):&=(-1)^{n-n'+l-l'}\partial^2_{x_1x_2}\nu(L_1,L_2,z) -
(-1)^{l-l'}\partial^2_{x_1 x_2}\nu(0,L_2,z) - \\ & \hspace{1.6cm}
(-1)^{n-n'}\partial^2_{x_1x_2}\nu(L_1,0,z) + \partial^2_{x_1
  x_2}\nu(0,0,z)\,, \\ \hat\nu^{b}_{x^2_1 x_2}(n-n',z):&=- (-1)^{l-l'}
\int_0^{L_1} d x_1 \, \partial^3_{x_1^2 x_2} \nu(x_1,L_2,z)\cos\big(
\frac{\pi(n-n') x_1}{L_1} \big) + \\ &\hspace{1.6cm} \int_0^{L_1} d
x_1 \, \partial^3_{x_1^2 x_2} \nu(x_1,0,z) \cos\big( \frac{\pi(n-n')
  x_1}{L_1} \big)\,, \\ \hat\nu^{b}_{x_1 x^2_2}(l-l',z):&=-
(-1)^{n-n'} \int_0^{L_2} d x_2 \, \partial^3_{x_1 x_2^2}
\nu(L_1,x_2,z)\cos\big( \frac{\pi(l-l') x_2}{L_2} \big) +
\\ &\hspace{1.6cm}\int_0^{L_2} d x_2 \, \partial^3_{x_1 x_2^2}
\nu(0,x_2,z) \cos\big( \frac{\pi(l-l') x_2}{L_2} \big)\,,
\\ \hat\nu_{x^2_1 x^2_2}(n-n',l-l',z):&=\int_\Omega d \bx \,
\partial^4_{x_1^2 x_2^2} \nu(\vx)\cos\big( \frac{\pi(n-n') x_1}{L_1}
\big)\cos\big( \frac{\pi(l-l') x_2}{L_2} \big).
\end{align*}
As before, the formula applies for $n\neq n'$ and $l \neq l'$.  The
cases with equality may be assumed null without loss of generality.
For the three terms considered below note that
\begin{multline*}
\frac{\sqrt{\lambda_{j'}}}{\sqrt{\lambda_{j}}(n-n')^{2} (l-l')^{2}}
\Big[\lambda_{j} - \Big(\frac{\pi}{L_{1}}\Big)^{2} n\,n' -
  \Big(\frac{\pi}{L_{2}}\Big)^{2} l\,l'\Big] =
\\ \frac{\sqrt{\lambda_{j'}}}{\sqrt{\lambda_{j}}}
\left[\Big(\frac{\pi}{L_{1}}\Big)^{2}\frac{n}{(n - n')(l - l')^{2}} -
  \Big(\frac{\pi}{L_{2}}\Big)^{2}\frac{l}{(n-n')^{2}(l - l')}\right].
\end{multline*}

\textit{(i) Term with $\nu^{c}_{x_1 x_2}(z)$. } This term is
controlled using Lemma \ref{lem:tech}.  For example, the critical term
that decays linearly in the index $n$ satisfies
\begin{align*}
\sum_{j>N}\Big|\sum_{j'>N}\frac{n\,\sqrt{\lambda_{j'}}\, \nu^{c}_{x_1
    x_2}(z)\big(\beta_{j'}v^{(2)}_{j'}(z)
  \big)}{\beta^{2}_{j'}\sqrt{\lambda_{j}}(n-n')(l-l')^{2}}\, \Big|^{2}
&\leq \Big|\frac{L_{1}}{\pi}\,
\nu^{c}_{x_{1}x_{2}}(z)\Big|^{2}\,\Big\|\frac{1}{n\,l^{2}}
\star\frac{\sqrt{\lambda_{j}}\,v^{(2)}_{j}(z)}{\beta_{j}}
\Big\|^{2}_{\ell^{2}}\\ &\leq C\,\Big|\frac{L_{1}}{\pi}\,
\nu^{c}_{x_{1}x_{2}}(z)\Big|^{2}\,\sum_{j>N} \big|
\beta_{j}\,v^{(2)}_{j}(z) \big|^{2}.
\end{align*}
The term with linear decay in the index $l$ is similar.

\textit{(ii) Terms with $\hat\nu^{b}_{x^2_1 x_2}(n-n',z)$ and
  $\hat\nu^{b}_{x_1 x^2_2}(l-l',z)$.} These are controlled using
either Lemma \ref{lem:tech}, or the standard Young's inequality for
discrete convolutions. For instance
\begin{align*}
\sum_{j>N}\Big|\hspace{-0.05in}
\sum_{j'>N}\frac{n\,\sqrt{\lambda_{j'}}\, \hat\nu^{b}_{x^{2}_1
    x_2}(n-n', z)\big(\beta_{j'}v^{(2)}_{j'}(z) \big)}{
  \beta^{2}_{j'}\sqrt{\lambda_{j}}(n-n')(l-l')^{2}}\,
\Big|^{2} \hspace{-0.05in}&\leq
\Big|\frac{L_{1}}{\pi}\Big|^{2}\,\Big\|\frac{\hat\nu^{b}_{x^{2}_1
    x_2}(n,
  z)}{n\,l^{2}}\star\frac{\sqrt{\lambda_{j}}\,v^{(2)}_{j}(z)}{\beta_{j}}
\Big\|^{2}_{\ell^{2}}\\ & \hspace{-3cm} \leq C\,
\Big|\frac{L_{1}}{\pi}\Big|^{2}\,\left(\sum_{n\neq0}\sum_{l\neq0}
\Big|\frac{\hat\nu^{b}_{x^{2}_1 x_2}(n,
  z)}{n\,l^{2}}\Big|\right)^{2}\sum_{j>N} \big|
\beta_{j}\,v^{(2)}_{j}(z) \big|^{2}\\ & \hspace{-3cm} \leq
C\,\Big(\sum_{n\neq0}\frac{1}{n^{2}}\Big)^{2}\big\|
\hat\nu^{b}_{x^{2}_1 x_2}(n, z) \big\|^{2}_{\ell^{2}} \sum_{j>N} \big|
\beta_{j}\,v^{(2)}_{j}(z) \big|^{2}.
\end{align*}

\textit{(iii) Term with $\hat\nu_{x^2_1 x^2_2}(n-n',l-l',z)$.}  This
term is estimated using Young's inequality for discrete convolutions.
For instance,
\begin{align*}
\sum_{j>N}\Big|\sum_{j'>N}&\frac{n\,\sqrt{\lambda_{j'}}\,
  \hat\nu_{x^{2}_1 x^{2}_2}(n-n', l - l',
  z)\big(\beta_{j'}v^{(2)}_{j'}(z)
  \big)}{\beta^{2}_{j'}\sqrt{\lambda_{j}}(n-n')(l-l')^{2}}\, 
\Big|^{2}  \\ 
&\leq
\Big|\frac{L_{1}}{\pi}\Big|^{2}\,\Big\|\frac{\hat\nu_{x^{2}_1
    x^{2}_2}(n,l,
  z)}{n\,l^{2}}\star\frac{\sqrt{\lambda_{j}}\,v^{(2)}_{j}(z)}{\beta_{j}}
\Big\|^{2}_{\ell^{2}}\\ &\leq C\,
\Big|\frac{L_{1}}{\pi}\Big|^{2}\,\left(\sum_{n\neq0}\sum_{l\neq0}
\Big|\frac{\hat\nu_{x^{2}_1 x^{2}_2}(n,l,
  z)}{n\,l^{2}}\Big|\right)^{2}\sum_{j>N} \big|
\beta_{j}\,v^{(2)}_{j}(z) \big|^{2}\\  &\leq
C\,\Big(\sum_{n\neq0}\frac{1}{n^{2}}\Big)^{2}\big\| \hat\nu_{x^{2}_1
  x^{2}_2}(n,l, z) \big\|^{2}_{\ell^{2}} \sum_{j>N} \big|
\beta_{j}\,v^{(2)}_{j}(z) \big|^{2}\, ,
\end{align*}
with constant $C:=C(w, L_1,L_2)$.  

Additionally, note that by the trace theorem
\begin{equation*}
\big|\nu^{c}_{x_{1}x_{2}}(z)\big|\leq
C\,\big\|\boldsymbol\tau\nu(\textbf{x},z)\big\|_{H^{3}(\partial\Omega)}\leq
C\,\big\|\nu(\textbf{x},z)\big\|_{H^{4}(\Omega)}\,,\quad \forall
z\geq0,
\end{equation*}
where $\boldsymbol\tau$ is the trace operator.  The processes
$\hat\nu^{b}_{x^2_1 x_2}(n,z)$ and $\hat\nu_{x^2_1 x^2_2}(n,l,z)$ are
Fourier coefficients of the process
$\partial^{3}\nu_{x^{2}_{1}x_{2}}(\vx)$,
$\textbf{x}\in\partial\Omega$, and
$\partial^{4}_{x^{2}_{1}x^{2}_{2}}\nu(\vx)$, $\textbf{x}\in\Omega$,
respectively, so we have
\begin{align*}
\big\|\hat\nu^{b}_{x^2_1 x_2}(n,z)\big\|_{\ell^{2}} &\leq
\|\boldsymbol\tau\nu(\textbf{x}, z)\|_{H^{3}(\partial\Omega)}\,,\quad
\text{and} \\ \big\|\hat\nu_{x^2_1 x^2_2}(n,l,z)\big\|_{\ell^{2}}
&\leq \|\nu(\textbf{x}, z)\|_{H^{4}(\Omega)}\,,\quad \forall z\geq0.
\end{align*}
Thus, we conclude that
\begin{equation}\label{AEe5}
  \frac{1}{2}\sum_{j>N} \sum^{\mathfrak{M}_{j}}_{s=1} \text{b}^{(s)}_{j} 
  \leq C\, \Big(\sum_{n\neq0} \frac{1}{n^{2}}\Big)^{2}\,\sup_{z\geq0}
  \big\|\nu(\textbf{x},z)\big\|^{2}_{H^{4}(\Omega)}\,\big\|v\big\|^{2}\,,
\end{equation}
for yet another constant $C:=C(w,L_{1}, L_{2})$.  Gathering \eqref{AEe2},
\eqref{AEe4} and \eqref{AEe5}, we obtain that the operator
$\tilde{\mathbb{Q}}_{\gamma}$ is bounded  
\begin{equation*}
  \| \tilde{\mathbb{Q}}_{\gamma}\| \leq C\,\sup_{z\geq0}
  \|\nu(\textbf{x}, z)\|_{H^{4}(\Omega)}.
\end{equation*}

Using the bounded auxilliary operator $\tilde{\mathbb{Q}}_{\gamma}$,
we now define the operator $\mathbb{Q}_{\gamma}$ that enters the
expression of $\mathbb{Q}$,
\begin{align*}
  \big(\mathbb{Q}_{\gamma}\,v\big)^{(s)}_{j}(z) :&= \frac{i}{2}
  \int^{z}_{0}d\zeta\,e^{-\beta_{j}(z-\zeta)}\sum_{j'>N}
  \sum^{\mathfrak{M}_{j'}}_{s'=1}(-1)^{s'}\gamma^{(ss')}_{jj'}(\zeta)\,
  v^{(s')}_{j'}(\zeta)\\
  &=\frac{1}{w^{(s)}_{j}}\big(\tilde{\mathbb{Q}}_{\gamma}\,
  \tilde{v}\big)^{(s)}_{j}(z)\,,\quad\;\; \tilde{v}^{(s)}_{j} =
  w^{(s)}_{j}\,v^{(s)}_{j}.
\end{align*}
It is bounded in the $w$--norm
\begin{equation*}
  \|\mathbb{Q}_{\gamma}v\|_{w}=\|\tilde{\mathbb{Q}}_{\gamma}
  \tilde{v}\|\leq\|\tilde{\mathbb{Q}_{\gamma}}\|\,\|\tilde{v}\|=
  \|\tilde{\mathbb{Q}_{\gamma}}\|\,\|v\|_{w}\, .
\end{equation*}
The remaining terms defining $\mathbb{Q}$ are estimated similarly.
\end{proof}

\vspace{0.1in}
\begin{lemma}
\label{lem:tech}
Let $j=(n,l)\in \mathbb{Z}\times\mathbb{Z}$ and define the convolution operator
\begin{equation*}
T(v)(j) = \Big(\frac{\kappa(l)}{n}\ast v\Big)(j),
\end{equation*}
for a sequence $v=\{v_{j}\}$. It satisfies the bound 
$ 
\big\|T(v)\big\|_{\ell^2}\leq \pi\big\| \kappa\big\|_{\ell^1}\big\|v\|_{\ell^2}.
$
\end{lemma}
\\\\
\emph{Proof:} Let $j'=(n',l')$ and define
$
\xi(n',l) := \sum_{l'}\kappa(l-l')v_{j'},
$
so that
\begin{equation*}
T(v)(j)=\sum_{n\neq n'} \frac{\xi(n',l)}{n-n'}.
\end{equation*}
Compute the $\ell^{2}$--norm,
\begin{equation*}
\big\|T(v)\big\|^{2}_{\ell^2}=\sum_{l}\left(\sum_{n}\Big|\sum_{n\neq
  n'} \frac{\xi(n',l)}{n-n'} \Big|^{2}\right)\leq
\pi^{2}\sum_{l}\big\|\xi(\cdot,l)\big\|^{2}_{\ell^{2}},
\end{equation*}
where we have used that, for each $l$, the discrete Hilbert transform
\cite{Gra} is bounded with norm $\pi$.  Thus,
\begin{equation*}
\big\|T(v)\big\|^{2}_{\ell^2}\leq\pi^{2}\sum_{l}\sum_{n}\big|\xi(n,l)
\big|^{2}=\pi^{2}\sum_{n}\left(\sum_{l}\Big|\sum_{l'}\kappa(l-l')v_{(n,l')}
\Big|^{2}\right).
\end{equation*}
Using Young's inequality, for each fixed $n$, leads to
\begin{equation*}
\big\|T(v)\big\|^{2}_{\ell^2}\leq\pi^{2}\Big(\sum_{l}\big|\kappa(l)
\big|\Big)^{2}\sum_{n}\big\|v(n,\cdot)\big\|^{2}_{\ell^{2}}=\pi^{2}\big\|
\kappa\big\|^{2}_{\ell^1}\big\|v\|^{2}_{\ell^2}. ~ ~ \Box
\end{equation*}

\section{Calculation of the matrix ${\bf Q}_j$}
\label{ap:explicit}
The expression of $\EE\{\widetilde m_{jj}(0)\}$ follows by direct 
calculation from 
\[
m_{j}(z) = m_{AA,jj}(z) + m_{AA,jj}^e(z)\,,
\]
definitions (\ref{eq:MAA}) and (\ref{eq:MAAe}), and integration by
parts. 

To write the contribution of the evanescent modes, let $\Psi_{jj'}(z)$
and $\Theta_{jj'}(z)$ be the matrices in $\mathbb{R}^{\mM_j \times
  \mM_{j'}}$ with entries given by the leading order stationary
processes in (\ref{eq:PP}-\ref{eq:DD}). They satisfy the
symmetry relations
\[
\Psi_{jj'}(z) = \Psi_{j'j}^T(z)\, , \qquad \Theta_{jj'}(z) = \Theta_{j'j}(z) =
\Theta_{j'j}^T(z)\, ,
\]
and we recall from (\ref{eq:structD}) that $\Theta_{jj'}(z)$ has only one
non-zero entry, for $(s,s') = 2$.  We obtain after straightforward 
calculations that
\begin{equation}
  \EE\left\{ m_{AA,jj}^{e}(0)\right\} = 
  \frac{i}{2} \left( \begin{array}{cc}
      \sqrt{\beta_j/k} & 0 \\
      0 & \sqrt{k/\beta_j} \end{array} \right) \mathcal{M}_j^{e}
  \left( \begin{array}{cc}
      \sqrt{k/\beta_j} & 0 \\
      0 & \sqrt{\beta_j/k} \end{array} \right)\, , \label{eq:AE1}
\end{equation}
where ${\mathcal M}_j^e$ is the $\mM_j \times \mM_j$ matrix with  entries
\begin{align}
  \mathcal{M}_j^{(ss')e} =&\sum_{l > N} \sum_{q=1}^{\mM_l}
  \frac{1}{\beta_j} \left[ \lambda_j \de_{s1} \EE \left\{
      \Psi_{jl}^{(sq)}(0) \Psi_{lj}^{(qs')}(0) \right\} - \EE\left\{
      \Psi_{jl}^{(sq)}(0) \Theta_{lj}^{(qs')}(0)
    \right\}\right] + \nonumber \\
  & \sum_{l > N} \sum_{q=1}^{\mM_l} \int_0^\infty ds \, e^{-\beta_l z}
  \sin (\beta_j z) \Big[ \EE\left\{ \Theta_{jl}^{(sq)}(0) \Psi_{lj}^{(qs')}(z)
  \right\} +
  \nonumber \\
  &\hspace{0.6in} \left(1 + \frac{\lambda_j}{\beta_j^2}
    \de_{s1}\right) \EE\left\{ \Psi_{jl}^{(sq)}(0) \Theta_{lj}^{(qs')}(z)
  \right\}\Big] + \nonumber \\
  & \sum_{l > N} \sum_{q=1}^{\mM_l} \int_0^\infty ds \, e^{-\beta_l z}
  \cos (\beta_j z) \Big[\frac{1}{\beta_j \beta_l} \EE\left\{
    \Theta_{jl}^{(sq)}(0) \Theta_{lj}^{(qs')}(z)
  \right\} + \nonumber \\
  &\hspace{0.6in}
  \beta_j \beta_l\left(1 + \frac{\lambda_j}{\beta_j^2}
    \de_{s1}\right) \left(\frac{\lambda_l}{\beta_l^2}
    \de_{q1}-1\right) 
 \EE \left\{ \Psi_{jl}^{(sq)}(0)
     \Psi_{lj}^{(qs')}(z)\right\}\Big]\, .\label{eq:AE1p}
\end{align}
Similarly, using the order $\epsilon$ terms in
(\ref{eq:PP}-\ref{eq:DD}), we write 
\begin{equation}
  \EE\left\{ m_{AA,jj}(0)\right\} = 
  \frac{i}{2} \left( \begin{array}{cc}
      \sqrt{\beta_j/k} & 0 \\
      0 & \sqrt{k/\beta_j} \end{array} \right) \mathcal{M}_j
  \left( \begin{array}{cc}
      \sqrt{k/\beta_j} & 0 \\
      0 & \sqrt{\beta_j/k} \end{array} \right)\, , \label{eq:AE2}
\end{equation}
with matrix ${\mathcal M}_j \in \mathbb{R}^{\mM_j \times \mM_j}$ defined by 
\begin{equation} {\mathcal M}_j^{(ss')} =
  \frac{\lambda_j}{4\beta_j}\de_{s' 2}\EE\{\psi_{jj}^{(ss')}(0)\} -
  \frac{\lambda_j}{2 \beta_j} \de_{s1}\EE\{\psi_{jj}^{(ss')}(0)\} -
  \frac{1}{4 \beta_j} \EE\{\theta_{jj}^{(ss')}(0)\}\, .
\label{eq:AE2p}
\end{equation}

In equations  (\ref{eq:AE1}) and (\ref{eq:AE2}) we assumed that 
$\mM_j = 2$. Otherwise we have 
\[
E\{m_{AA,jj}^e(0)\} = \frac{i}{2} {\mathcal M}_j^e\, , 
\qquad E\{m_{AA,jj}(0)\} = \frac{i}{2} {\mathcal M}_j\, , \qquad 
\mM_j = 1\, ,
\]
with scalar valued ${\mathcal M}_j^e$, ${\mathcal M}_j$ equal to the
$(s,s') = (1,1)$ entries in (\ref{eq:AE1p}) and (\ref{eq:AE2p}).  The
imaginary matrix $i {\boldsymbol \kappa}_j$ in equation (\ref{eq:E6})
is the sum of (\ref{eq:AE1}) and (\ref{eq:AE2}).

The matrix ${\bf Q}_j$ is given by 
\begin{equation}
{\bf Q}_j = \left( \begin{array}{cc}
      \sqrt{\beta_j/k} & 0 \\
      0 & \sqrt{k/\beta_j} \end{array} \right) \mathcal{U}_j
  \left( \begin{array}{cc}
      \sqrt{k/\beta_j} & 0 \\
      0 & \sqrt{\beta_j/k} \end{array} \right) + 
i {\boldsymbol \kappa}_j\, 
\label{eq:Qj}
\end{equation}
with matrix ${\mathcal U}_j \in \mathbb{C}^{\mM_j \times \mM_j}$. The real 
part of its entries is 
\begin{align}
 \operatorname{Re}\left[{\mathcal U}_j^{(ss')}\right] =& -\frac{1}{4} \sum_{l=1}^N \sum_{q=1}^{\mM_l}
\int_0^\infty \hspace{-0.1in}dz \, \cos[(\beta_l-\beta_j)z] \left(\frac{k^2}{\beta_j}
\de_{s1} + \beta_j \de_{s2}\right) \left(\frac{k^2}{\beta_l}
\de_{q1} + \beta_l \de_{q2}\right) \times \nonumber \\
& \hspace{0.5in}\EE\left\{ 
\left(\Psi_{jl}^{(sq)}(0) + \frac{\Theta_{jl}^{(sq)}(0)}{\beta_j^2}\right)
\left(\Psi_{lj}^{(qs')}(z) + \frac{\Theta_{lj}^{(qs')}(z)}{\beta_l^2}\right)
\right\},
\end{align}
and the imaginary part is 
\begin{align}
  \operatorname{Im}\left[{\mathcal U}_j^{(ss')}\right] =& \frac{1}{4}
  \sum_{l=1}^N \sum_{q=1}^{\mM_l} \Big[
  \left(\frac{\lambda_j}{\beta_j}\de_{s1}-
    \frac{\lambda_l}{\beta_l}\de_{q1} \right) \EE\left\{
    \Psi_{jl}^{(sq)}(0)\Psi_{lj}^{(qs')}(0)\right\} -\nonumber \\
  &\hspace{0.4in} \frac{1}{\beta_j} \EE\left\{
    \Psi_{jl}^{(sq)}(0)\Theta_{lj}^{(qs')}(0)\right\} + \frac{1}{\beta_l}
  \EE\left\{ \Theta_{jl}^{(sq)}(0)\Psi_{lj}^{(qs')}(0)\right\} \Big] -
  \nonumber \\
  & \frac{1}{4} \sum_{l=1}^N \sum_{q=1}^{\mM_l}\int_0^\infty
  \hspace{-0.1in}dz \, \sin[(\beta_l-\beta_j)z]
  \left(\frac{k^2}{\beta_j} \de_{s1} + \beta_j
    \de_{s2}\right)\left(\frac{k^2}{\beta_l} \de_{q1} +
    \beta_l \de_{q2}\right) \times \nonumber \\
  &\hspace{0.4in}\EE\left\{ \left(\Psi_{jl}^{(sq)}(0) +
      \frac{\Theta_{jl}^{(sq)}(0)}{\beta_j^2}\right)
    \left(\Psi_{lj}^{(qs')}(z) +
      \frac{\Theta_{lj}^{(qs')}(z)}{\beta_l^2}\right) \right\}.
\end{align}
\section{Power spectral density of a stationary matrix process}
\label{ap:psd}
Let  $M(z)$ be an $m \times n$ matrix with entries given by stationary processes and covariance
\begin{equation*}
\text{R}_{M}(z):=\mathbb{E}\big\{ M^{\star}(z)\, M(0)\big\}.
\end{equation*}
Its power spectral density 
\begin{equation*}
\text{S}_{M}(z)=\int^{\infty}_{-\infty} \text{d}z'\, \text{R}_{M}(z') e^{izz'},
\end{equation*}
is easily verified to be a Hermitian matrix, and we show next that it
is also positive semidefinite for any $z\in\mathbb{R}$.  Indeed,
\begin{align*}
\big( \text{S}_{M}(z)\textbf{x},\textbf{x} \big)
&=\int^{\infty}_{-\infty}\text{d}z'\,\big(\text{R}_{M}(z')\textbf{x},\textbf{x}\big)
e^{i\,z z'}=\int^{\infty}_{-\infty} \text{d}z'\,\mathbb{E}\big\{ \big(
M(z')\textbf{x},M(0)\textbf{x}\big) \big\} e^{ i\,zz'},
\end{align*}
for all $ {\bf x} \in \mathbb{R}^{n}$, and the vector
$\boldsymbol\mu(z)=M(z)\textbf{x}$ is stationary for any fixed
$\textbf{x}$. Therefore
\begin{align*}
\big( \text{S}_{M}(z)\textbf{x},\textbf{x}
\big)=\hspace{-0.02in}\int^{\infty}_{-\infty}
\hspace{-0.055in} dz'\,\mathbb{E}\big\{ \big(
\boldsymbol\mu(z'),\boldsymbol\mu(0)\big) \big\} e^{
  i\,zz'}=\sum^{m}_{i=1}\int^{\infty}_{-\infty}
\hspace{-0.055in} dz'\,\mathbb{E}\big\{\,
\overline{\boldsymbol\mu_{i}(z')}\;\boldsymbol\mu_{i}(0)\,\big\} e^{
  i\,zz'} \geq 0,
\end{align*} 
with the inequality  implied by Bochner's theorem.
\section{The evolution of the mean powers}
\label{ap:cone}
We show here that since ${\bf P}_o \in \mathfrak{C} \subset
\mathfrak{X}$, the solution ${\bf P}(Z)$ of (\ref{eq:TE1}) remains  in the cone $\mathfrak{C} =\mathfrak{C}_1 \times
\ldots \times \mathfrak{C}_N$ for all $Z$. Writing (\ref{eq:TE1})
component-wise and using (\ref{eq:TE3}), we obtain
\begin{equation}
\partial_Z {\bf P}_j(Z) = \Upsilon^+({\bf P})_j(Z) + {\bf Q}_j {\bf P}_j(Z) + 
{\bf P}_j(Z) {\bf Q}_j^\star, \qquad Z > 0.
\end{equation}
Equivalently, 
\begin{equation}
\partial_Z \left[e^{- {\bf Q}_j Z} {\bf P}_j(Z) e^{- \bf{Q}_j^\star
    Z}\right] = e^{- {\bf Q}_j Z} \Upsilon^+({\bf P})_j(Z) e^{- {\bf
    Q}^\star_j Z}, 
\end{equation}
and integrating in $Z$ we obtain
\begin{equation}
{\bf P}_j(Z) = e^{{\bf Q}_j Z} {\bf P}_{j,o} e^{{\bf Q}_j^\star Z} +
\int_0^Z dz \, e^{{\bf Q}_j (Z-z)} \Upsilon^+({\bf P})_j(z) e^{{\bf
    Q}^\star_j (Z-z)}.
\end{equation}
That ${\bf P}_j(Z) \in \mathfrak{C}_j$ follows from ${\bf P}_{j,o} \in
\mathfrak{C}_j$ and (\ref{eq:TE3p}).

\bibliographystyle{plain} \bibliography{BIBLIO}

\end{document}